%% file: letter.tex
\documentclass[aps,prl,twocolumn,superscriptaddress,reprint,nofootinbib]{revtex4-2}
\usepackage{amsmath,amssymb,amsthm,bm}
\usepackage{graphicx,booktabs}
\usepackage[hidelinks]{hyperref}
\newcommand{\bp}{\mathbf{p}}
\newcommand{\dd}{\mathrm{d}}
\newcommand{\Zg}{Z}
\newtheorem*{prop}{Proposition}
\begin{document}
\title{Time-Reversal Selection Rule for Twist Disorder}
\newcommand{\affMSE}{School of Materials Science and Engineering, Beihang University, Beijing 100191, China}
\newcommand{\affAI}{National Key Laboratory of Artificial Intelligence for Material Science, Beihang University, Beijing 100191, China}
\newcommand{\affMcGill}{Center for the Physics of Materials and Department of Physics, McGill University, Montreal, Quebec H3A 2T8, Canada}
\newcommand{\affNano}{Nanoacademic Technologies Inc., Suite 802, 666 Sherbrooke West, Montreal, Quebec H3A 1E7, Canada}
\author{Peng~Kang}\email[Contact author: ]{pengkang@buaa.edu.cn}\affiliation{\affMSE}\affiliation{\affAI}\affiliation{\affMcGill}\affiliation{\affNano}
\author{Da~Wan}\affiliation{\affMSE}\affiliation{\affAI}
\author{Shulin~Bai}\affiliation{\affMSE}\affiliation{\affAI}
\author{Pengfei~Zhang}\affiliation{\affMSE}\affiliation{\affAI}
\author{Vincent~Michaud-Rioux}\affiliation{\affNano}
\author{Zhen~Li}\affiliation{\affMSE}\affiliation{\affAI}
\author{Yu~Liu}\affiliation{\affMSE}\affiliation{\affAI}
\author{Lei~Zheng}\email[Contact author: ]{zhenglei@buaa.edu.cn}\affiliation{\affMSE}\affiliation{\affAI}
\author{Li-Dong~Zhao}\affiliation{\affMSE}\affiliation{\affAI}
\author{Huibin~Xu}\affiliation{\affMSE}\affiliation{\affAI}
\begin{abstract}
Random layer rotations are ubiquitous in layered matter, from turbostratic films to rotationally disordered crystals, and strongly suppress transport across the layers. Generated by a symmetry operation, this disorder is absent on the rotation axis, so the thickness laws of conduction are set by how fast backscattering vanishes there. We show that a time-reversal selection rule decides this: each twist harmonic of the interlayer bond is either a frame rotation, whose backscattering cancels, or a coupling change, which scatters. A zero-set theorem turns the rule into universal thickness laws for electrons and phonons, fixed by band-edge symmetry; random twist increments rescale the backscattering but keep its order. Where the axis is unprotected, band nodes provide a second route: a transparent energy with heavy-tailed disorder and an $N^{-3}$ law. A first-principles forward-channel model of black phosphorus shows both routes; its random-twist films act as a symmetry filter that, relative to aligned stacks, retains a hundred times more of the electron than of the hole conductance at 64 layers. The symmetry of band edges and phonon branches thus decides what crosses a disordered stack, making stacking disorder a design element.
\end{abstract}
\maketitle

\emph{Introduction}---Random stacking angles are a generic disorder of layered matter, from turbostratic films to rotationally disordered crystals. Films of randomly twisted layers conduct heat across the layers almost three orders of magnitude worse than along them~\cite{Kim2021}; stacking and interface disorder suppress phonon transport~\cite{Chiritescu2007,Eriksson2023,Juntunen2019,Luckyanova2018,Wang2025}, as do dislocation-driven stacking changes~\cite{Liu2025}, and twist controls vertical tunnelling~\cite{Srivastava2021}, carrier mobility~\cite{Kang2017,Kang2020} and phonon velocities~\cite{KangSnSe2026}. How, then, does transport across $N$ randomly rotated layers scale with $N$? One-dimensional disorder localizes every state~\cite{Anderson1958}, but a conductance averages over channels and energies, and it can decay algebraically if backscattering vanishes, exactly or to leading order, on some set of them: at low phonon frequencies~\cite{Dhar2001}, at special angles~\cite{Sipe1988,Joy2020}, at resonances and in spectral windows of correlated disorder~\cite{Dunlap1990,Flores1989,IzrailevKrokhin1999,IzrailevRev2012,Vanoni2026}, or in symmetry-protected channels~\cite{EggarterRiedinger1978,Brouwer1998,Wakabayashi2007}.

Rotational disorder is special [Fig.~\ref{fig:matrix}(a,b)]. Interlayer coupling that conserves the in-plane momentum $\bp$ turns each channel into a chain~\cite{SiriviboonVolkov2026}, and a rotation enters only through $R_\theta^{-1}\bp$: on the rotation axis, $\bp=0$, it is no disorder at all. Every channel with $p=|\bp|>0$ still localizes, but its localization length can diverge as $p\to0$, so the thickness law is set by how fast rotational backscattering vanishes on the axis. Which twists scatter, and at what order? The twist dependence of a single interface cannot tell, because the chain is a matrix chain and part of that dependence merely rotates the band-edge frame.

Here we show that time reversal sorts the twist harmonics of the interlayer coupling into frame rotations, whose backscattering cancels, and coupling changes, which scatter [Fig.~\ref{fig:matrix}(c)]. This parity rule, a selection rule for twist scattering, fixes the scattering order $\kappa$ of a band edge from the point group of the layer and its monolayer wavefunctions, and a zero-set theorem turns $\kappa$ into a thickness law. Where the axis is not protected, band nodes open a second route: given a factorized coupling and an open band, rotational disorder cancels at a transparent energy, the nodes make the residual disorder heavy tailed, and the thermally averaged conductance falls as $N^{-3}$. Symmetry fixes the exponents; the material decides only which band edges it has and how strongly its layers couple. The rule holds for electrons and phonons in C$_n$ layers ($n\ge2$) and survives random twist increments, whereas a strictly regular twist makes the chain periodic or quasiperiodic~\cite{SiriviboonVolkov2026}. A first-principles forward-channel model of black phosphorus realizes both routes; its random-twist films act as a symmetry filter, with a hundredfold electron--hole contrast in retained conductance at 64 layers.

\begin{figure*}[t]\centering
\includegraphics[width=\textwidth]{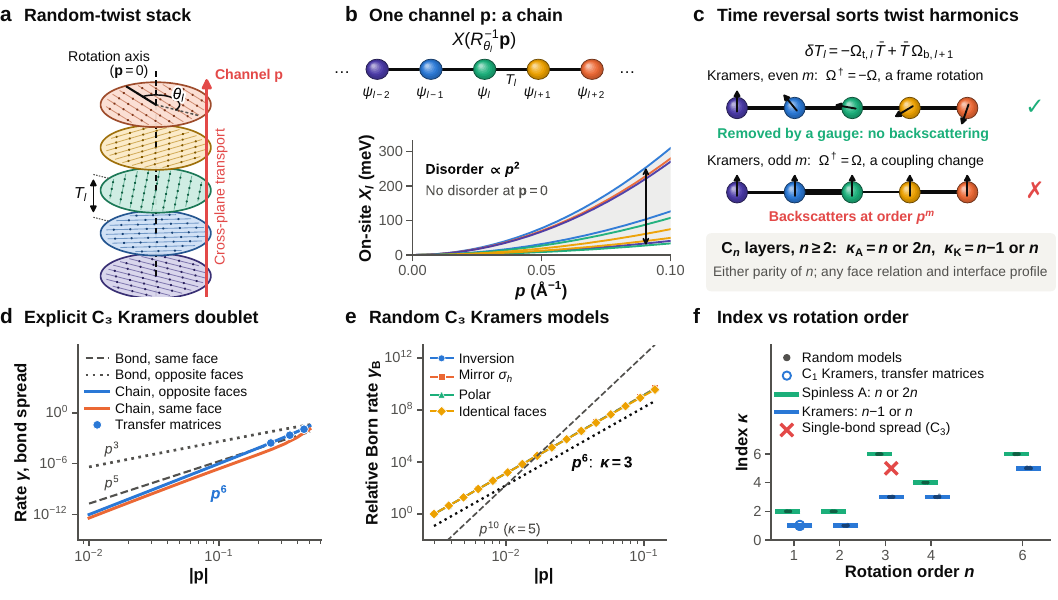}
\caption{Time reversal decides which twists scatter. (a)~At fixed in-plane momentum $\mathbf p$, cross-plane transport through a random-twist stack is a chain with bonds $T_l$. (b)~Its on-site energies $X(R_{\theta_l}^{-1}\mathbf p)$ for the PBE phosphorene conduction band: rotation is disorder that vanishes at $\mathbf p=0$. (c)~Time reversal sorts the twist harmonics of the bond into frame rotations (anti-Hermitian $\Omega$, a gauge) and coupling changes (Hermitian $\Omega$), which backscatter. (d)~C$_3$ Kramers doublet: the single-bond spread grows as $p^5$ or $p^3$, the Born rate of the chain as $p^6$. (e)~Born rate of random C$_3$ Kramers models for four face relations (coincident). (f)~Index against rotation order; bars: the Proposition.}\label{fig:matrix}
\end{figure*}

\emph{Channel chains}---For a band-edge multiplet with on-site block $X(\mathbf q)$ in the layer frame, channel $\bp$ obeys
$E\psi_l=X^{\rm lab}_l\psi_l+T_l\psi_{l+1}+T_{l-1}^\dagger\psi_{l-1}$,
with $X^{\rm lab}_l=D(\theta_l)X(R_{\theta_l}^{-1}\bp)D(\theta_l)^\dagger$ and $D(\theta)$ the rotation of the multiplet. Between states at the same $\bp$ an incommensurate interface couples only their planar Fourier components at $\bp$~\cite{Bistritzer2010,Koshino2015}. Umklapp components couple distinct channels and, when their partners are off shell, correct this forward coupling only at second order~\cite{SM}. The forward bond is bilinear in the face components,
\begin{equation}
T_l=G^{\rm t}_l(\bp)^\dagger\,\tau\,G^{\rm b}_{l+1}(\bp),\qquad G_l(\bp)=U_s(\theta_l)G(R_{\theta_l}^{-1}\bp),\label{eq:T}
\end{equation}
where $G^{\rm t,b}$ are the components near the top and bottom faces (spin $\times$ multiplet index, on one or several interface profiles) and $\tau$ is a real lab-frame kernel. From plane-wave wavefunctions, $G$ and its slope are read in the van der Waals gap and $\tau$ is Bardeen's kernel, so $T$ is in eV~\cite{Bardeen1961,SM}. At $\bp=0$ only components with $J_z=s_z$ couple: the averaged coupling $|\bar T|\propto p^r$ has $r=0$ if the C$_n$ class of the multiplet contains $J_z\equiv s_z$, and $r=2\mu_0$ otherwise, with $\mu_0$ the smallest $|J_z-s_z|$ in the class. We define the index $\kappa$ by the Born backscattering rate of the full matrix chain between its open channels, $\gamma_B\propto p^{2(\kappa-r)}$.

\emph{Thickness law}---A typical transmission decays as $e^{-2N\gamma}$, with $\gamma$ the Lyapunov exponent of the effective chain, but the normalized conductance $g(N)$, an average over the spectral measure, is carried by the channels where $\gamma$ is small. We prove~\cite{SM} that if $\gamma$ vanishes on a set $\Zg$ of codimension $c$ as $\rho^{2m}$ in the distance $\rho$ from $\Zg$, and the measure grows as $\rho^{s+c-1}$, then $N^\alpha g(N)$ tends to an explicit amplitude with $\alpha=(c+s)/2m$, provided that near $\Zg$ the mean transmission is a controlled function of $N\gamma$, a kernel set by the class of the disorder (below). On the axis $c=2$ and $2m=2(\kappa-r)$~\cite{DerridaGardner1984}:
\begin{equation}
g(N)\simeq A\,N^{-(2+s)/2(\kappa-r)} .\label{eq:law}
\end{equation}
The theorem concerns the effective disorder: an orbital side-coupled to a level of random strength has disorder of order one, yet its eliminated on-site term vanishes at one energy ($c=1$), where the theorem gives $g\simeq7.474\,N^{-1/2}$ and exact transport $g(8192)=0.0825$ against 0.0826 [Fig.~\ref{fig:laws}(a)].

\emph{Parity rule}---In the lab gauge every entry of $G$ carries twist harmonics $e^{i(J_z-s_z-m)\theta}$ from its in-plane harmonics $m\equiv J_z-s_z\pmod n$, of order $p^{|m|}$. Time reversal fixes the character of each harmonic in the bond:

\begin{prop}
Let the faces couple through any number of real interface profiles with any real kernel. For a $J_z=\pm\frac12$ Kramers pair of a C$_n$ layer ($n\ge2$) the leading twist dependence of the bond is its spin-flip harmonic, of in-plane order $m=n-1$, with the form $\delta T_l=-\Omega_{{\rm t},l}\bar T+\bar T\,\Omega_{{\rm b},l+1}$, where $\Omega_{{\rm t},l}$ is set by the twist of layer $l$ and $\Omega_{{\rm b},l+1}$ by that of layer $l+1$. Time reversal makes $\Omega$ anti-Hermitian for even $m$, a rotation of the multiplet frame whose Born backscattering amplitude vanishes identically, and Hermitian for odd $m$, a change of coupling strengths, which scatters. For spinless states even harmonics change the modulus of the bond and odd ones its phase. Hence, for $p_z$/$s$-type faces, any relation between the two faces and generic coefficients, spinless A-type states and Kramers pairs have the Born indices
\begin{equation}
\kappa_{\rm A}=\begin{cases}n & n\ \text{even}\\2n & n\ \text{odd}\end{cases},\qquad
\kappa_{\rm K}=\begin{cases}n-1 & n\ \text{even}\\ n & n\ \text{odd}\end{cases}.\label{eq:kappa}
\end{equation}
\end{prop}
\noindent Appendix~A gives the proof and its conditions; exact Lyapunov exponents support the Proposition beyond Born order~\cite{SM}.

The rule does not need independent twists. For random, independent increments $\theta_{l+1}-\theta_l$ the frame-rotation cancellation holds bond by bond, for every sequence, and each remaining twist harmonic adds a finite, non-negative spectral weight of the sequence to the Born rate: correlations rescale the rate but generically keep $\kappa$. In a spinless C$_2$ and a Kramers C$_3$ model, helices with fluctuating pitch lower the rate up to 250-fold with the index unchanged within 0.1, the C$_3$ frame-rotation cancellation included~\cite{SM}.

Three natural proxies for twist scattering fail, and the rule shows why. First, the twist dependence of one bond does not measure it: the chain joins the top face of one layer to the bottom face of the next and removes frame rotations (Appendix~B), so a C$_3$ Kramers doublet whose single-bond spread grows as $p^5$ or $p^3$ backscatters as $p^6$ [Fig.~\ref{fig:matrix}(d)]. Second, parity, not the size of an admixture, decides whether a harmonic scatters. Third, a projected spin splitting~\cite{ZhangZunger2014,LiAppelbaum2018} is not the material object: Kramers indices do not depend on inversion, and a PT-symmetric C$_2$ pair has $\kappa=1$ with no on-site splitting. What matters is the face-projected density matrix; C$_1$ layers are protected only if all profiles share one spinor frame~\cite{SM}. Random models with one to three profiles [Fig.~\ref{fig:matrix}(e,f)] and Slater--Koster MX$_2$ layers with spin--orbit coupling ($\kappa=2.94$ and 3.00 for D$_{3h}$ and D$_{3d}$) confirm Eq.~\eqref{eq:kappa}~\cite{SM}; Table~\ref{tab:mat} lists examples.

\begin{figure}[t]\centering
\includegraphics[width=\columnwidth]{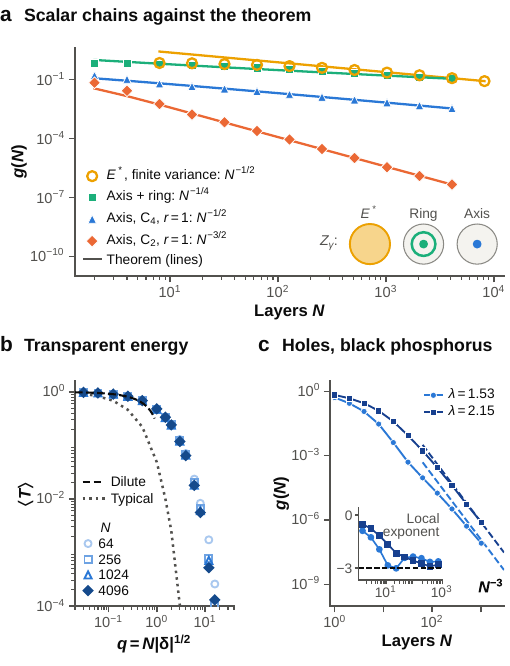}
\caption{Thickness laws from zero sets. (a)~Mean conductance of scalar chains (markers, exact transport) against the theorem with closed-form amplitudes (lines); insets: zero sets in the channel plane. (b)~Transparent energy at nodes: mean transmission of the factorized chain [Eq.~\eqref{eq:dimer}] against $q=N|\delta|^{1/2}$ ($2^{14}$ sequences per $N$; $\cos k=-0.90$). (c)~Thermal hole means of black phosphorus (300\,K) and $A_{\rm th}N^{-3}$ with the kernel amplitude at $N=16384$ (dashed); inset: local exponents.}\label{fig:laws}
\end{figure}

\begin{table}[b]\centering\small
\caption{Classes and laws on the rotation axis ($c=2$; bandwidth measure $s=r$ for electrons). ``exp.'': no protected axis. BP from first-principles face data (this work); phonons and MoS$_2$ from the models of~\cite{SM}. $^\dagger$Kramers value 1 below $p^*\approx(2$--$5)\times10^{-4}$\,\AA$^{-1}$ (thermal energies). $^\ddagger$Above $\lambda_c$ the transparent energy $E^*$ gives $N^{-3}$ instead (forward-channel model, outside the spin--orbit window $p\lesssim p_h$).}\label{tab:mat}
\begin{ruledtabular}
\begin{tabular}{llccc}
material & carrier / branch & $r$ & $\kappa$ & $g\propto$\\
\hline
BP & $\Gamma$ conduction edge (A) & 0 & 2$^\dagger$ & $N^{-1/2}$\\
 & $\Gamma$ valence edge (B) & 2 & 2 & exp.$^\ddagger$\\
 & flexural phonons & 0 & 4 & $N^{-1/4}$\\
 & in-plane phonons & 2 & 2 & exp.\\
MoS$_2$ & $\Gamma$ valence ($\pm\frac12$) & 0 & 3 & $N^{-1/3}$\\
 & out-of-plane phonons & 0 & 6 & $N^{-1/6}$\\
 & in-plane phonons & 2 & 4 & $N^{-3/4}$\\
\end{tabular}
\end{ruledtabular}
\end{table}

\begin{figure*}[t]\centering
\includegraphics[width=\textwidth]{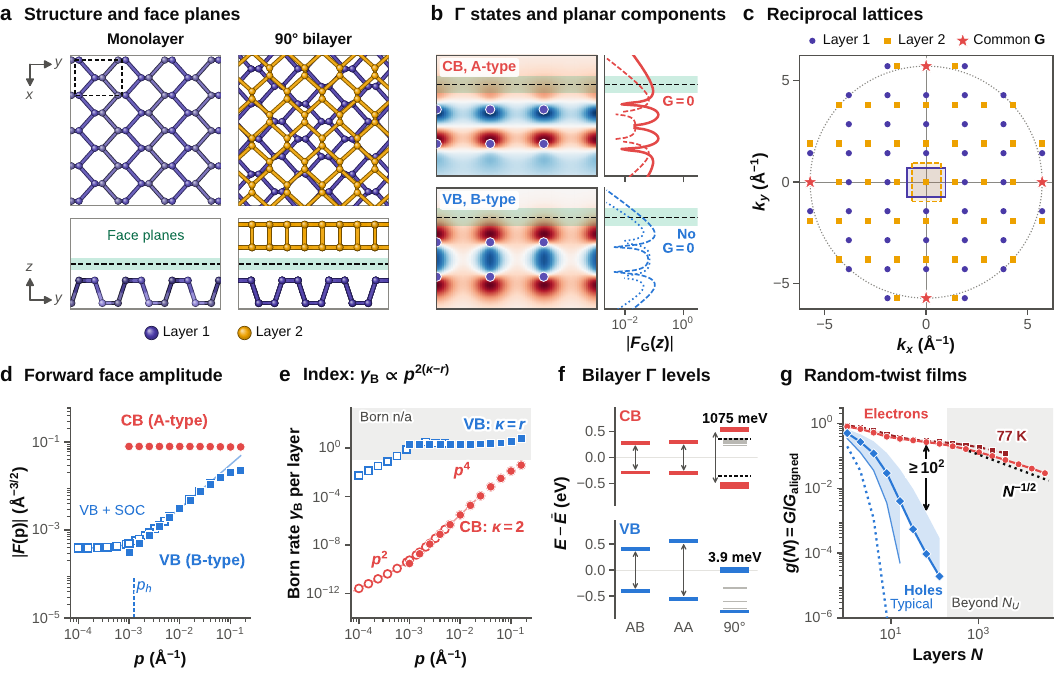}
\caption{Black phosphorus from first principles (PBE): twist disorder as a symmetry filter. (a)~Monolayer and rigid $90^\circ$ bilayer (layer distance $D=5.239$\,\AA; $x$: zigzag, $y$: armchair; dashed box: unit cell), with the admissible face planes (green) and midplane (dashed); structural data in~\cite{SM}. (b)~Monolayer $\Gamma$ states (Re$\,\psi$) and planar components: $\mathbf G=0$ (solid), first armchair (dashed) and zigzag (dotted; zero for the conduction edge) shells. (c)~Reciprocal lattices of the two layers; stars: first common vectors. (d)~Forward face amplitude and (e)~Born rate per layer (monolayer kernel); open: with spin--orbit coupling; lines: fits; grey: Born not applicable. (f)~$\Gamma$ levels of rigid bilayers about each pair's centroid; at $90^\circ$ line width: weight of the isolated-layer pair (grey: below 2\%); dashed: $\pm|\bar T_c(0)|$, monolayer kernel. (g)~Mean conductance of random-twist films relative to aligned stacks (exact forward-channel transport): electrons (300, 77\,K) and holes (300\,K) at $\lambda=1.53$ (markers), bands over $\lambda=1$--2.15 (arrow: band edges); thin line: holes at $\lambda=1$, to 16 layers; dotted: typical hole channel and guide $N^{-1/2}$; grey: beyond the electron umklapp estimate $N_U$.}\label{fig:bp}
\end{figure*}

\emph{Transparent energies at nodes}---A band edge whose forward face is odd under a mirror (B-type) couples as $p$, so bond and rotational on-site disorder both scale as $p^2$ and the axis is not protected: $\kappa=r$. Yet at small $p$ its bond factorizes, $t_l=t_0s_ls_{l+1}$, with $s_l=|\sin\alpha_l|$ and $\alpha_l$ the angle of $\bp$ in the frame of layer $l$, and its on-site energy is $e_l=E^*+2\Delta s_l^2$, with $E^*$ that of a layer whose node faces $\bp$. With $\phi_l=s_l\psi_l$,
\begin{equation}
\Big[\frac{E-E^*}{s_l^2}-2\Delta\Big]\phi_l=t_0\,(\phi_{l+1}+\phi_{l-1}).\label{eq:dimer}
\end{equation}
At $E^*$ on-site and bond disorder cancel through a local transformation known for correlated disorder~\cite{Flores1989}, and if the mapped band is open, $\Delta<t_0$, a film between matched leads transmits perfectly for every twist sequence. Near $E^*$ the nodes make the disorder $V_l=\delta/s_l^2$, $\delta=(E-E^*)/t_0$, heavy tailed, $P(V>v)=\frac2\pi\arcsin\sqrt{\delta/v}$, with infinite variance, outside the finite-variance class of one-dimensional localization~\cite{TitovSchomerus2003,Bienaime2008,Texier2020}. Hence $\gamma\propto|\delta|^{1/2}$, and the mean transmission becomes a function of $q=N|\delta|^{1/2}$ alone, a kernel of its own: dilute strong scatterers give $1-\langle T\rangle\simeq q/2\sqrt{\sin k}$ and $-\langle\ln T\rangle\simeq2q/\sqrt{\sin k}$ ($k$: wave number at $E^*$), a ratio of $\frac14$ instead of the finite-variance value 1~\cite{SM}. Exact transport shows this collapse from $N=256$ to 4096 [Fig.~\ref{fig:laws}(b)], and each channel transmits $\int\langle T\rangle\,\dd E=Jt_0N^{-2}$, with $J$ constant to 1\%. In a real band, higher angular harmonics leave a residue $\gamma(E^*)\simeq C_\lambda p^4$ (Appendix~C). The thermal zero set is then a mixed-order point, $p=0$, $E=E^*$, with $\gamma\propto p^4$ along momentum and $\propto|E-E^*|^{1/2}/p$ along energy; exact transport shows the mean transmission approaching a limit at fixed $u=C_\lambda Np^4$ and $v=N|\delta|^{1/2}$~\cite{SM}, and the change of variables gives $g\propto N^{-3}$ [Fig.~\ref{fig:laws}(c)].

\emph{Black phosphorus}---Because the index is a property of monolayer wavefunctions, a monolayer calculation predicts which band edges are protected, and at what order. For black phosphorus [PBE; Fig.~\ref{fig:bp}(a)] we read the planar components of the band-edge states around $\Gamma$ on five planes in the van der Waals gap. The conduction edge is A-type: its planar average is finite [Fig.~\ref{fig:bp}(b)] and its Born rate grows as $p^4$, so $r=0$ and $\kappa=2.00$ on every plane. The valence edge is B-type: its forward amplitude grows linearly and $\kappa=r=2$ [Fig.~\ref{fig:bp}(d,e)]. These parities agree with band representations~\cite{RudenkoKatsnelson2014,FariaJunior2019} and with the electron--hole contrast of twisted bilayers~\cite{Sevik2017,Zhao2021,Cao2016,Kang2017}; in thick random stacks they become thickness laws. Spin--orbit coupling matters only at the smallest momenta (Table~\ref{tab:mat}; Appendix~C).

A rigid $90^\circ$ bilayer, whose coupling is forward [Fig.~\ref{fig:bp}(c)], gives coupling scales 1.5 to 2.15 times the monolayer estimate [Fig.~\ref{fig:bp}(f); Appendix~C]; scanning this ratio over $\lambda=1$--2.15 crosses the hole threshold $\lambda_c=1.37$ and leaves the indices unchanged~\cite{SM}.

Exact transport of random-twist films uses the full face Hamiltonian from PBE, with aligned leads and a thermal measure [Fig.~\ref{fig:bp}(g); Appendix~C]. Electrons keep $g=0.32$--0.36 of the aligned conductance at 16 layers and 0.23--0.24 at 128 (300\,K, all $\lambda$), with a local exponent approaching $-1/2$ ($-0.47$ at $2.3\times10^4$ layers). A typical electron channel transmits only about $2\times10^{-3}$ at 16 layers: the channels near the axis carry the mean, as the theorem predicts. Holes decay exponentially for $\lambda<\lambda_c$ and algebraically above it, with local exponents reaching $-2.6$ ($\lambda=1.53$) and $-2.8$ (2.15) at 700 layers [Fig.~\ref{fig:laws}(c)]. Twist disorder thus acts as a symmetry filter: it retains at least 7 times more of the electron than of the hole conductance at 16 layers, and $10^2$ times more at 64.

These results hold in the forward channel. First-shell umklapp partners couple with 0.34--0.37 of the forward strength but are on shell mainly within about $\pm3.6^\circ$ of alignment, so electrons keep their momentum over an estimated $N_U\sim2\times10^2$ (300\,K) to $4\times10^2$ layers (77\,K)~\cite{SM}; gated vertical transport~\cite{Srivastava2021} through such films can measure the electrons' retained conductance.

\emph{Phonons}---With the transverse wave vector as $\bp$ the analysis carries over. In a stress-free, mechanically stable paracrystal~\cite{Hosemann}, out-of-plane phonons decay as $N^{-1/\kappa_z}$ ($\kappa_z=6$, 4 and 4 for C$_6$, C$_4$ and flexural C$_2$ layers), and in-plane phonons, whose interlayer coupling is rank one and $O(p^2)$~\cite{Maity2020}, as $N^{-3/4}$ in C$_3$ and C$_6$ layers and exponentially in C$_2$ and C$_4$ layers, where $\kappa=r$~\cite{SM}.

\emph{Outlook}---Rotational disorder thus has a structure of its own: generated by a symmetry operation, it vanishes on the rotation axis, and time reversal decides, harmonic by harmonic, how fast its backscattering vanishes. The exponents that follow are set by the symmetry of the band edge, while material details change only the amplitudes and random twist increments only the rates; no natural proxy---single-bond spectra, admixture sizes, projected spin splittings or the attenuation of a typical channel---captures them. Three directions open. Stacking disorder becomes a design element for thermoelectrics~\cite{Bai2026} and beyond: choosing layers by the parity of their band edges and phonon branches selects what crosses a disordered stack, as the hundredfold electron--hole contrast of the black-phosphorus model illustrates, and rotated photonic or acoustic slabs offer clean tests, layer by layer. Twist statistics become a control knob for the rates: between random stacks and the helical crystals grown around screw dislocations~\cite{Liu2019,Sutter2019}, random increments keep the scattering orders, and strictly regular twists make the chain periodic or quasiperiodic~\cite{SiriviboonVolkov2026}. And disorder generated by random elements of other symmetry groups, such as random twinning, may likewise fall silent at the momenta the group fixes---a route to classifying disorder by the symmetry that generates it.

\begin{acknowledgments}
We acknowledge support from the National Science Fund for Distinguished Young Scholars (Grant No.~51925101), the National Natural Science Foundation of China (Grants No.~52450001 and No.~12104370), the Tianmushan Laboratory Research Project (Projects No.~TK2024D006 and No.~TK2023C021), and the ``Pioneer'' and ``Leading Goose'' R\&D Program of Zhejiang (Project No.~2024SSYS0084). P.K. is grateful to Compute Canada and the High Performance Computing Center of McGill University for computing resources and support during the early stages of this work. The authors take full responsibility for the ideas, reasoning and rigor of this work; the large language model Claude Opus 5.5 (Anthropic) served as a supporting tool, under the authors' direction and review, for derivations, code, data analysis and writing.
\end{acknowledgments}

\emph{Data availability}---The first-principles face data are in the Supplemental Material package; the scripts and validation kit that reproduce every number and figure are at \url{https://github.com/kpleo/twist_theory}.

\section*{End Matter}

\section*{Appendix A: Proof of the parity rule}\label{app:parity}
The Proposition requires $\bar T\neq0$ on the axis and an energy away from band edges, crossings of the averaged chain and channel thresholds. Take the lab momentum along $x$ and let $a=\pm$ label $s_z,J_z=\pm\frac12$. (i)~At $p=0$ only the $m=0$ harmonics survive, so every profile block is diagonal, $G_\alpha(0)={\rm diag}(c_\alpha,c_\alpha^*)$, since time reversal maps the $(+,+)$ entry to the conjugate of the $(-,-)$ one. (ii)~The smallest twist-dependent harmonics are the spin-flip ones, of in-plane order $n-1$; time reversal, $G_{+-}(\mathbf q)=-G_{-+}(-\mathbf q)^*$, fixes their form,
\begin{equation}
\delta G_\alpha(\theta)=p^{\,n-1}\begin{pmatrix}0&(-1)^na_\alpha^*e^{-in\theta}\\ a_\alpha e^{in\theta}&0\end{pmatrix}+O(p^n).\label{eq:dG}
\end{equation}
(iii)~Summing over profiles with a real kernel gives $\bar T={\rm diag}(t_0,t_0^*)+O(p)$ and $\delta T_l=-\Omega_{{\rm t},l}\bar T+\bar T\,\Omega_{{\rm b},l+1}+O(p^n)$, where the purely spin-flip generators $\Omega$, of order $p^{n-1}$, obey $\Omega^\dagger=-\Omega$ for odd $n$ and $\Omega^\dagger=\Omega$ for even $n$. (iv)~Let the averaged chain be $\bar T=cU$ ($U$ unitary) and $\bar X=\varepsilon\mathbb 1$ (Kramers degeneracy). A forward channel $(k_a,\phi_a)$ and a backward channel $(k_b,\phi_b)$ of eigenphases $\phi$ of $U$ at the same energy obey $k_b+\phi_b=-(k_a+\phi_a)$, and for anti-Hermitian $\Omega$ the Born amplitude of $\delta T_l$,
\begin{multline}
m_l=c\,(\Omega_{{\rm t},l})_{ba}\big[e^{-i(k_b+\phi_b)}-e^{i(k_a+\phi_a)}\big]\\+c\,e^{i(\phi_b-\phi_a)}(\Omega_{{\rm b},l+1})_{ba}\big[e^{i(k_a+\phi_a)}-e^{-i(k_b+\phi_b)}\big],\label{eq:ml}
\end{multline}
vanishes identically; for Hermitian $\Omega$ the same algebra gives $m_l=-2c\,e^{i(k_a+\phi_a)}[(\Omega_{{\rm t},l})_{ba}-e^{i(\phi_b-\phi_a)}(\Omega_{{\rm b},l+1})_{ba}]$, generically non-zero. For odd $n$ the $O(p^{n-1})$ amplitude therefore cancels, and scattering starts at $O(p^n)$, from the spin-conserving harmonics, from the rotation acting on the $O(p)$ non-scalar part of the averaged chain, and from products of fluctuations: $\kappa=n$. For even $n$ the harmonic scatters at its own order: $\kappa=n-1$. For spinless A-type states the harmonic of order $n$ is real for even $n$, a modulus change ($\kappa=n$), and imaginary for odd $n$, a phase that is an exact gauge of a scalar chain, so the modulus first changes at order $2n$: $\kappa=2n$. Neither a face relation nor the number of profiles enters~\cite{SM}.

\section*{Appendix B: Why one bond is not enough}\label{app:onebond}
Single-bond quantities, such as the twist spread of the singular values of $T_l$, miss two features of the chain. First, an interface couples the top face of one layer to the bottom face of the next. For square invertible faces and $\tau=1$ the substitution $\phi_l=G^{\rm b}_l\psi_l$ gives the exact face chain
\begin{equation}
Y_l\phi_l=K_l^\dagger\phi_{l+1}+K_{l-1}\phi_{l-1},\label{eq:facechain}
\end{equation}
with $K_l=G^{\rm t}_l(G^{\rm b}_l)^{-1}$ and $Y_l=(G^{\rm b}_l)^{-\dagger}(E-X_l)(G^{\rm b}_l)^{-1}$: lab-frame matrices whose eigenbases rotate from layer to layer, which singular values of one bond cannot see. Second, part of the twist dependence is a rotation of the multiplet frame, which the chain removes by a gauge transformation. The centrosymmetric C$_3$ Kramers doublet of Fig.~\ref{fig:matrix}(d) shows both: single-bond spread $\propto p^5$ (same face) or $p^3$ (opposite faces), chain rate $p^6/(1-p^2)^2$, $\kappa=3$, confirmed by transfer matrices.

\section*{Appendix C: The transparent energy and first-principles transport}\label{app:transport}
\emph{First-principles inputs.} The PBE monolayer (projector-augmented waves, experimental lattice) gives the planar components and slopes of the band-edge states at 577 momenta around $\Gamma$ on five gap planes. The rigid $90^\circ$ bilayer ($D=5.239$\,\AA, 96 atoms) shares no reciprocal-lattice vector shorter than 5.72\,\AA$^{-1}$, so its coupling is forward. In the projection of its $\Gamma$ states on those of the isolated layers, two states 1.075\,eV apart carry 81\% of the conduction pair, with a coupling of 0.54\,eV in a two-state model and 0.60--0.76\,eV downfolded, 1.5 to 2.15 times the monolayer's midplane Bardeen estimate, 0.35\,eV~\cite{SM}. Spin--orbit coupling lowers the conduction index to the Kramers value 1 only below $p^*$ (Table~\ref{tab:mat}), which matters beyond about $10^{10}$ layers, and makes the holes a Kramers pair below $p_h=(1.2$--$1.4)\times10^{-3}$\,\AA$^{-1}$, with $(3$--$5)\times10^{-5}$ of the electrons' forward coupling.

\emph{Residue.} At $E^*$ the mapped chain of Eq.~\eqref{eq:dimer} is exact for any band: $\varepsilon_l=(e_l-E^*)/s_l^2$ and $\tau_l=t_l/s_ls_{l+1}$ are smooth in the angles, because the band is even in $q_y$ and the B-type face odd. Expanding the band, $E=E_0+e_{10}q_x^2+e_{01}q_y^2+e_{20}q_x^4+e_{11}q_x^2q_y^2+e_{02}q_y^4$, and the faces, $G=q_y(c_0+c_xq_x^2+c_yq_y^2)$, with $K_0=c_0^{{\rm t}\,T}\tau c_0^{\rm b}$ and $C=\lambda|K_0|$, gives the clean chain $\cos k=-(e_{01}-e_{10})/2C$ and, from the angle dependence of on-site and bond corrections, correlated disorder of relative order $p^2$~\cite{Tessieri2001}, with the Born residue
\begin{equation}
\begin{split}
\gamma(E^*)&=\frac{|B|^2p^4}{64\sin^2k}+o(p^4),\\
B&=A+2e^{ik}(f_{{\rm L}y}-f_{{\rm L}x})+2e^{-ik}(f_{{\rm R}y}-f_{{\rm R}x}),
\end{split}\label{eq:residue}
\end{equation}
where $A=(e_{20}-e_{11}+e_{02})/C$, $f_{{\rm L}a}=c_a^{{\rm t}\,T}\tau c_0^{\rm b}/K_0$ and $f_{{\rm R}a}=c_0^{{\rm t}\,T}\tau c_a^{\rm b}/K_0$. For black phosphorus $C_\lambda=|B|^2/64\sin^2k=3.55\times10^3$\,\AA$^4$ ($\lambda=1.53$) and $5.76\times10^2$\,\AA$^4$ ($\lambda=2.15$); exact Born sums and chains of up to $5\times10^8$ layers ($N\gamma\ge190$) agree with each other within about 2\% for $p\le0.02$ (0.057)\,\AA$^{-1}$ and tend to it as $p\to0$; as $p\to0$ the transparent energy exists for $\lambda>\lambda_c=|e_{01}-e_{10}|/2|K_0|=1.373$.

\emph{Mixed-order law.} With $u=C_\lambda Np^4$, $v=N|\delta|^{1/2}$ and $\dd^2p\,\dd E=p\,\dd p\,\dd\varphi\,t_0\,\dd\delta$, $t_0=Cp^2$, a scaling limit $\langle T\rangle\to\Psi(u,v,\pm)$ gives
\begin{equation}
g\simeq A_{\rm th}N^{-3},\qquad A_{\rm th}=\frac{\pi Cw^*}{2C_\lambda\mathcal Z}\sum_\pm\int\dd u\int2v\,\dd v\,\Psi,\label{eq:Ath}
\end{equation}
where $w^*$ is the Boltzmann density at $E^*$ and $\mathcal Z$ the aligned normalization; the azimuthal integral assumes that $\Psi$ does not depend on the direction of $\bp$ relative to the leads, which a second lead azimuth supports within the errors~\cite{SM}. Exact transport from $N=256$ to 65536, with independent random streams per point, tests this scaling limit. At $\lambda=2.15$ the means at fixed $(u,v)$ change by a median of 2.2\% from $N=1024$ to 4096 and 0.6\% from 4096 to 16384 (points with $\langle T\rangle>10^{-4}$); between 4096 and 16384 they agree point by point within their errors ($\chi^2=210$ for 207 points, 153 for the 152 not dominated by a few sequences), as do 15 points at 65536 ($\chi^2=18$), although a coherent rise of up to 3\% per factor 4 in $N$ remains. At $\lambda=1.53$, where $E^*$ lies near the band edge, $O(p^2)$ corrections decay slowly. For $\lambda=2.15$, $A_{\rm th}=846\pm4$ and $868\pm8$ at $N=4096$ and 16384 on the base grid, and $834\pm5$ and $837\pm6$ with the quadrature error removed to leading order (Richardson, from a grid refined by all midpoints; other rules differ by at most 0.5\%); the tails beyond the grid add an extrapolated 1.1--1.4\%. The extrapolated amplitude is thus stable to a few per cent over the tested thicknesses, with residual finite-size and tail-extrapolation uncertainty. The direct $N^3g(1024)=841\pm26$ (about 3\% lower with a finer energy grid) is consistent, and the kernel integrated at the same $N$ reproduces the direct thermal means within 5\%.

\emph{Faces and chain.} At each grid momentum the midplane face values and slopes of an edge state are fixed to a real gauge with the sign set by the symmetry of the edge, and represented, with the band $E(\mathbf q)$, by even polynomials of degree 10 in $(q_x,q_y)$, times $q_y$ for B-type edges~\cite{SM}. Channel $\bp$ has on-site energies $E(R_{\theta_l}^{-1}\bp)$ and bonds $t_l=\lambda|T(R_{\theta_l}^{-1}\bp,R_{\theta_{l+1}}^{-1}\bp)|$; twists are continuous, i.i.d.\ and uniform, and the leads are semi-infinite aligned stacks. We sample $|\bp|\le0.16$\,\AA$^{-1}$ with 32 energies per channel at equal Boltzmann probability in the lead band, plus 40 geometric points per side around $E^*$ for holes, and report the mean over 16--512 twist sequences and the typical channel $e^{\langle\ln T\rangle}$. Refining grids and polynomials changes the electron means by less than 0.2\% and the hole means by at most 3\%~\cite{SM}.
\end{document}


\maketitle
\vspace{-6pt}
\begin{abstract}
\noindent This supplement to the letter gives the proofs (the thickness theorem with class-dependent kernels, transparent energies at nodes with their heavy-tailed kernel, and mixed-order zero sets), the matrix-index classification with its validation and its robustness to random twist increments, the stable phonon model, the held-out tests, the thermoelectric window, the material protocol with its validation kit, and the first-principles calculations for black phosphorus: face data, the $90^\circ$ bilayer and its effective model, exact transport of random-twist films with the full face Hamiltonian, the transparent energy of the holes with its residue and two-variable kernel, and the umklapp and spin--orbit limits. Section~\ref{sec:summary} states the scope and the status of every claim.
\end{abstract}

\tableofcontents
\newpage

\input{sec_summary}
\input{sec_framework}
\input{sec_index}
\input{sec_theorem}
\input{sec_phonon}
\input{sec_spacing}
\input{sec_predict}
\input{sec_thermo}
\input{sec_beyond}
\input{sec_validation}
\input{sec_dft}
\appendix
\input{sec_appendix}

\input{sec_refs}

%% file: sec_summary.tex
\section{Overview and claims ledger}\label{sec:summary}

\paragraph{The question and the answer.} How fast does the cross-plane conductance of a random-twist stack decay with the number of layers, and what decides it? The letter answers in three steps, each proved or tested here. (i)~\emph{A thickness theorem}: the mean conductance decays algebraically near the zero set $\Zg$ of the effective Lyapunov exponent, with an exponent fixed by the codimension, the local order and the measure, and an amplitude fixed by a mean-transmission kernel that belongs to the class of the disorder, not to the Lyapunov exponent (Sec.~\ref{sec:theorem}). (ii)~\emph{A parity rule}: on the rotation axis the order of the zero set is a property of the full matrix chain of the facing wavefunctions; time reversal sorts the twist harmonics into frame rotations, whose backscattering cancels, and coupling changes, which scatter. This fixes the index $\kappa$ in closed form for spinless and Kramers band edges with $r=0$ in C$_n$ layers ($n\ge2$), and through random models or face-resolved material data for the other classes (Sec.~\ref{sec:index}). Random twist increments rescale the Born rate through non-negative spectral weights of the twist sequence and, generically, leave $\kappa$ unchanged; a strictly regular twist makes the chain periodic or quasiperiodic (Sec.~\ref{sec:corrtwist}). (iii)~\emph{Transparent energies at nodes}: where the axis is not protected, a nodal (B-type) band edge with a factorized coupling and an open mapped band transmits perfectly at an energy $E^*$; near it the residual disorder is heavy tailed, and at the resulting mixed-order zero set the thermally averaged conductance falls as $N^{-3}$ (Corollaries~\ref{cor:heavy} and~\ref{cor:mixed}, Sec.~\ref{sec:thicknessdft}). Black phosphorus realizes (ii) and (iii) from first principles within the forward-channel model (Sec.~\ref{sec:dft}). The remaining sections test the framework in a stable, stress-consistent model (Secs.~\ref{sec:phonon}--\ref{sec:thermo}), treat what lies beyond the coherent forward channel (Sec.~\ref{sec:beyond}) and list the material calculations that identify each input (Sec.~\ref{sec:validation}).

\paragraph{Scope.} All results are derived for incommensurate interfaces in the forward-channel model (Assumption~\ref{as:cont}): in-plane momentum is conserved in the projected forward block, which contains the couplings between states at the same lab-frame momentum when no reciprocal-lattice vectors of adjacent layers coincide. Generalized umklapp couplings connect distinct momentum channels already at first order; eliminated when their partner states are off shell ($\|H_{PQ}\|/\Delta\ll1$), they correct the retained block only at second order. On-shell umklapp partners and moir\'e processes near alignment lie outside the model; Sec.~\ref{sec:umklapp} estimates their range by an umklapp thickness $N_U$, which needs the umklapp matrix elements and the partner energies, not only the planar Fourier ratios of a monolayer calculation (Sec.~\ref{sec:validation}). The matrix index is a Born-order statement (Sec.~\ref{sec:parity}); exact transfer-matrix exponents are compared with it over accessible windows (Sec.~\ref{sec:matrixval}). Twists are independent and uniform unless stated; Sec.~\ref{sec:corrtwist} treats correlated sequences.

\paragraph{Relation to earlier work.} The ingredients have precedents; what this work adds is the connection between them. (i)~\emph{Channel chains of rotated layers.} At fixed in-plane momentum a stack of rotated anisotropic layers is a one-dimensional chain; for a constant twist increment the chain is quasiperiodic (Aubry--Andr\'e type), and channels below a critical momentum stay extended through the stack~\cite{SiriviboonVolkov2026}. For independent random twists every channel with $p>0$ localizes, and the thickness law is set by how fast backscattering vanishes on the axis, which the parity rule fixes from the symmetry of the layer (Sec.~\ref{sec:index}). (ii)~\emph{Direction-selective transparency.} Brewster anomalies of randomly stratified optical media~\cite{Sipe1988} and normal incidence in twist-angle-disordered bilayer graphene~\cite{Joy2020} are zero sets, at least to leading order, in the sense of Sec.~\ref{sec:theorem}; in the graphene case the angle-integrated typical transmission across lateral domain walls decays as a power of their number, and a shift of the Dirac points cuts the power law off exponentially. The theorem concerns the ensemble-mean transmission across layers, whose kernel depends on the class of the disorder (H4), so such typical-value laws are not its amplitudes. (iii)~\emph{Correlated and heavy-tailed disorder.} Correlated disorder can cancel at one energy~\cite{Flores1989,Tessieri2001,Hilke1997}, and infinite-variance disorder changes the scaling of one-dimensional localization~\cite{TitovSchomerus2003,Bienaime2008,Texier2020}. Here the nodes of the band-edge faces, together with the factorized coupling and an open band, generate a one-sided tail of index $\frac12$ at fixed positions, the residue $\propto p^4$ cuts it off, and the joint scaling in momentum and energy gives the thermal law $N^{-3}$ (Corollaries~\ref{cor:heavy} and~\ref{cor:mixed}). Stealthy disorder removes the leading orders of backscattering through a gap in its spectrum and delocalizes effectively over long but finite lengths~\cite{Vanoni2026}; our zero sets follow from the layer symmetry instead. (iv)~\emph{Phonons.} Disorder in the layer sequence~\cite{Juntunen2019} or at the interfaces~\cite{Luckyanova2018} suppresses coherent cross-plane phonon transport in superlattices, and so do disordered twist sequences in multilayer graphene~\cite{Wang2025}; the index predicts which branches keep algebraic tails in random-twist stacks (Sec.~\ref{sec:phonon}). (v)~\emph{Ordered and correlated twists.} Helical van der Waals crystals grow around axial screw dislocations with a fixed mean twist rate, continuous~\cite{Sutter2019} or in discrete jumps between nanoplates~\cite{Liu2019}, and a strictly constant increment makes the channel chain quasiperiodic~\cite{SiriviboonVolkov2026}. In Eshelby-twisted GeS nanowires the layers rotate by less than $0.01^\circ$ per layer, and the measured changes of the cross-plane conductivity are attributed to a strain-stabilized stacking change and to compression near the dislocation core, not to the rotations~\cite{Liu2025}. Between such ordered stacks and the independent twists of turbostratic films, random increments rescale the backscattering rate through the spectral weights of the twist sequence and generically leave the index unchanged (Sec.~\ref{sec:corrtwist}).

\begin{table}[!htbp]\centering\footnotesize\setlength{\tabcolsep}{4pt}
\caption{Claims ledger. Status: \textbf{T} theorem or Born-order result with checkable hypotheses; \textbf{P} held-out prediction in the stable model, with uncertainty; \textbf{C} conditional on named material parameters; \textbf{S} supplementary (evidence limited to a stated regime).}\label{tab:ledger}
\begin{tabular}{L{4.8cm}L{8.1cm}L{2.6cm}}
\toprule
claim & evidence & status\\
\midrule
$g\propto N^{-(c+s)/2m}$ on $\Zg$ of the effective exponent, with amplitude; the kernel belongs to the disorder class ($\Phi_1$ for finite variance) & proof (App.~\ref{app:theorem}; rate under a global remainder bound); $r=1$, zero-ring and transparent-energy tests & T\\
exponential decay if $\gamma\ge\gamma_{\min}>0$ (uniform large-deviation bound); algebraic if $\Zg\neq\emptyset$ & corollary; C$_2$/C$_4$ in-plane chains decay exponentially & T + P\\
matrix index of $r=0$ band edges of C$_n$ layers, $n\ge2$: $n$ / $2n$ (spinless A), $n-1$ / $n$ (Kramers), any face relation and number of interface profiles & Proposition~\ref{prop:matrix} at Born order (bond-parity and frame-rotation lemmas on the full bond contraction, App.~\ref{app:matrix}); random models, 4 face relations $\times$ 5 seeds, and with two and three interface profiles; exact transfer matrices over accessible windows & T (Born order; generic coefficients) + P\\
C$_1$ Kramers pairs: $\kappa=1$ & single interface profile (common spinor frame) only; with two profiles the axis is not protected (explicit counterexample, Sec.~\ref{sec:multiprofile}) & T (restricted)\\
index of $r>0$ classes and doublets & random-model classification (Table~\ref{tab:kappa}); per material from face-resolved data (kit) & P + C\\
index under correlated twists (random increments) & Born rate as a sum of non-negative spectral weights of the twist sequence, frame-rotation cancellation for every sequence (Sec.~\ref{sec:corrtwist}); random C$_2$ and C$_3$ Kramers models, three statistics: rates up to 250-fold lower, local index unchanged; exact exponents at larger $p$ & T (Born order; generic increments) + P\\
forward-channel coupling: bilinear kernel on forward amplitudes, $T(0)$ twist independent for $n\ge2$ & proof (App.~\ref{app:coupling}) under a weak-overlap L\"owdin convention; T2 test specified & T (stated convention)\\
branch exponents (z $1/\kappa_z$; in-plane $3/4$ or exponential) & stable $3\times3$ chains, $N\le2048$, gated & T + P\\
spacing-noise transparency surface $\omega^*(\bp)$ & Born derivation (Corollary~\ref{cor:transparent}); transmission peaks within the frequency grid of $\omega^*$ (within 0.22\,rad/ps for $|\bp|\ge0.56$\,\AA$^{-1}$; every third channel checked) & T (Born order) + P\\
joint disorder from separate runs (additivity) & $\sigma_d$, $N$ and $\beta_{\rm eq}$ hold-out (five families; each interior family left out in turn); paired bootstrap bands and sampling floor & P\\
electron conductance without parameters & 5 spectra $\times$ 4 noise levels, smooth excepted; rms 5.8\% against a sampling floor of 3.1\% & P\\
Seebeck change (Fermi--Dirac) without parameters & 0.04--0.06\,$k_B/e$ rms (sampling floor 0.02--0.03), $\mu=-0.10$ to $+0.05$\,eV & P\\
phonon spacing response set by $\beta_{\rm eq}$, not $\beta_{\rm load}$ & five stiffness families & P + C\\
black phosphorus: $\Gamma$ electrons protected (axis law); $\Gamma$ holes not protected ($\kappa=r$) & PBE face data on five admissible gap planes (Sec.~\ref{sec:dft}): conduction $r=0$, $\kappa=2.00$; valence $\kappa=r=2$; spin--orbit crossovers $p^*\approx(2$--$5)\times10^{-4}$ (thermal) and $p_h=(1.2$--$1.4)\times10^{-3}$\,\AA$^{-1}$; $90^\circ$ bilayer (pseudo-wavefunction surrogates): no single constant calibrates a stack, so the scale is scanned ($\lambda=1$--2.15); exact transport of the full face Hamiltonian: electrons $g(16)=0.32$--0.36, exponent $-0.47$ at $2.3\times10^4$ layers; retained fraction of holes $\ge7$ and $10^2$ times lower at 16 and 64 layers; estimates: electron umklapp $N_U\sim2$--$4\times10^2$, $N^*\approx10^{10}$--$10^{11}$; hole umklapp unbounded; non-forward residue weight $\approx8\%$ & C (inputs computed; forward channel; scale scanned)\\
transparent energy at nodes (B-type chains): exact cancellation at $E^*$ for $\Delta<t_0$; heavy tail of index $\frac12$, $\gamma\propto|E-E^*|^{1/2}$, kernel of its own; $N^{-2}$ per channel (ideal chain); residue $C_\lambda p^4$; thermal $N^{-3}$ & mapping and dilute kernel (Corollary~\ref{cor:heavy}); ideal-chain collapse, $J$ constant to 1\%; residue analytic ($C_\lambda=3.55\times10^3$, $5.76\times10^2$\,\AA$^4$), Born sums and long chains agree within 2\% ($p\le0.02$\,\AA$^{-1}$); two-variable kernel consistent at $\lambda=2.15$ from $N=4096$ to 65536 (slow residual rise; $A_{\rm th}$ stable to a few \% over the tested $N$), slower at 1.53, reproduces thermal means within 5\%; local exponents $-2.6$ to $-2.9$ at 360--720 layers & T (mapping, dilute limit, residue) + C (kernels supported numerically; $N^{-3}$ by Corollary~\ref{cor:mixed})\\
mixed-order zero sets: $\alpha=\sum_j\nu_j(c_j+s_j)$ & change of variables under a joint scaling limit of the mean transmission (Corollary~\ref{cor:mixed}) & T (given the limit)\\
selectivity / optimal-ZT window & Fermi--Dirac, inside 0--0.09\,\AA; thresholds in $\beta_{\rm eq}/\beta_e$ as right-censored bootstrap intervals with the fraction of draws not reaching them; direct brackets as an internal consistency check; needs $B$, $g_{\rm reg}$ & C\\
incoherent-limit exponent & Drude scaling argument; B\"uttiker-probe tests not conclusive; scope: channel-conserving dephasing & C + S\\
umklapp floor $N_U$ & second-order effective chain; near-aligned treatment assumed; black phosphorus: first shell from computed components ($\rho_U=0.34$--0.37, order-of-magnitude estimate $N_U\sim2$--$4\times10^2$; transmission factors, velocities and fluxes not resolved) & S + C\\
registry loss $g_{\rm reg}$ & definition; model value shows the sign is material specific & C\\
\bottomrule
\end{tabular}
\end{table}

\paragraph{Open material inputs.} The class and face-resolved amplitudes of each band-edge multiplet, the spin--orbit crossover momenta $p^*$ and $p_h$ and the umklapp ratio (T0-m); $\bar T$, $\beta_e$, $\beta_{\rm eq}$, $F_{zz}(0)$, $\ell_{\rm ph}$ (T1); $g_{\rm reg}$ (T1-d); the parameter-free twisted-bilayer and trilayer tests (T2). Sec.~\ref{sec:validation} lists each calculation with the quantities it identifies. For black phosphorus the T0-m inputs of the $\Gamma$ edges, the scale and distance dependence of the conduction coupling at rigid twisted interfaces, the first-shell umklapp ratio and the finite-$N$ transport of the full face Hamiltonian are computed in Sec.~\ref{sec:dft}; a constant coupling valid in the multilayer window, the umklapp transmission factors, relaxed twisted interfaces and the phonon inputs are not.

%% file: sec_framework.tex
\section{Framework: channels, couplings and stress-free interfaces}\label{sec:framework}

\subsection{Channel decomposition}
Layer $l$ of an $N$-layer stack is rotated by $\theta_l$ about the stacking axis $z$, sits at height $z_l$ with spacing $d_l=z_{l+1}-z_l$ to the next layer, and is shifted laterally by the registry $\bm\tau_l$. Four assumptions define the model; each is relaxed later or tested by a material calculation.
\begin{assumption}[forward channel]\label{as:cont} Interfaces are incommensurate, and only the part of the interlayer coupling that conserves the lab-frame in-plane momentum $\bp$ is kept (the momentum-conserving part of the general coupling rule $\mathbf k+\mathbf G=\mathbf k'+\mathbf G'$ between incommensurate layers~\cite{Bistritzer2010,Koshino2015}, as separated in Ref.~\cite{SiriviboonVolkov2026}). Between states at the same $\bp$ this is exact (Proposition~\ref{prop:coupling}); umklapp couplings $H_{PQ}$ to partner states $Q$ at $\bp+\mathbf G-R_\theta\mathbf G'$ enter the chain of channel $\bp$ at second order when the partners are off shell, $|E-\varepsilon_Q|\ge\Delta>0$ with $\|H_{PQ}\|/\Delta\ll1$ (Sec.~\ref{sec:umklapp}). The twist angle alone does not guarantee this spectral condition. Interfaces close to alignment, where $\mathbf G-R_\theta\mathbf G'$ is small (moir\'e), are outside the model.\end{assumption}
\begin{assumption}[coherence]\label{as:coh} Transport is elastic and phase coherent across the stack. Phase breaking is treated in Sec.~\ref{sec:coherence}.\end{assumption}
\begin{assumption}[independent twists]\label{as:iid} The $\theta_l$ are i.i.d.\ with a distribution that has a continuous part.\end{assumption}
\begin{assumption}[mechanical equilibrium]\label{as:eq} Every interface is at its own force-free spacing: spacing variations come from the interfaces themselves, not from unbalanced loads (Sec.~\ref{sec:stressfree}).\end{assumption}
No assumption on the orbital ($J_z$) purity of the band edges is needed; Proposition~\ref{prop:coupling} holds under a stated orthogonalization convention.

Under \ref{as:cont} the stack splits into independent one-dimensional problems, one per $\bp$. For a band-edge multiplet with on-site block $X(\bp)$ in the layer frame,
\begin{equation}
E\,\psi_l=X\!\left(R_{\theta_l}^{-1}\bp\right)\psi_l+T_{l,l+1}(\bp)\,\psi_{l+1}+T_{l-1,l}(\bp)^\dagger\,\psi_{l-1},\label{eq:echain}
\end{equation}
where each layer's basis is the rotated copy of the reference layer's basis. For phonons, with mass-scaled $3\times3$ blocks,
\begin{equation}
\omega^2 u_l=\Big[R_{\theta_l}D\!\left(R_{\theta_l}^{-1}\bp\right)R_{\theta_l}^{\!\top}+F_{l-1}(0)+F_l(0)\Big]u_l-F_l(\bp)\,u_{l+1}-F_{l-1}(\bp)^{\dagger}u_{l-1},\label{eq:phchain}
\end{equation}
with $D$ the monolayer dynamical matrix and $F_l(\bp)$ the continuum force-constant tensor of interface $l$. Conductances follow from the Landauer formula channel by channel: $G_e=\sum_{\bp}\int\dd E\,\mathcal T_{\bp}(E)(-\partial_Ef)$ and, classically, $G_{\rm ph}=(k_B/2\pi)\int\dd\omega\sum_{\bp}\mathcal T_{\bp}(\omega)$.

\paragraph{Channel domain.} Every layer mode must be counted once. A layer rotated by $\theta$ maps the lab momentum $\bp$ to its own crystal momentum $R_\theta^{-1}\bp$; inside the disk inscribed in the Brillouin zone, $|\bp|\le|\mathbf G_1|/2$, this is a first-zone momentum for every $\theta$. Beyond it, channels at $\bp$ and $\bp+\mathbf G$ describe the same layer mode, coupled by different Fourier components of the interlayer interaction; treating them as independent chains counts folded modes again. A grid to $|\bp|=1.6$\,\AA$^{-1}$ would cover 2.8 times the zone area of the model layers and put half of the aligned phonon conductance into re-counted channels. We use the inscribed disk (0.90\,\AA$^{-1}$ for the triangular model layers, 0.78\,\AA$^{-1}$ for the square ones). The corner regions left out hold 9\% of the zone for triangular and 21\% for square layers; they are where umklapp terms are largest and the forward-channel picture is least reliable anyway.

\subsection{The continuum coupling as a bilinear face kernel}\label{sec:coupling}
Let $\Psi_i(\bp)$ ($i=1,2$) be the band-edge multiplets of two isolated layers at lab momentum $\bp$, with $H_i\Psi_i=\Psi_iE_i$, $H_i=K+V_i$ and $H=K+V_1+V_2$ for the pair ($K$ kinetic, $V_i$ layer-periodic, local or non-local). With the overlap $S_{12}=\Psi_1^\dagger\Psi_2$ and $K_{12}=\Psi_1^\dagger K\Psi_2$, the two eigen-equations give exactly
\begin{equation}
H_{12}\equiv\Psi_1^\dagger H\Psi_2=E_1S_{12}+S_{12}E_2-K_{12}.\label{eq:H12}
\end{equation}
Isolated-layer states overlap, so $H_{12}$ is not yet the hopping of an orthonormal chain. We use symmetric (L\"owdin) orthogonalization to first order in the overlap, which gives
\begin{equation}
T_{12}=H_{12}-\tfrac12\big(E_1S_{12}+S_{12}E_2\big)=\tfrac12\big(E_1S_{12}+S_{12}E_2\big)-K_{12}+O(S^2).\label{eq:Torth}
\end{equation}
At larger overlap, orthogonalization adds longer-range hoppings and diagonal terms of order $S^2$. They are of the same symmetry type and are part of the disorder model below, not a new mechanism; the weak-overlap convention is an assumption, stated as such.
\begin{proposition}[continuum coupling]\label{prop:coupling}
Let the interface be incommensurate. Then (a) every term of Eq.~\eqref{eq:Torth} involves only the in-plane Fourier components of the two states at wavevector exactly $\bp$ (no reciprocal vector of either layer), and (b) the coupling is a bilinear form in these components,
\begin{equation}
T_{l,l+1}(\bp)=\sum_{\alpha\beta}G^{\rm t}_{l,\alpha}(\bp)^\dagger\,\tau_{\alpha\beta}(\bp;d_l)\,G^{\rm b}_{l+1,\beta}(\bp),\qquad G^{\rm t,b}_{l,\alpha}(\bp)=U_s(\theta_l)\,G^{\rm t,b}_\alpha\!\left(R_{\theta_l}^{-1}\bp\right),\label{eq:formfactor}
\end{equation}
where $\alpha,\beta$ label a finite set of $z$ profiles near the interface (for example the outward $p_z$ lobes of the outer atoms), $G^{\rm t,b}_\alpha(\mathbf q)$ is the projection of the planar Fourier component at in-plane wavevector $\mathbf q$ of the multiplet, near its top (bottom) face, on profile $\alpha$ (spin $\times$ multiplet index), $U_s$ is the spin rotation, and $\tau_{\alpha\beta}$ is the lab-frame kernel between profiles, spin conserving. A single plane-wave component carries no in-plane angular momentum of its own: all angular dependence sits in the $\mathbf q$ dependence of $G$. The one-profile case, $T=t\,G^{\rm t\dagger}G^{\rm b}$, is a separability approximation; cross-profile terms are kept in general.
\end{proposition}
Consequences, as \emph{generic} statements (the symmetry allows a coefficient; it does not force it to be non-zero):
\begin{enumerate}
\item At $\bp=0$ only the in-plane average ($m_l=0$ component) of each state enters, so only multiplet components with $J_z=s_z$ couple. The coupling at $\bp=0$ is generically non-zero iff the C$_n$ class of the multiplet contains $J_z\equiv s_z$ ($J_z\equiv0$ spinless, $\pm1/2$ with spin); then $r=0$. Otherwise $r=2\mu_0$ generically, with $\mu_0$ the smallest $|J_z-s_z|$ in the class. A vanishing coefficient, an extra symmetry or a profile cancellation can raise $r$.
\item For $n\ge2$ the class fixes $s_z$ of the in-plane average, and $T(0)$ is independent of the twist up to basis phases. In C$_1$ layers the spinor of the in-plane average is not tied to $z$ and rotates with the layer: with one profile (or profiles sharing one spinor frame) $T(0)$ is twist dependent through $U_s$, a frame rotation (Sec.~\ref{sec:index}); profiles in different frames make $T(0)$ a sum of different rotations, which is not a rotation (Sec.~\ref{sec:multiprofile}).
\item Twist dependence of $T(0)$ beyond this comes only from exact coincidences $\mathbf G_1=R_\theta\mathbf G_2$ (commensurate cells), from interface potentials at mixed wavevectors (moir\'e terms) and from umklapp processes at second order, all beyond \ref{as:cont} (Sec.~\ref{sec:umklapp}).
\end{enumerate}
\noindent\emph{Proof} (App.~\ref{app:coupling}). Eq.~\eqref{eq:H12} follows from $\Psi_1^\dagger(K+V_1)\Psi_2=E_1S_{12}$ and $\Psi_1^\dagger(K+V_2)\Psi_2=S_{12}E_2$. A term periodic in one layer changes the in-plane wavevector by a reciprocal vector of that layer; for incommensurate layers only the components at $\bp$ of both states survive in $S_{12}$ and $K_{12}$. Expanding the $z$ dependence of these components in interface profiles gives Eq.~\eqref{eq:formfactor}; rotating a layer rotates its in-plane components only through $U_s$. $\square$

Twist-dependent $p=0$ couplings of the form $\sum_m|c_m|^2\tau_me^{im\theta}$ with $\tau_{m\neq0}\neq0$ do not occur: registry-averaged couplings between $m_l\neq0$ components vanish identically, because the in-plane integral of an $m_l\neq0$ function is zero. Two consequences are practical. First, the coupling of twisted interfaces is the registry average of the signed aligned coupling, $\bar T(d)=\langle t(\bm\tau,d)\rangle_{\bm\tau}$, and commensurate twisted cells test it with no fitted parameter (tier T2); the sign of $t(\bm\tau)$ must come from bond or phase data, not from amplitudes (Sec.~\ref{sec:validation}). Second, the $J_z$ content of the two faces matters at finite $p$, through the matrix chain (Sec.~\ref{sec:index}).

\subsection{Phonon interfaces: one transform, force balance and stability}\label{sec:stressfree}
For an interface built from pair interactions $\phi$ between atoms of areal densities $\rho_1,\rho_2$, define $U(\bp;d)=\rho_1\rho_2\int\dd^2r\,e^{i\bp\cdot\mathbf r}\phi(\sqrt{r^2+d^2})$. Integration by parts gives the whole continuum tensor from this one function:
\begin{equation}
F_{ij}(\bp)=-p_ip_j\,U,\qquad F_{iz}(\bp)=F_{zi}(\bp)=-ip_i\,\partial_dU,\qquad F_{zz}(\bp)=\partial_d^2U\qquad(i,j\in\{x,y\}).\label{eq:Ftensor}
\end{equation}
We evaluate $U$ and its $d$ derivatives by direct quadrature at every channel momentum and every interface spacing (\texttt{interface7.py}), so the identities hold to machine precision: $F_{ij}(0)=0$ exactly, $F_{TT}=0$, and no interpolation in $p$ or $d$ is used (a linearly interpolated table leaves a residual $F_{xx}(0)\approx-4\times10^{-8}$ and the wrong small-$p$ order in its first segment). Four properties follow.
\begin{enumerate}
\item In-plane coupling is rank~1 and $O(p^2)$: structural superlubricity, $r=2$.
\item $F_{iz}=-ip_i\,\partial_dU(0;d)+O(p^3)$. The coefficient is the normal force per unit density. It vanishes only at the force-free spacing, which for the model Morse continuum is $d_{\rm eq}=3.41569$\,\AA; at 3.40\,\AA\ it does not.
\item At the force-free spacing $F_{zz}(0)-F_{zz}(\bp)=-p^2U(0)/2+O(p^4)>0$ (App.~\ref{app:stress}): the in-phase bending of neighbouring layers costs positive energy, proportional to the adhesion energy, and the $O(p)$ normal--in-plane coupling vanishes. This makes the out-of-plane sector stable at small $p$. It does not make the whole stack stable. The in-plane coupling has no diagonal counterpart ($F_{ij}(0)=0$), and in-phase density waves of neighbouring layers attract: the aligned longitudinal branch is $\omega_L^2=p^2\big[c_L^2-2|u_0|\cos k_z\big]+O(p^4)$ with $u_0=U(0;d)/(AM)$, which is stable only if $c_L^2>2|u_0|$. The production model satisfies this with $c_L^2/2|u_0|=19.7$ (\texttt{inplane\_margin8.py}), and the stability gate checks every chain numerically; the analytic statement is limited to the normal sector.
\item Imposing spacings by hand at a fixed potential leaves each layer with the unbalanced force $\partial_dU(0;d_{l-1})-\partial_dU(0;d_l)$ and an $O(p)$ coupling of in-plane and normal displacements. Chains built this way have negative modes.
\end{enumerate}

\paragraph{Stress-free paracrystal.} In a free-standing stack every interface sits at its own equilibrium. Spacing variations are then variations of the interfaces themselves (registry, intercalants, relaxation). We model interface $l$ as a Morse interface $a_l\phi_{\rm rep}+b_l\phi_{\rm att}$ (repulsive and attractive parts, each shifted-force at the cutoff), with $(a_l,b_l)$ fixed by two conditions: the interface is force free at $d_l$, and its normal stiffness lies on an equilibrium family,
\begin{equation}
\partial_dU_l(0;d_l)=0,\qquad k_l\equiv\partial_d^2U_l(0;d_l)=k_{\rm eq}\,e^{-\beta_{\rm eq}(d_l-d_{\rm eq})}.\label{eq:family}
\end{equation}
The slope $\beta_{\rm eq}$ of the equilibrium family is a material parameter. It is not the load slope $\beta_{\rm load}=-\partial_d\ln\partial_d^2U(0;d)$ of a fixed potential, which describes uniform compression. For the model potential $\beta_{\rm load}=3.045$\,\AA$^{-1}$, while shifting the Morse radius, a rigid displacement of the potential, gives $\beta_{\rm eq}=-0.25$\,\AA$^{-1}$: the stiffness hardly changes. A registry scan measures $\beta_{\rm eq}$ directly, from the pairs $(d_{\rm eq}(\bm\tau),k(\bm\tau))$ (Sec.~\ref{sec:validation}). We run five families, $\beta_{\rm eq}=\beta_{\rm load}$ (L), $3\beta_{\rm load}/4$ (Q$_3$), $\beta_{\rm load}/2$ (H), $\beta_{\rm load}/4$ (Q$_1$) and 0 (R). All amplitudes stay positive over the simulated range of spacings.

\subsection{Three kinds of structural disorder}\label{sec:kinds}
\begin{enumerate}
\item \textbf{Rotation is disorder that vanishes on the axis.} In the lab gauge the on-site block and the face amplitudes are twist independent at $\bp=0$ (for $n\ge2$) and twist dependent at finite $p$. Part of that dependence is a rotation of the multiplet frame, which a chain can undo; the rest scatters, with the index $\kappa$ of Sec.~\ref{sec:index}.
\item \textbf{Spacing is bond disorder}, with sensitivity $\beta_e=-\partial_d\ln|\bar T|$ for electrons and $\beta_{\rm eq}$ for phonons. For phonons it enters the bond and the sum-rule diagonal together, which protects $\bp\to0$, $\omega\to0$.
\item \textbf{Registry is a gauge for twisted interfaces and bond disorder for aligned ones.} Within \ref{as:cont} registry phases are removed by a gauge transformation. An aligned interface keeps all Fourier components, $t(\bm\tau)=\sum_{\mathbf G}c_{\mathbf G}e^{i\mathbf G\cdot\bm\tau}$ with $c_0=\bar T$; stacking faults of aligned stacks are bond disorder with $\sigma_{\rm reg}^2={\rm Var}_{\bm\tau}|t|/\langle|t|\rangle^2$. For phonons the aligned crystal also has a finite shear coupling that twisted interfaces lose entirely (Sec.~\ref{sec:registry}).
\end{enumerate}

\subsection{Notation}\label{sec:notation}
The rotational disorder of channel $\bp$ has variance $\sigma(p)^2={\rm Var}_\theta X_\theta(\bp)$ (Frobenius-normalized for blocks); $\sigma=v_\sigma p^\kappa$ and $\Veff=\sqrt2\,\sigma=v_Vp^\kappa$. The Lyapunov exponent $\gamma$ is defined by $\langle\ln\mathcal T\rangle=-2N\gamma$. The single-chain function $\Phi_1(s)$ is the disorder-averaged transmission of a weakly disordered chain at one energy as a function of $s=2N\gamma$; it is computed once and has no model parameter. For channels with closed (evanescent) components, $\gamma$ always means the exponent of the effective chain after those components are eliminated.

%% file: sec_index.tex
\section{The index from the matrix chain}\label{sec:index}

The thickness law of a band-edge multiplet is fixed by two orders: the order $r$ at which the interface coupling vanishes on the axis, and the order $\kappa$ at which the stack backscatters, $\gamma\propto p^{2(\kappa-r)}$. Neither is a single-bond quantity. An interface couples the top face of one layer to the bottom face of the next, and single-bond singular values do not see how the directions (polarizations) of a matrix chain rotate from layer to layer: an explicit centrosymmetric C$_3$ Kramers pair, whose single-bond spreads vary at order $p^5$, has $\kappa=3$ (\texttt{c3\_face\_check.py}; Sec.~\ref{sec:matrixval}). We therefore define $\kappa$ by the backscattering of the full matrix chain.

\subsection{Face-resolved chain and definition of the index}\label{sec:facechain}
A layer carries a band-edge multiplet of $d$ states with $J_z$ content $\{J\}$; its top and bottom faces have amplitudes $G^{\rm t}(\mathbf q)$, $G^{\rm b}(\mathbf q)$ on the face channels (spin $\times$ interface profile; Proposition~\ref{prop:coupling}), and its on-site block is $X(\mathbf q)$. Symmetry fixes their form: C$_n$ covariance $G(R\mathbf q)=U_s(t)G(\mathbf q)D(t)^\dagger$, $X(R\mathbf q)=DXD^\dagger$ with $D={\rm diag}(e^{-iJt})$; time reversal; and the relation between the faces,
\begin{center}\small
\begin{tabular}{ll}
centrosymmetric layer (inversion) & $G^{\rm b}(\mathbf q)=G^{\rm t}(-\mathbf q)$, $X(\mathbf q)=X(-\mathbf q)$\\
horizontal mirror $\sigma_h$ & $G^{\rm b}(\mathbf q)=U_s(\pi)G^{\rm t}(\mathbf q)D(\pi)^\dagger$, $[X,D(\pi)]=0$\\
polar layer & $G^{\rm b}$ independent of $G^{\rm t}$\\
identical faces (test case) & $G^{\rm b}=G^{\rm t}$
\end{tabular}
\end{center}
In the lab gauge, $G^{\rm lab}_l(\bp)=U_s(\theta_l)G(R_{\theta_l}^{-1}\bp)D(\theta_l)^\dagger$ and $X^{\rm lab}_l=D(\theta_l)X(R_{\theta_l}^{-1}\bp)D(\theta_l)^\dagger$, the chain of channel $\bp$ has bonds $T_l=t\,G^{\rm t\,lab\dagger}_lG^{\rm b\,lab}_{l+1}$ and on-site blocks $X^{\rm lab}_l$. Let $(\bar T,\bar X)$ be the twist averages and $(\delta T_l,\delta X_l)$ the fluctuations. For an open forward channel $a$ (wavenumber $k_a$, velocity $v_a$, polarization $u_a$) and backward channel $b$ of the averaged chain at energy $E$, the Born backscattering amplitude per layer is
\begin{equation}
m_l=u_b^\dagger\big[\delta X_l+e^{ik_a}\delta T_l+e^{-ik_b}\delta T_l^\dagger\big]u_a,\qquad
\gamma_B=\frac12\Big\langle\sum_b\frac{C_0+2\,{\rm Re}\,C_1e^{i(k_a-k_b)}}{|v_av_b|}\Big\rangle_a,\label{eq:matrixborn}
\end{equation}
with $C_j=\langle m_l^*m_{l+j}\rangle$ ($m_l$ depends on $\theta_l,\theta_{l+1}$ only, so $C_{j\ge2}=0$ for i.i.d.\ twists). We define $r$ by $|\bar T|\propto p^r$ and $\kappa$ by $\gamma_B\propto p^{2(\kappa-r)}$. Degenerate channels are treated in a basis that diagonalizes the velocity in the degenerate subspace. Because the Born expansion is around the averaged chain, it is accurate when the twist-dependent part of the blocks is small; it fails for an $O(1)$ frame rotation at $p=0$, which occurs only for C$_1$ Kramers pairs (below).

\paragraph{Why single bonds are not enough.} For square invertible faces the substitution $\phi_l=G^{\rm b}_l\psi_l$ gives the face chain of Eq.~\eqref{eq:facechain}: on-site blocks $Y_l=(G^{\rm b}_l)^{-\dagger}(E-X_l)(G^{\rm b}_l)^{-1}$ and bonds $tK_l^\dagger$, with $K_l=G^{\rm t}_l(G^{\rm b}_l)^{-1}$ the transfer from the bottom to the top face of a layer. Both are lab-frame matrices on the face channels; the multiplet basis has dropped out. For identical faces this is the inverse Gram matrix, $Y_l=(E-\varepsilon)(G_lG_l^\dagger)^{-1}$. The disorder that scatters is the twist dependence of $K_l$ and $Y_l$ that is not a rotation of the frame; singular values of $T_l$ miss the rotation of the eigenbasis of $Y_l$ relative to $K_l$. In the explicit C$_3$ pair $G_lG_l^\dagger$ has a random $p^3\sigma_z$ part, while the bond singular values of the same-face chain vary only at $p^5$.

\subsection{Removable and non-removable twist dependence}\label{sec:parity}
Three facts, proved in App.~\ref{app:matrix}, decide which twist dependence scatters at Born order.
\begin{enumerate}
\item \emph{Harmonic selection.} In the lab gauge the entry $(s_z,J_z)$ of a face amplitude carries twist phases $e^{i(J_z-s_z-m)\theta}$ from its in-plane harmonics $m\equiv J_z-s_z\pmod n$, $O(p^{|m|})$. For a $J_z=\pm\frac12$ Kramers pair ($n\ge2$) the twist-dependent part starts at $|m|=n$ in spin-conserving entries and at $|m|=n-1$ in spin-flip entries; the twist-independent spin-flip harmonic $m=\pm1$ makes the averaged chain non-scalar at order $p$.
\item \emph{Time-reversal parity at the bond.} The leading twist dependence of a Kramers bond is its spin-flip harmonic, of in-plane order $m=n-1$. For any number of interface profiles and any real profile kernel it has the form $\delta T_l=-\Omega_{{\rm t},l}\bar T+\bar T\Omega_{{\rm b},l+1}$, with $\Omega_{{\rm t},l}$ set by the twist of layer $l$ and $\Omega_{{\rm b},l+1}$ by that of layer $l+1$. Time reversal makes $\Omega$ anti-Hermitian for even $m$ (a rotation of the multiplet frame) and Hermitian for odd $m$ (a change of the coupling strengths) (Lemmas~\ref{lem:faces} and \ref{lem:bond}). Spin-conserving harmonics have a Hermitian part at either parity. For spinless states even harmonics change the modulus of the scalar bond and odd ones its phase.
\item \emph{Frame rotations do not backscatter at Born order.} When the averaged chain is a scalar times a unitary, the first-order Born amplitude of such a rotation vanishes identically between every forward and backward channel (Lemma~\ref{lem:frame}); the rotation backscatters only through the non-scalar $O(p^a)$ part of the averaged chain, at order $p^{m+a}$. A rotation common to both faces of a layer is also removed exactly, at all orders, by a local gauge transformation. When the faces rotate differently, the gauge that removes the rotations accumulates along the chain; we make no all-order statement for that case and compare the Born orders with exact exponents (Sec.~\ref{sec:matrixval}).
\end{enumerate}
\begin{proposition}[protected band edges, Born order]\label{prop:matrix}
For an A-type spinless state or a $J_z=\pm\frac12$ Kramers pair of a C$_n$ layer with $n\ge2$ ($r=0$), with $p_z$/$s$-type face channels on any number of interface profiles, any face relation and generic coefficients, the Born backscattering rate has the index
\begin{equation}
\kappa_{\rm A}=\begin{cases}n & n\ \text{even}\\ 2n & n\ \text{odd}\end{cases},\qquad
\kappa_{\rm K}=\begin{cases}n-1 & n\ \text{even}\\ n & n\ \text{odd}\end{cases}.\label{eq:kappamatrix}
\end{equation}
For C$_1$ layers the same construction gives $\kappa_{\rm A}=2$ and $\kappa_{\rm K}=1$ when all interface profiles of a face share one spinor frame at $p=0$, in particular for a single profile; with several profiles a C$_1$ Kramers pair is in general not protected (Sec.~\ref{sec:multiprofile}).
\end{proposition}
\noindent\emph{Conditions.} The Proposition is a statement about the Born order of the chain. It assumes (i)~a nonzero averaged bond on the axis, $\bar T(0)\neq0$ ($r=0$; the local frame lemma divides by its scale $t_0$); (ii)~an open energy that is not exceptional: away from the edges of the averaged chain's bands, where the velocity denominators of the rate vanish, from crossings of its channels, from energies where the number of open channels changes, and from special cancellations such as chiral or band-centre anomalies of off-diagonal disorder; (iii)~generic coefficients. At exceptional energies the power counting of the rate can change; these are conditions of the statement, not counterexamples. Exact transfer-matrix exponents (Sec.~\ref{sec:matrixval}) follow the Born rates over the accessible, partly pre-asymptotic windows; they support the index beyond Born order but do not prove it for every distribution, energy and channel.

The Kramers values follow from the two parity cases: for even $n$ the leading spin-flip harmonic ($|m|=n-1$, odd) is Hermitian and scatters at its own order; for odd $n$ it is a frame rotation, and scattering starts at order $n$, from the spin-conserving harmonic and from the rotation acting on the $O(p)$ spin-flip part of the averaged chain. The explicit C$_3$ pair is the case $n=3$: its $p^2$ term is exactly the anti-Hermitian frame rotation, and its $p^3[\cos3\theta_{l+1}-\cos3\theta_l]\sigma_z$ is the Hermitian remainder. The proof (App.~\ref{app:matrix}) works on the full bond contraction $G_{\rm t}^\dagger\tau G_{\rm b}$ and needs neither square nor invertible face amplitudes: for $n\ge2$ every profile's $p=0$ amplitude is diagonal in $s_z=J_z$, time reversal relates the two spin-flip entries of every profile in the same way, and a profile sum with a real kernel preserves both relations. Random models with two and three profiles and generic positive kernels confirm the bond structure: the part of $\Omega_{\rm t}$ of the wrong parity vanishes linearly in $p$ (relative size $1$--$4\times10^{-3}$ at $p=0.003$ for $n=2$, 3, 4, 6), and chains whose face fluctuations are restricted to the part outside the range of the averaged amplitude, which carries 8--250\% of the bond fluctuation and which a left inverse of that amplitude does not see, still give Eq.~\eqref{eq:kappamatrix} (\texttt{models/\allowbreak checks10.py}).

Three consequences. (i)~The Kramers indices are 1 (single profile), 1, 3, 3, 5 for C$_1$, C$_2$, C$_3$, C$_4$, C$_6$ \emph{whether or not the layer is centrosymmetric}: the face relation does not enter, and the hidden spin texture of centrosymmetric layers does not lower the index. (ii)~The relevant material object is not a projected spin splitting, which vanishes for a PT-degenerate pair~\cite{ZhangZunger2014,LiAppelbaum2018}, but the gauge-invariant face-projected density matrix $\Pi(\mathbf q)=\Phi\Phi^\dagger$, $\Phi=[G^{\rm t};G^{\rm b}]$, summed over the degenerate pair, together with $\Phi X\Phi^\dagger$; the kit's matrix tool checks their C$_n$ covariance and builds the chain from $\Phi$ and $X$ (Sec.~\ref{sec:validation}). (iii)~Parity, not the size of the admixture, decides whether a harmonic scatters: in C$_3$ layers a $p^2$ spin-flip admixture of any size is invisible to transport at order $p^2$.

\paragraph{C$_1$ Kramers pairs.} The spinor of the in-plane average is not tied to $z$, so each twisted copy is rotated by an $O(1)$ frame rotation at $p=0$ (Proposition~\ref{prop:coupling}, item 2). The rotation common to both faces can be removed layer by layer (the kit aligns each copy to the twist average), but a relative rotation of the two faces' frames is a gauge that first-order Born theory around the averaged chain turns into a spurious $p$-independent floor. For these pairs the index is taken from transfer matrices: with one interface profile $\gamma\propto p^{2}$, $\kappa=1$ (Sec.~\ref{sec:matrixval}).

\subsection{Several interface profiles}\label{sec:multiprofile}
A face couples through a few interface profiles (Proposition~\ref{prop:coupling}). For C$_1$ Kramers pairs this matters. Time reversal makes the $p=0$ amplitude of each profile quaternion real, i.e.\ a scalar times an SU(2) frame, but different profiles may have different frames, and a sum of different rotations is not a rotation. Example: two profiles with $G_1=\sqrt{0.4}\,\mathbb 1$ and $G_2=\sqrt{0.1}\,i\sigma_x$ on both faces and a unit profile kernel give, in a suitable basis, the $p=0$ bond $T_l=0.4\,\mathbb 1+0.1\,e^{i(\theta_{l+1}-\theta_l)\sigma_z}$, whose singular values $\sqrt{0.17+0.08\cos\Delta\theta}$ vary between 0.3 and 0.5 with the relative twist; no unitary change of basis removes this. The chain is disordered on the axis: with zero on-site energy its smallest Lyapunov exponent at $E=0.25$ is $0.0178\pm0.0001$ ($10^5$ layers, 12 chains; \texttt{models/\allowbreak c1\_profiles.py}). For $n\ge2$ the argument after Proposition~\ref{prop:matrix} applies instead.

We tested both statements with random multiplets carrying two interface profiles of in-plane symmetry $l=0$ per spin and a random positive profile kernel, otherwise built as in Sec.~\ref{sec:matrixtable} (\texttt{multiprofile8.py}):
\begin{itemize}
\item $n\ge2$: Kramers $\pm\frac12$ pairs of C$_2$, C$_3$, C$_4$, C$_6$ layers and spinless A states of C$_2$, C$_3$ layers, inversion-related and polar faces, three seeds (36 models). At the common-window energy all 36 give Eq.~\eqref{eq:kappamatrix} within 0.04 (median deviation 0.0004). Of the 60 fits at an energy inside one band, 54 agree within 0.15; four are larger (1.90, 3.53, 3.73, 5.44, all polar models of one seed, pre-asymptotic), and in two polar models of another seed the edge of the second band crosses the fit energy inside the fitted range of $p$, which changes the set of open channels and invalidates the fit.
\item C$_1$ Kramers pairs, exact exponents ($10^5$ layers, 8 chains): with one profile the exponent vanishes on the axis to within the estimator floor ($1.9\times10^{-6}$ at $p=0$) and grows as $p^{1.97}$ ($\kappa=1$); with two profiles it stays finite on the axis, $2.4\times10^{-4}$ for inversion-related faces and $4.5\times10^{-2}$ for polar faces, where the singular values of the $p=0$ bond vary by 5.6\% and 72.5\% over the relative twist.
\end{itemize}
The C$_1$ entries of Proposition~\ref{prop:matrix} and Table~\ref{tab:kappa} therefore hold only for faces whose profiles share one spinor frame; the entries for $n\ge2$ hold with any number of profiles.

\subsection{The remaining classes}\label{sec:matrixtable}
Classes with $r>0$ and spinless doublets are not covered by Proposition~\ref{prop:matrix}. Table~\ref{tab:kappa} gives their generic orders from random C$_n$-, time-reversal- and face-symmetric models with one interface profile (\texttt{matrix8.py}, \texttt{matrix8\_table.py}): top and bottom faces, an interior block that keeps the multiplet normalized, L\"owdin normalization at every $\mathbf q$, a random on-site block, five seeds, $p=0.004$--0.04 (in units of the model's inverse length), the common-window energy and one energy inside each band. The $r=0$ rows reproduce Eq.~\eqref{eq:kappamatrix} for every face relation (median over seeds and energies within 0.01; \texttt{index\_rule8.py}, which also checks the parity lemma harmonic by harmonic for C$_2$--C$_6$).

\begin{table}[!tbp]\centering\small
\caption{Matrix-chain index $\kappa$ (coupling and on-site disorder together) and coupling order $r$, for $p_z$/$s$-type face channels. Entries: centrosymmetric / horizontal mirror / polar layer when they differ. ``np'': not protected ($\kappa\le r$; $\gamma$ does not vanish on the axis). Rows with $r=0$ are Proposition~\ref{prop:matrix}; the others are random-model orders, the smallest fitted value over five seeds and all energies (the generic order; full ranges in \texttt{matrix8\_table.log}, where larger values at some seeds or energies reach, e.g., 6.6 for Kramers $\pm\frac52$ in centrosymmetric C$_6$ layers).}\label{tab:kappa}
\begin{tabular}{L{4.3cm}ccccc}
\toprule
multiplet (class), $r$ & C$_1$ & C$_2$ & C$_3$ & C$_4$ & C$_6$\\
\midrule
spinless A ($J_z\equiv0$), $r=0$ & 2 & 2 & 6 & 4 & 6\\
spinless B ($J_z\equiv n/2$), $r=n$ & -- & np & -- & np & np\\
Kramers $J_z\equiv\pm\frac12$, $r=0$ & 1$^a$ & 1 & 3 & 3 & 5\\
Kramers $\pm\frac32$, $r=2$ & -- & -- & np & 3 / 3 / np$^b$ & 5 / 5 / 3\\
Kramers $\pm\frac52$, $r=4$ & -- & -- & -- & -- & 5 / 5 / np$^b$\\
spinless doublet $L_z=\pm1$, $r=2^c$ & -- & -- & 4 & np & 4\\
spinless doublet $L_z=\pm2$, $r=4^c$ & -- & -- & -- & -- & np\\
flexural phonon, $r=0$ & 4 & 4 & 6 & 4 & 6\\
in-plane acoustic phonon, $r=2$ & np & np & 4 [3$^d$] & np & 4\\
\bottomrule
\end{tabular}
\\[2pt]\raggedright\footnotesize $^a$Transfer matrices (first-order Born is invalid, see text); for one interface profile or profiles sharing one spinor frame (Sec.~\ref{sec:multiprofile}). $^b$On-site Kramers mixing at $O(p)$ in polar layers ($\kappa=1<r$). $^c$With $p_z$-type faces only one combination of the doublet couples at order $r$ ($\bar T$ has rank 1); the index refers to that branch away from the energy of the other, near which scattering is resonantly enhanced. A polar C$_3$ doublet has $\kappa_{\rm on}=3$, but its $p^3\sigma_z$ term acts only between the two branches, so the active branch has 4. $^d$Gyrotropic L--T term of non-centrosymmetric C$_3$ layers, only with an open T channel (Sec.~\ref{sec:phonon}). Phonons: interlayer forces act on lab-frame displacements, so the on-site rule gives the index.
\end{table}

Readings. (i)~The C$_3$ Kramers $\pm\frac32$ class is not protected: its spin-down component has a twist-dependent harmonic of order $p$ against a mean of order $p^2$, so the bond fluctuates at its own order ($\kappa=r=2$). (ii)~Spinless doublets with $p_z$-type faces have a rank-one averaged coupling: the doublet splits into an active and a dark branch, and the active branch has $\kappa=4$ in C$_3$ and C$_6$ layers with or without inversion. (iii)~Polar layers lose protection of $J_z=\pm\frac32$ (C$_4$) and $\pm\frac52$ (C$_6$) pairs through an $O(p)$ on-site Kramers mixing, as the on-site rule (Sec.~\ref{sec:onsite}) predicts.

\subsection{Worked atomistic examples}\label{sec:tbexamples}
To test the rule on material-compatible layers we built Slater--Koster monolayers with spin--orbit coupling (\texttt{tb8.py}): a trigonal-prismatic MX$_2$ layer (D$_{3h}$, metal $d$ and chalcogen $p$ orbitals, MoS$_2$-type parameters and spin--orbit couplings of 0.075 and 0.052\,eV), the same layer in octahedral (centrosymmetric D$_{3d}$) coordination, and a phosphorene-type puckered layer (D$_{2h}$) as a $p_z$ model with Rudenko--Katsnelson-type hoppings and as an $sp^3$ model with spin--orbit coupling. For each, the $\Gamma$-derived multiplet is projected at every $\mathbf q$ of a polar grid on the outward $p_z$ lobes of the outer atoms (the forward, $\mathbf G=0$, component: the sum of the $p_z$ coefficients of a face), and the kit's matrix tool computes $r$ and $\kappa$ (\texttt{vk\_matrix\_index.py}). In these models the basis orbitals are atomic functions of fixed sign and shape, so the forward amplitude of a face is the sum of its outward-lobe coefficients. For first-principles input this is not so: a Wannier Hamiltonian does not fix the sign or shape of its orbitals, and a constant re-phasing of the basis leaves every band unchanged while swapping A- and B-type assignments of coefficient sums. The kit therefore takes the forward amplitude either from physical forward weights of the Wannier functions (\texttt{vk\_face\_from\_tb.py}; without them its output is labelled conditional) or directly from plane-wave wavefunctions as the planar Fourier component on a plane in the gap (\texttt{vk\_face\_from\_pw.py}, Sec.~\ref{sec:validation}), which involves no basis at all. Results (\texttt{tb8.log}):
\begin{center}\small
\begin{tabular}{llccl}
\toprule
model & multiplet at $\Gamma$ & $r$ & $\kappa$ & note\\
\midrule
MX$_2$, D$_{3h}$ (MoS$_2$-type) & $d_{z^2}$ Kramers pair & 0.0 & 2.94$^a$ & Prop.~\ref{prop:matrix}: 3\\
MX$_2$, D$_{3d}$ (centrosymmetric) & $d_{z^2}$ Kramers pair & 0.0 & 3.00 & Prop.~\ref{prop:matrix}: 3\\
same, without spin--orbit coupling & $d_{z^2}$ (A-type) & $-0.18$ / 0.00 & 5.7$^b$ / 5.97 & spinless: 6\\
phosphorene $p_z$ model & conduction edge (A) & 0.0 & 2.00 & protected\\
phosphorene $p_z$ model & valence edge (B) & 2.0 & 2.00 & not protected\\
phosphorene $sp^3$, with SOC & $p_z$-like A pair$^c$ & 0.0 & 1.1--2.0$^c$ & 2 at 0.02--0.1\,\AA$^{-1}$\\
phosphorene $sp^3$, with SOC & valence-edge Kramers pair & --$^d$ & --$^d$ & flagged: crossover\\
\bottomrule
\end{tabular}
\\[1pt]{\footnotesize $^a$Fit over $p=0.004$--0.04\,\AA$^{-1}$; the local index is 2.99 at 0.002\,\AA$^{-1}$ and falls to 2.6 at 0.1\,\AA$^{-1}$. $^b$D$_{3h}$ / D$_{3d}$; pre-asymptotic at $p\le0.1$\,\AA$^{-1}$ for D$_{3h}$ (local coupling order $-0.07$ to $-0.32$; the kit reports it as conditional). $^c$The analogue of the phosphorene conduction edge, which is $p_z$-derived~\cite{RudenkoKatsnelson2014}. With these generic $sp^3$ parameters a $p_x$-like pair lies 0.33\,eV below it at $\Gamma$; that pair has a spin--orbit-induced forward amplitude of 0.03 and is not used. Local index of the $p_z$-like pair 1.13, 1.26, 1.47, 1.69, 1.84, 1.93, 1.96, 1.96, 1.91 at $p=0.003$, 0.004, 0.006, 0.010, 0.016, 0.025, 0.040, 0.063, 0.10\,\AA$^{-1}$: toward the Kramers value 1 of Proposition~\ref{prop:matrix} as $p\to0$ (it crosses 1.5 near $p^*\approx0.007$\,\AA$^{-1}$), near the spinless value 2 between 0.02 and 0.1\,\AA$^{-1}$. $^d$A forward amplitude of 0.022 from spin--orbit mixing against a $p^2$ term: local orders drift ($r$ from $-0.07$ to $-0.73$ over 0.004--0.016\,\AA$^{-1}$); the kit prints the fitted numbers but labels them conditional, with the flags ``coupling order not a power law'' and ``index drifts''.}
\end{center}
Spin--orbit coupling thus sets the index of the MX$_2$ $\Gamma$ pair (3 instead of 6), through the face amplitudes as much as through the on-site term. In phosphorene, with in-plane $pp\pi$ hopping $t_1<0$ and a dominant inter-sublayer hopping $|t_2|>2|t_1|$, the four $\Gamma$ states are $E=\pm(2t_1\pm t_2)$ in the two-hopping approximation (further shells shift them to $-6.51$, $-1.65$, $+0.39$ and $+6.13$\,eV in the model without changing their symmetry); the valence edge is the combination in which the two atoms of a face have opposite $p_z$ amplitudes, odd under the C$_2$ rotation about the stacking axis. Its forward face amplitude vanishes, so a twisted interface does not couple $\Gamma$ holes at $p=0$ ($r=2$), and its twist disorder is of the same order ($\kappa=r$). The conduction edge is even and protected ($r=0$, $\kappa=2$). In the two-hopping approximation the ordering, and hence the conclusion, depends only on $t_1<0$ and $|t_2|>2|t_1|$; the parities agree with first-principles band representations of phosphorene~\cite{FariaJunior2019} and with our plane-wave face data, in which the valence edge's forward amplitude is below $10^{-10}$ and the conduction edge's is 0.080\,\AA$^{-3/2}$ (Sec.~\ref{sec:dft}).

Spin--orbit coupling mixes an $r=0$ component into the hole pair, of relative size $\sim\lambda_{\rm SO}/\Delta$. In the $sp^3$ model its forward face amplitude at $\Gamma$ is 0.022, against 1.225 for the $p_z$-like A-type pair, and the twist-averaged hole coupling $|\bar T|$ is $3\times10^{-4}$ of the electrons' for $p\lesssim0.01$\,\AA$^{-1}$ (\texttt{tb8.log}). The $p^2$ term of the B-type coupling reaches the admixed constant near $p_h\approx0.02$\,\AA$^{-1}$, where the two interfere ($|\bar T|$ has a minimum near 0.024\,\AA$^{-1}$); beyond about 0.04\,\AA$^{-1}$ the coupling grows as $p^2$, to $3\times10^{-3}$ of the electrons' at 0.08\,\AA$^{-1}$. Well inside the window the hole pair is a C$_2$ Kramers pair with a protected but $3\times10^{-4}$ times weaker coupling (local index 1.0--1.7 at 0.003--0.010\,\AA$^{-1}$); beyond $p_h$ its coupling is B-type and fluctuates at its own order, so these holes are localized within a few layers. Umklapp processes (Sec.~\ref{sec:umklapp}) add a random $p=0$ coupling of the holes that need not be small compared with this weak forward coupling. They can change the open channels, the effective bands and the interference that sets the window, and we do not know the sign of their net effect on the contrast; the $\Gamma$ splittings of the twisted bilayer (Sec.~\ref{sec:validation}) measure the combined coupling of the holes. In the models, random-twist phosphorene stacks therefore contrast electrons and holes: $\Gamma$-valley electrons cross with the axis law, while holes have no protected axis and decay algebraically only through a transparent energy (above a threshold coupling, Sec.~\ref{sec:thicknessdft}) or in the window $p\lesssim p_h$, where their coupling is more than three orders of magnitude weaker than the electrons'. The contrast is a statement about finite thicknesses and thermal carriers, not an asymptotic blocking: as $N\to\infty$ the holes of the window follow the Kramers law $N^{-1}$ with a small prefactor.

Spin--orbit coupling also makes the conduction edge a C$_2$ Kramers pair, whose index is 1 (Proposition~\ref{prop:matrix}); in the $sp^3$ model the local index falls toward it only below $p^*\approx0.007$\,\AA$^{-1}$ and is 1.84--1.96 between 0.016 and 0.1\,\AA$^{-1}$ (footnote $c$), which covers the thermal range (0.03--0.09\,\AA$^{-1}$ for the armchair and zigzag masses at 300\,K). The momenta that still transmit shrink as $p_N\propto N^{-1/4}$ under the spinless law, so the spinless law $N^{-1/2}$ holds until $p_N$ reaches $p^*$, a thickness $(p_{\rm th}/p^*)^4\approx4\times10^2$--$3\times10^4$ times the onset of the axis law in the model ($10^5$--$10^8$ with the first-principles $p^*$ of Sec.~\ref{sec:dft}); only beyond it does the Kramers law $N^{-1}$ take over. The magnitudes used here ($\lambda_{\rm SO}$, $p^*$, $p_h$, the hole coupling) come from generic tight-binding parameters with a unit coupling kernel, and they depend on that kernel: scaling it by 2 and 4 moves $p^*$ from 0.0068 to 0.0119 and 0.0194\,\AA$^{-1}$, while $r$ and the asymptotic $\kappa$ do not change (\texttt{models/\allowbreak checks10.py}). A monolayer plane-wave calculation identifies the parities and the face amplitudes; $p^*$, $p_h$ and the hole/electron coupling ratio need in addition the interface kernel in energy units (Sec.~\ref{sec:validation}), and an umklapp thickness needs umklapp matrix elements and partner energies (Sec.~\ref{sec:umklapp}).

\subsection{Validation}\label{sec:matrixval}
Four checks test the Born index against exact results (\texttt{validate8.py}, \texttt{validate8\_extra.py}, \texttt{validate8\_bc.py}, \texttt{validate8\_e.py}, \texttt{validate8\_parse.py}; Fig.~\ref{fig:matrixval}).

\emph{(A) Born rate against exact Lyapunov exponents.} For eleven representative classes and face relations we computed the Lyapunov spectrum of the same lab-gauge chain at the energy of the Born calculation from transfer matrices (QR re-orthogonalization, $N=2\times10^5$ layers, 12 chains), at three momenta where the Born rate lies between $3\times10^{-5}$ and $3\times10^{-3}$ (for the Kramers $\pm\frac32$ class, whose rate grows only as $p^2$, between $5\times10^{-5}$ and $2\times10^{-3}$ at smaller $p$; \texttt{validate8\_extra.py}). Below this window the finite-$N$ estimator floor dominates (measured on the twist-averaged chain, whose true exponent is zero: $\le1.1\times10^{-6}$); above it Born theory is not expected to hold.
\begin{center}\small
\begin{tabular}{llccccccc}
\toprule
class & faces & $p$ window & \multicolumn{2}{c}{local exponent} & $\gamma_{\min}/\gamma_B$ & $\gamma_{\max}/\gamma_B$ & $2(\kappa-r)$\\
 & & & Born & exact & & & asymptotic\\
\midrule
spinless A, C$_2$ & polar & 0.068--0.21 & 3.92 & 3.86 & 0.99--1.10 & -- & 4\\
spinless A, C$_3$ & inversion & 0.45--0.6 & 12.81 & 12.14 & 0.75--0.92 & -- & 12\\
spinless A, C$_4$ & polar & 0.26--0.41 & 7.98 & 7.67 & 0.99--1.20 & -- & 8\\
Kramers $\pm\frac12$, C$_2$ & inversion & 0.02--0.1 & 2.04 & 2.07 & 0.90--0.97 & 1.09--1.14 & 2\\
Kramers $\pm\frac12$, C$_2$ & polar & 0.02--0.052 & 1.99 & 2.06 & 0.46--0.51 & 1.51--1.52 & 2\\
Kramers $\pm\frac12$, C$_3$ & inversion & 0.19--0.34 & 6.74 & 6.65 & 0.86--0.90 & 0.99--1.04 & 6\\
Kramers $\pm\frac12$, C$_3$ & mirror $\sigma_h$ & 0.31--0.55 & 7.02 & 6.83 & 0.58--0.68 & 1.32--1.38 & 6\\
Kramers $\pm\frac12$, C$_3$ & polar & 0.16--0.31 & 5.93 & 5.68 & 0.46--0.59 & 1.45--1.46 & 6\\
Kramers $\pm\frac12$, C$_4$ & inversion & 0.28--0.45 & 8.17 & 8.23 & 0.81--0.85 & 1.14--1.22 & 6\\
Kramers $\pm\frac12$, C$_6$ & polar & 0.31--0.41 & 11.14 & 11.87 & 0.55--0.67 & 1.27--1.36 & 10\\
Kramers $\pm\frac32$, C$_4$ & inversion & 0.0032--0.019 & 2.00 & 2.11 & 0.81--0.98 & 1.03--1.15 & 2\\
\bottomrule
\end{tabular}
\end{center}
Readings. (i)~Over every window the local exponents of the Born rate and of the smallest exact exponent agree within 0.73 (median 0.11), so the Born rate gives the order. Several windows lie at momenta where the local exponent still exceeds its asymptotic value (last column); the asymptotic orders are those of the small-$p$ Born fits of Table~\ref{tab:kappa}. (ii)~The ratio of the smallest exact exponent to the channel-averaged Born rate is an $O(1)$, class-dependent factor: 0.75--1.20 for one open channel (spinless A), 0.81--0.98 for two open channels with inversion-related faces, and 0.46--0.68 for two open channels with other face relations (largest exponent 1.27--1.52 times the Born rate). With two open channels the two exponents split, by up to a factor 1.5 when the faces are related by inversion and by a factor 2--3 otherwise; the channel-averaged Born rate lies between them. The factor enters amplitudes, not exponents, which is why amplitudes of multichannel classes need the multichannel kernel (Sec.~\ref{sec:theorem}).

\emph{(B) C$_1$ Kramers pairs} with one interface profile (transfer matrices only, Sec.~\ref{sec:index}): for inversion-related, polar and identical faces, two seeds each, the smallest exponent grows as $p^{1.96}$--$p^{2.11}$ between $p=0.05$ and 0.2, i.e.\ $\kappa=0.98$--1.06, against 1 from Proposition~\ref{prop:matrix}. With two profiles the axis is not protected (Sec.~\ref{sec:multiprofile}).

\emph{(C) The explicit C$_3$ doublet.} For opposite faces the Born rate and the exact exponent agree with its closed form $p^6/(1-p^2)^2$: Born $+2\%$, $+4\%$, $+7\%$ and exact $-0.2\%$, $-4\%$, $-11\%$ at $p=0.25$, 0.35, 0.45. For identical faces the twist-averaged band narrows with $p$ (to half its width at $p=0.25$) and the fixed-fraction energy approaches its edges. The Born local exponent there is 5.6--5.9 for $p\le0.25$ and reaches 6 only as $p\to0$ ($\gamma_B/p^6$ falls from 0.37 at $p=0.02$ to 0.18 at $p=0.25$). At $p=0.25$ the two exact exponents are 0.35 and 1.03 times the Born rate, the same two-channel pattern as in (A); at $p=0.35$ and 0.45 the Born rate exceeds both exact exponents by up to a factor of 10, because the energy is then close to a band edge of the narrowed band. The index is 3 for both face relations; only the identical-face amplitude at $p\gtrsim0.3$ is outside first-order Born theory.

\emph{(D) The polar C$_3$ doublet} (spinless, $L_z=\pm1$, $r=2$; footnote $c$ of Table~\ref{tab:kappa}). Its averaged coupling has rank one, so the bond is singular and the doublet chain has no transfer matrix. We eliminated the doublet exactly at every layer: with the face scalars $t_l=g^{\rm t}_l\cdot\psi_l$, $b_l=g^{\rm b}_l\cdot\psi_l$ the chain becomes a one-channel recursion for $(b_l,t_{l-1})$ with $2\times2$ transfer matrices built from $g(E-X_l)^{-1}g^\dagger$ (\texttt{validate8\_e.py}). At the energy of the first band the exact exponent equals the Born rate within 4\% at $p=0.34$ and 0.53 but is only 0.41 of it at $p=0.24$ ($1.34\pm0.27\times10^{-5}$ against $3.2\times10^{-5}$). At the energy of the second band, which lies near the level of the uncoupled combination, the exact exponent exceeds the Born rate by a factor 3.3--4.9: the eliminated on-site term is resonant there and first-order Born theory does not apply. The local exponents over these windows (Born 5.5 and exact 6.6 in band 1, 5.2 and 5.6 in band 2) differ by more than in check (A) and are far from the asymptotic 4. For this class the table value is a Born order that the exact exponents neither confirm nor exclude at accessible rates; it is not used for any material statement.

\begin{figure}[!tbp]\centering
\includegraphics[width=\textwidth]{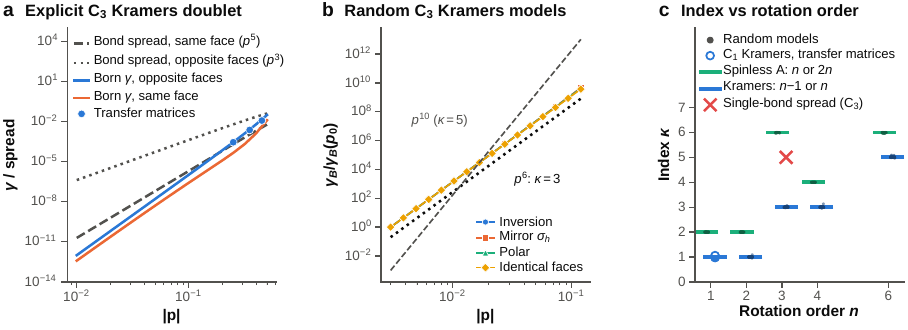}
\caption{The index from the matrix chain (letter Fig.~1). (a)~Explicit centrosymmetric C$_3$ Kramers doublet: twist spread of single-bond singular values for the same face ($p^5$) and for opposite faces ($p^3$); Born backscattering rate of the chain for opposite and identical faces (lines, $\propto p^6$ as $p\to0$); exact transfer-matrix exponents for opposite faces (circles); closed form $p^6/(1-p^2)^2$ (thin). (b)~Born rate of random C$_3$ Kramers models for four face relations, all $\propto p^6$ ($\kappa=3$). (c)~Index against rotation order: fitted values of random models (five seeds, all face relations, all energies; dots), Eq.~\eqref{eq:kappamatrix} (bars), C$_1$ Kramers pairs with one interface profile from transfer matrices (open circles), and the value suggested by single-bond spreads (cross).}\label{fig:matrix}
\end{figure}
\begin{figure}[!tbp]\centering
\includegraphics[width=\textwidth]{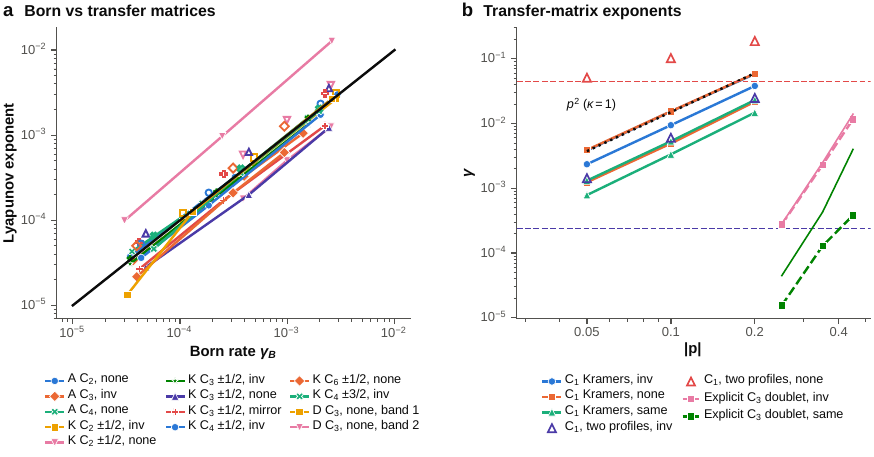}
\caption{Born theory against exact exponents. (a)~Smallest exact Lyapunov exponent (filled) and largest (open, Kramers classes) against the Born rate of the matrix chain at the same momentum and energy, for the eleven classes of check (A) and the two bands of the doublet of check (D); line: equality. (b)~Exact exponents of C$_1$ Kramers pairs with one interface profile (three face relations, two seeds; dotted: $p^2$, $\kappa=1$), with two profiles (open triangles; dashed lines: their values on the axis, $p=0$), and of the explicit C$_3$ doublet for opposite and identical faces (squares), with their Born rates (thin lines).}\label{fig:matrixval}
\end{figure}

\subsection{Correlated twist sequences}\label{sec:corrtwist}
The chains above take independent, uniform twists. Real stacks can carry correlated ones: a helix whose pitch fluctuates, twist diffusion, or rotations concentrated at interfaces between thicker blocks. Let the increments $\Delta_l=\theta_{l+1}-\theta_l$ be independent, with $|f_m|<1$ for every harmonic $m\neq0$ that enters (defined below; this holds for any distribution not confined to a translated lattice of angles $\Delta_0+2\pi j/m$, e.g.\ a continuous one or a point mass plus a continuous part, as for rotations concentrated at block interfaces), and $\theta_0$ uniform, so that the sequence is stationary with a uniform one-point distribution. The averaged on-site block $\bar X$ is unchanged; the averaged bond $\bar T=\langle T(\theta_l,\theta_l+\Delta_l)\rangle$ changes at the order of the twist dependence. Between a forward channel $a$ and a backward channel $b$ of the averaged chain at the same energy, the Born rate~\eqref{eq:matrixborn} becomes
\begin{equation}
\gamma_B=\frac12\Big\langle\sum_b\frac{S_{ab}(k_a-k_b)}{|v_av_b|}\Big\rangle_a,\qquad S_{ab}(q)=\sum_{r}\langle m_l^*m_{l+r}\rangle\,e^{iqr},\label{eq:corrrate}
\end{equation}
with the backscattering amplitude $m_l=\langle b|\delta X_l|a\rangle+e^{ik_a}\langle b|\delta T_l|a\rangle+e^{-ik_b}\langle b|\delta T_l^\dagger|a\rangle$ of layer $l$ and bond $l$ ($\delta X_l=X_l-\bar X$, $\delta T_l=T_l-\bar T$). For independent twists only $r=0,\pm1$ contribute, and Eq.~\eqref{eq:corrrate} reduces to Eq.~\eqref{eq:matrixborn}. A harmonic of one angle, $e^{im\theta_l}$, enters through the structure factor of the sequence,
\begin{equation}
S_m(q)=\sum_r\big\langle e^{-im\theta_l}e^{im\theta_{l+r}}\big\rangle e^{iqr}=\frac{1-|f_m|^2}{|1-f_m e^{iq}|^2},\qquad f_m=\langle e^{im\Delta}\rangle,\label{eq:structurefactor}
\end{equation}
A bond harmonic depends on two adjacent angles, $e^{im_1\theta_l+im_2\theta_{l+1}}=e^{iM\theta_l}e^{im_2\Delta_l}$ with $M=m_1+m_2$, so the part of $m_l$ that carries $e^{iM\theta_l}$ has the form $e^{iM\theta_l}g_M(\Delta_l)$, with the spectrum
\begin{equation}
S_M(q)=\big\langle|g_M|^2\big\rangle+2\,{\rm Re}\,\frac{\langle g_M\rangle\big\langle g_M^*e^{iM\Delta}\big\rangle e^{iq}}{1-f_Me^{iq}},\label{eq:bondspectrum}
\end{equation}
which reduces to Eq.~\eqref{eq:structurefactor} for constant $g_M$. Because $\theta_l$ is uniform, different $M$ do not interfere: $S_{ab}(q)=\sum_MS_M(q)$, and each $S_M$, the power spectrum of a stationary sequence, is finite, non-negative and smooth in $q$ for $|f_M|<1$. Correlations therefore rescale the leading spectrum, of order $p^{2(\kappa-r)}$, by weights that depend on $p$ only through the slow change of $k_a-k_b$, and they keep its order unless every contributing $S_M$ vanishes at $q=k_a-k_b$, which requires coefficients or an increment support tuned to that wave number. The frame-rotation cancellation of Proposition~\ref{prop:matrix} is of another kind: it holds bond by bond, for every sequence, so frame rotations do not enter $m_l$ at all. The correlations of harmonic $m$ decay as $|f_m|^{|r|}$, over $-1/\ln|f_m|$ layers ($\simeq1/(1-|f_m|)$ for $|f_m|$ near 1); beyond this length the thickness theorem of Sec.~\ref{sec:theorem} applies with the rescaled rates wherever its hypotheses hold. A strictly regular twist, $|f_m|=1$, concentrates $S_m$ into Bragg peaks, and the chain becomes periodic or quasiperiodic~\cite{SiriviboonVolkov2026}.

\emph{Test} (\texttt{models/\allowbreak corr17.py}). Random face-resolved models (seed 1) of a spinless A-type C$_2$ edge ($\kappa=2$) and of a C$_3$ Kramers pair ($\kappa=3$: its order-2 spin-flip harmonic is a frame rotation that cancels), with twists on a $0.5^\circ$ grid: independent uniform angles, and random walks with Gaussian increments of mean $\mu$ and standard deviation $\sigma$, $(\mu,\sigma)=(10^\circ,3^\circ)$ (a helix with fluctuating pitch; correlation lengths 180 and 80 layers for $m=2$ and 3), $(0,20^\circ)$ (twist diffusion; 4 and 2) and $(30^\circ,10^\circ)$ (16 and 7). Born rates follow from the averaged periodogram of $m_l$ (1600 sequences of $2\times10^4$ layers, the same sequences at every $p$) at the common-window energy; for independent twists they reproduce Eq.~\eqref{eq:matrixborn} within 3\%.
\begin{center}\small
\begin{tabular}{llccccc}
\toprule
model & twists & \multicolumn{3}{c}{$\gamma_B/\gamma_B^{\rm ind}$ at $p=$} & \multicolumn{2}{c}{local exponent}\\
 & & 0.04 & 0.08 & 0.16 & 0.04--0.08 & 0.08--0.16\\
\midrule
C$_2$, spinless A & independent & 1 & 1 & 1 & 3.997 & 3.990\\
 & walk$(10^\circ,3^\circ)$ & 0.0040 & 0.0040 & 0.0041 & 4.001 & 4.010\\
 & walk$(0,20^\circ)$ & 0.173 & 0.173 & 0.173 & 3.998 & 3.995\\
 & walk$(30^\circ,10^\circ)$ & 0.081 & 0.081 & 0.085 & 4.002 & 4.048\\
C$_3$, Kramers $\pm\frac12$ & independent & 1 & 1 & 1 & 5.870 & 5.748\\
 & walk$(10^\circ,3^\circ)$ & 0.0066 & 0.0066 & 0.0066 & 5.859 & 5.751\\
 & walk$(0,20^\circ)$ & 0.325 & 0.307 & 0.298 & 5.789 & 5.705\\
 & walk$(30^\circ,10^\circ)$ & 0.117 & 0.103 & 0.107 & 5.680 & 5.804\\
\bottomrule
\end{tabular}
\end{center}
Readings. (i)~Correlations lower the rate by factors of 3 to 250 that change by at most 4\% over a factor 4 in $p$ for the C$_2$ edge and by at most 12\% for the C$_3$ pair, whose common-window energy moves with $p$ (sampling errors 2.4--2.7\%, from the scatter of the per-sequence totals over all channel pairs). (ii)~The local exponents agree with those of independent twists within 0.2 (paired sampling errors below 0.01 for the C$_2$ edge and 0.05 for the C$_3$ pair), i.e.\ the index within 0.1. The C$_3$ Kramers pair keeps $\kappa=3$ (local exponents 5.7--5.9, tending to 6); were its order-2 frame rotation not cancelled, the exponent would be 4. (iii)~Exact Lyapunov exponents at $p=0.3$ and 0.45 (transfer matrices along sampled sequences, 8 chains, up to $10^6$ layers) follow the correlated Born rates: 0.91--1.08 of them for the C$_2$ edge, except 0.69 at the largest correlated rate (walk$(30^\circ,10^\circ)$, $p=0.45$, $\gamma_B=0.011$); for the C$_3$ pair 0.37--0.64 of them, the two-channel factor of check (A) in Sec.~\ref{sec:matrixval} (0.49--0.54 here for independent twists), except 0.21 for the helix at $p=0.45$, where the common-window energy moves across a peak of the structure factor.

\subsection{On-site index: helicity counting}\label{sec:onsite}
Rotating a layer by $\theta$ turns its block into $X_\theta(\bp)=U(\theta)X(R_\theta^{-1}\bp)U(\theta)^\dagger$. Expand $X$ in terms $p^\mu e^{i\nu\varphi}E_h$ with $E_h$ an internal operator of helicity $h$, $UE_hU^\dagger=e^{ih\theta}E_h$. A rotation multiplies the term by $e^{iL\theta}$, $L=h-\nu$. C$_n$ symmetry allows $L\in n\mathbb Z$, analyticity requires $|\nu|\le\mu$ and $\nu\equiv\mu\pmod2$, and time reversal fixes the parity of each internal component; inversion removes the odd ones. The twist-dependent part has $L\neq0$, so
\begin{equation}
\kappa_{\rm on}=\min\big\{\mu\ge\mu_{\min}:\ \exists\,(h,\nu),\ |\nu|\le\mu,\ \nu\equiv\mu\ (2),\ (-1)^\mu=\pi_h,\ h-\nu\in n\mathbb Z\setminus\{0\}\big\},\label{eq:kon}
\end{equation}
with $\mu_{\min}=2$ for acoustic in-plane blocks (sum rule), 4 for flexural blocks (rotational invariance of a free layer) and 0 otherwise. This on-site rule is checked by random models for 70 entries (\texttt{index7.py}). It is part of the matrix index: the table above includes the on-site disorder, and $\kappa\le\kappa_{\rm on}$ except where an on-site term acts only between closed branches (footnote $c$).

\paragraph{Band edges off the axis.} An on-site term linear in $p$, a Rashba-type splitting $\lambda p$ allowed only in layers without inversion, splits the monolayer band and, if it is isotropic, moves the band minimum to a ring $p_0=\lambda m^*/\hbar^2$ of depth $E_R=\lambda^2m^*/2\hbar^2$; an anisotropic $\lambda$ gives isolated minima instead. The axis protects thermal carriers only if $E_R\ll k_BT$; for $\lambda=0.05$\,eV\,\AA\ and $m^*=0.5$, $E_R/k_BT=0.003$ at 300\,K. An index $\kappa=1$ does not imply such a term. In a layer with inversion and time reversal every band is Kramers degenerate at every $\mathbf q$, and the $\kappa=1$ of a C$_2$ Kramers pair comes from its face amplitudes alone: random centrosymmetric C$_2$ models give $\kappa=1.000$--$1.001$ with an on-site splitting that vanishes identically, while their polar counterparts split as $p^{1.00}$ (\texttt{models/\allowbreak checks10.py}). The kit therefore reads the linear splitting and its anisotropy from the monolayer bands (T0-a) and makes the ring check when a linear term is present, whatever the fitted $\kappa$.

\paragraph{The two materials.} Black phosphorus (D$_{2h}$, C$_2$ about the stacking axis): first-principles face data (Sec.~\ref{sec:dft}) give a conduction edge with $r=0$, $\kappa=2.00$ and a valence edge with $r=\kappa=2$ on every admissible gap plane (a scale invariance; its Born rates exceed 0.1 per layer, so the tool does not return a Born index), spin--orbit crossovers at $p^*=(0.9$--$1.7)\times10^{-3}$ (planes and coupling bracket, at an energy 23\% of the bandwidth above the band bottom; $(1.7$--$4.9)\times10^{-4}$\,\AA$^{-1}$ at thermal energies) and $p_h=(1.24$--$1.36)\times10^{-3}$\,\AA$^{-1}$ and a forward hole/electron coupling ratio $(3.2$--$4.5)\times10^{-5}$; exact transport of the full face Hamiltonian gives electrons a conductance approaching the axis law $N^{-1/2}$ and holes, in the forward channel, a fraction of their aligned conductance at least 7 times smaller than the electrons' at 16 layers and $10^2$ times smaller at 64 (exponential below the transparent-energy threshold $\lambda_c$, approaching $N^{-3}$ above; Sec.~\ref{sec:thicknessdft}); in the phosphorene models of Sec.~\ref{sec:tbexamples} the $\Gamma$ conduction edge is A-type ($r=0$, $\kappa\approx2$, 1.84--1.96 at 0.016--0.1\,\AA$^{-1}$; spin--orbit coupling drives it toward the Kramers value 1 only below $p^*\approx0.007$\,\AA$^{-1}$) and the valence edge B-type ($r=2$, not protected; spin--orbit coupling adds a forward coupling $3\times10^{-4}$ of the electrons' in a window $p\lesssim p_h\approx0.02$\,\AA$^{-1}$), a finite-thickness electron--hole contrast in the models; flexural $\kappa_z=4$; in-plane acoustic not protected. The MoS$_2$ $\Gamma$ valence band (D$_{3h}$, horizontal mirror) is a $J_z=\pm\frac12$ Kramers pair: $r=0$, $\kappa=3$ from both the on-site spin--orbit $p^3$ term and the face amplitudes (2.94 in the atomistic model), crossing over to the scalar value 6 where the spin--orbit amplitudes are negligible. The C$_3$ change of Table~\ref{tab:kappa} leaves MoS$_2$ at $\kappa=3$ because its on-site index was already 3.

%% file: sec_theorem.tex
\section{The thickness theorem on the zero set of the effective Lyapunov exponent}\label{sec:theorem}

\subsection{Statement}
Let $x=(\bp,E)$ label channels and energies (for phonons $E=\omega$), with measure $\dd\mu(x)=w(x)\,\dd^2p\,\dd E$ (bandwidth, Fermi--Dirac window, or classical $\dd\omega$). Write the normalized conductance as
\begin{equation}
g(N)=\frac{\int\dd\mu(x)\;t_N(x)}{\int\dd\mu(x)},
\end{equation}
where $t_N(x)$ is the disorder-averaged transmission of the open channel(s) at $x$, normalized to the aligned stack. Let $\gamma(x)$ be the Lyapunov exponent of the \emph{effective} chain at $x$, obtained after closed components are eliminated, and $\Zg=\{x\in{\rm supp}\,\mu:\gamma(x)=0\}$ its zero set.

\begin{theorem}[thickness law]\label{thm:general}
Assume:
\begin{description}
\item[H1 (geometry)] $\Zg$ is a finite union of compact $C^2$ submanifolds $Z_i$ of codimension $c_i\ge1$ in ${\rm supp}\,\mu$, each with a tubular neighbourhood in coordinates $(\zeta,\rho,\hat n)$: position $\zeta$ on $Z_i$, normal distance $\rho$, normal direction $\hat n$.
\item[H2 (local orders, positivity)] In the tube
\[
\gamma=\Gamma_i(\zeta,\hat n)\,\rho^{2m_i}\,[1+O(\rho)],\qquad \dd\mu=\omega_i(\zeta,\hat n)\,\rho^{s_i+c_i-1}\,\dd\rho\,\dd\hat n\,\dd\zeta\,[1+O(\rho)],
\]
with $\Gamma_i,\omega_i$ continuous, $0<\Gamma_{\min}\le\Gamma_i\le\Gamma_{\max}$, $m_i>0$ and $s_i>-c_i$.
\item[H3 (uniform decay away from $\Zg$)] For every $\delta>0$ there are $C,\mu_\delta>0$ with $t_N(x)\le Ce^{-\mu_\delta N}$ wherever ${\rm dist}(x,\Zg)\ge\delta$.
\item[H4 (kernel)] Locally uniformly in the scaling argument, i.e.\ for every $Q<\infty$ uniformly where $2N\gamma(x)\le Q$, $t_N(x)=\Phi_{\zeta\hat n}(2N\gamma(x))\,[1+O(\rho)]$, with a kernel $\Phi_{\zeta\hat n}=\Phi_i(\zeta,\hat n;\cdot)$ continuous in its argument and in $(\zeta,\hat n)$; in the tubes $t_N\le\bar\Phi(2N\gamma)$ with a non-increasing $\bar\Phi$ and $\int_0^\infty q^{\alpha_i-1}\bar\Phi(q)\,\dd q<\infty$. The kernel belongs to the disorder class of the effective chain, not to $\gamma$, and may depend on the component, the point and the normal direction (e.g.\ differ on the two sides of $\Zg$): for one open channel with weak finite-variance disorder and random phases $\Phi=\Phi_1$ (anomalous energies such as the band centre are excluded~\cite{SchomerusTitov2003}); effective disorder with a heavy tail has its own one-channel kernel (Corollary~\ref{cor:heavy}); for several open channels $\Phi$ is the multichannel kernel of the effective chain, which need not be characterized by one Lyapunov exponent (Remark~\ref{rem:multi}).
\item[H4$'$ (rate)] A remainder bound valid \emph{throughout} the tubes, $|t_N(x)-\Phi_{\zeta\hat n}(2N\gamma(x))|\le C\rho\,R(2N\gamma(x))$ with a non-increasing $R$, and $\int_0^\infty q^{\alpha_i-1+1/2m_i}\,[R(q)+\bar\Phi(q)+q\sup_{\zeta,\hat n}\sup_{q/2\le\xi\le2q}|\Phi_{\zeta\hat n}'(\xi)|]\,\dd q<\infty$.
\end{description}
Then, with $\alpha_i=(c_i+s_i)/2m_i$ and $\alpha=\min_i\alpha_i$, $N^\alpha g(N)\to\sum_{\alpha_i=\alpha}A_i$ under H1--H4. Under H4$'$ in addition,
\begin{equation}
\begin{gathered}
g(N)=\sum_{i:\,\alpha_i=\alpha}A_i\,N^{-\alpha}\left[1+O\!\left(N^{-1/2m_i}\right)\right]+O\!\left(N^{-\alpha'}\right),\\
A_i=\frac{1}{2m_i\int\dd\mu}\int_{Z_i}\dd\zeta\oint\dd\hat n\;\omega_i\,(2\Gamma_i)^{-\alpha}J_{\alpha}(\zeta,\hat n),\qquad J_\alpha(\zeta,\hat n)=\int_0^\infty q^{\alpha-1}\Phi_{\zeta\hat n}(q)\,\dd q.
\end{gathered}\label{eq:general}
\end{equation}
where $\alpha'$ is the next larger exponent among the $\alpha_i$ (the term is absent if there is none).
\end{theorem}
\noindent\emph{Proof} (App.~\ref{app:theorem}). By H3 the complement of the tubes contributes $O(e^{-\mu_\delta N})$. In a tube substitute $q=2N\Gamma_i\rho^{2m_i}$ at fixed $(\zeta,\hat n)$; then $\rho^{c_i+s_i-1}\dd\rho=(2N\Gamma_i)^{-\alpha_i}q^{\alpha_i-1}\dd q/2m_i$. At fixed $q$, $\rho=O(N^{-1/2m_i})$; dominated convergence (majorant of H4) gives the limit, and with H4$'$ the $O(\rho)$ remainders of H2 and H4 integrate to the stated relative error. Components with $\alpha_i>\alpha$ contribute $O(N^{-\alpha_i})$. $\square$ The exponent does not depend on the form of $\Phi$, only on its integrability; the amplitude does, through $J_\alpha(\zeta,\hat n)$, which stays inside the direction integral unless the kernel is the same everywhere on $Z_i$. The exponent does depend on the measure, through $s$: contact weights and the density of channels at the zero set enter there, as boundary conditions and low-frequency weights change the length exponent of heat conduction through disordered harmonic chains and crystals~\cite{Dhar2001,Kundu2010}.

\emph{Why H4$'$ needs a global remainder.} A bound that holds only at fixed $N\gamma$ does not control the tail of the $q$ integral, which moves with $N$. Example: $\rho\in[0,1]$, $\gamma=\rho^2/2$, $\alpha=1/2$, $\Phi(q)=e^{-q}$ and $t_N=e^{-q}+\frac12e^{-q/2}\,\mathbf 1\{q>\ln\ln N\}$. At every fixed $q$, $t_N=\Phi(q)$ for large $N$, and the majorant has every moment; yet the relative correction to the leading term is ${\rm erfc}\big(\sqrt{\ln\ln N/2}\big)/\sqrt2=0.096$, 0.062, 0.041 at $N=10^4$, $10^8$, $10^{16}$, not $O(N^{-1/2})$. The global bound of H4$'$ excludes it; the limit statement (H1--H4) is unaffected.

For channels whose band shrinks on the axis ($r>0$) the tube in $(\bp,E)$ is not an ordinary tubular neighbourhood, because the band collapses there. The theorem is then applied in resolved (blown-up) coordinates $(\bp,\varepsilon)$ with $\varepsilon=(E-E_{\rm edge}(\bp))/W(\bp)\in[0,1]$, and to the pushed-forward measure, in which the tube is a product; the bandwidth enters $s$ ($s=r$ for electrons, whose band width is $\propto p^r$; $s=1$ for acoustic in-plane phonons, whose frequency width is $\propto p^{r-1}$).

\begin{remark}[several open channels]\label{rem:multi} With $M$ open channels at $x$ the transmission is governed by $M$ Lyapunov exponents, not one, and the kernel $\Phi$ of H4 is that of the effective $M$-channel chain. The exponent $\alpha$ of Eq.~\eqref{eq:general} depends only on how the exponents vanish on $\Zg$ (all generically with the same order $2m$), so it is unchanged; the amplitude $A$ requires the multichannel kernel. In this note amplitudes are quoted only for one-channel cases. For the matrix index (Sec.~\ref{sec:index}) the channel-averaged Born rate and the smallest Lyapunov exponent differ by an $O(1)$, class-dependent factor, so only exponents are claimed there.\end{remark}

For a scalar channel with on-site disorder $\Veff\propto\rho^{\kappa}$ and coupling $T\propto\rho^r$, Born theory gives $\gamma\propto\Veff^2/T^2$, so $2m=2(\kappa-r)$ and Eq.~\eqref{eq:general} is the law $N^{-(c+s)/2(\kappa-r)}$. What the theorem requires is the zero set and local order of the \emph{effective} exponent, including its positivity away from $\Zg$, not of the bare disorder.

\subsection{Corollaries and named exceptions}
\begin{corollary}[rotation axis]\label{cor:point} For the axis $\{\bp=0\}\times$(energy window), $c=2$ and $\zeta=E$. With $r=s=0$ this is $N^{-1/\kappa}$; with the classical phonon measure and a coupling vanishing as $p^2$ it gives the in-plane law of Sec.~\ref{sec:phonon}.\end{corollary}
\begin{corollary}[exponential decay]\label{cor:exp} If $\gamma\ge\gamma_{\min}>0$ on ${\rm supp}\,\mu$ and the channels obey a uniform large-deviation bound $t_N(x)\le Ce^{-cN\gamma(x)}$, $c>0$ (at weak disorder the mean transmission decays at the rate $\gamma/2$ up to algebraic factors, since $\Phi_1(s)\sim s^{-3/2}e^{-s/4}$), then $g(N)\le Ce^{-c\gamma_{\min}N}$. Conversely, under H1--H4 a non-empty $\Zg$ gives algebraic decay. The bare condition $\kappa\le r$ implies neither: eliminating closed components can create zeros of $\gamma$ elsewhere (Corollary~\ref{cor:transparent}).\end{corollary}
\begin{corollary}[transparent energies]\label{cor:transparent} Suppose the effective disorder of a channel vanishes at an energy $E^*$ because the bare disorder cancels, against the energy-dependent self-energy of an eliminated component or between correlated on-site and bond terms. If (i) $E^*$ lies in the interior of the clean band ($\sin k^*$ bounded away from zero, so that $\Gamma$ stays bounded as H2 requires), (ii) the effective disorder has finite variance near $E^*$, and (iii) no eliminated component has a pole (Schur resonance) near $E^*$, then generically $\gamma\propto(E-E^*)^2$: $\Zg$ has a component of codimension 1 in energy with $m=1$, and $g\propto N^{-(1+s)/2}$, i.e.\ $N^{-1/2}$ for $s=0$, however large the bare disorder is (as for the window of states whose localization length exceeds the sample around the delocalized energies of dimer and correlated continuous models~\cite{Dunlap1990,Hilke1997}). Condition (ii) fails when the cancelling transformation is singular at a node; that case is Corollary~\ref{cor:heavy}.\end{corollary}
\begin{corollary}[transparent energies at nodes: heavy tails]\label{cor:heavy} Suppose that at $E^*$, inside the clean band, a local transformation $\phi_l=s_l\psi_l$ removes the bare disorder exactly, as for correlated on-site and bond disorder of the factorized form $t_l=t_0s_ls_{l+1}$, $e_l=E^*+2\Delta s_l^2$ (Sec.~\ref{sec:thicknessdft}; cf.~\cite{Flores1989}), so that near $E^*$ the only disorder of the transformed chain is $V_l=\delta\,Y_l$ with $\delta=(E-E^*)/t_0$, i.i.d.\ $Y_l=1/s_l^2\ge1$ and a tail $P(Y>y)\simeq a\,y^{-\beta}$, $0<\beta<2$ (infinite variance). Strong scatterers ($V_l$ of order one) then occur with density $\propto|\delta|^\beta$, which makes $q=N|\delta|^\beta$ the natural scaling variable and suggests $\gamma\propto|\delta|^{\beta}$ ($2m=\beta$) for weak, non-resonant disorder; we prove neither for general $\beta$. If the mean transmission has a scaling limit $\Psi_\beta(q,\operatorname{sgn}\delta)$ with an integrable majorant (H4 with this kernel), the change of variables alone gives $g\propto N^{-(1+s)/\beta}$ for one channel; that kernel is not $\Phi_1$ (below). For a linear node and uniform twists, $Y=1/\sin^2\alpha$ with $P(Y>y)=(2/\pi)\arcsin y^{-1/2}\simeq(2/\pi)y^{-1/2}$ exactly: $\beta=\frac12$, the domain of attraction of the one-sided stable law of index $\frac12$, and $g\propto N^{-2}$ for $s=0$. In the dilute limit $q\to0$ a single strong scatterer, $T_1(V)=[1+V^2/4\sin^2k]^{-1}$ with $k$ the clean wave number at $E^*$, gives
\begin{equation}
1-\langle T\rangle=\frac{q}{2\sqrt{\sin k}}+o(q),\qquad -\langle\ln T\rangle=\frac{2q}{\sqrt{\sin k}}+o(q),\label{eq:dilute}
\end{equation}
a ratio $\frac14$ (for a tail of index $\beta$ the same integrals give $\beta/2$), whereas the weak-scattering kernel $\Phi_1$ has ratio 1 at small $q$.\end{corollary}
\noindent\emph{Proof of the dilute limit.} The tail density of $V$ is $p(V)=\pi^{-1}|\delta|^{1/2}V^{-3/2}$ for $V\gg|\delta|$ (one sign, that of $\delta$); scatterers are independent, so to first order in $q$, $1-\langle T\rangle=N\int p(V)[1-T_1(V)]\,\dd V$ and $-\langle\ln T\rangle=N\int p(V)\ln[1+V^2/4\sin^2k]\,\dd V$. With $V=2w\sin k$, $\int_0^\infty w^{1/2}(1+w^2)^{-1}\dd w=\pi/\sqrt2$ and $\int_0^\infty w^{-3/2}\ln(1+w^2)\,\dd w=2\sqrt2\,\pi$ give Eq.~\eqref{eq:dilute} (the limit is taken at fixed $q$ as $N\to\infty$, then $q\to0$); the bulk of the distribution ($V\sim\delta$) contributes $O(N\delta^2)=O(q^4/N^3)$. For a tail of index $\beta$, $\int_0^\infty w^{1-\beta}(1+w^2)^{-1}\dd w=\pi/[2\sin(\pi\beta/2)]$ and $\int_0^\infty w^{-1-\beta}\ln(1+w^2)\,\dd w$ is $2/\beta$ times it, so the ratio is $\beta/2$: the value $\frac14$ is a property of the linear node, not of heavy tails in general. Beyond the dilute limit the kernel is the mean transmission of a Poisson set of strong scatterers with random phases between them; we do not derive it. Its existence and an integrable tail, hence H4, are supported numerically over the sampled range (Tests below), as is the law $g\propto N^{-2}$; they are not proved. Heavy-tailed disorder is known to take one-dimensional localization out of the finite-variance class of single-parameter scaling~\cite{TitovSchomerus2003}, with scaling classes set by the tail index for symmetric, uncorrelated tails at weak disorder~\cite{Texier2020}; these results do not cover the one-sided, node-generated tail of index $\frac12$ here, whose dilute limit Eq.~\eqref{eq:dilute} we derive directly. As for random potentials with large local fluctuations~\cite{Bienaime2008}, the derivation averages the full single-scatterer response, not its second moment, which diverges. The heavy tail is in the strength of scatterers at fixed positions; it differs from L\'evy-distributed spacings, for which the number of scatterers in a given length is itself random~\cite{Falceto2010}.
\begin{corollary}[mixed-order zero sets]\label{cor:mixed} Let $\Zg$ contain an isolated point $x_0$ near which, in coordinates $y=(y_1,\dots,y_J)$ with $y_j\in\mathbb R^{c_j}$, $\dd\mu=w_0\prod_j|y_j|^{s_j}\dd y_j\,[1+o(1)]$, and suppose that the mean transmission has a scaling limit: there are $\nu_j>0$ such that, with $y_j=N^{-\nu_j}\eta_j$, $t_N\to\Psi(\eta)$ as $N\to\infty$ at fixed $\eta$, dominated by $\bar\Psi$ with $\int\bar\Psi\prod_j|\eta_j|^{s_j}\dd\eta_j<\infty$, and that H3 holds away from $x_0$. Then
\begin{equation}
N^{\alpha}g(N)\to\frac{w_0}{\int\dd\mu}\int\Psi(\eta)\prod_j|\eta_j|^{s_j}\dd\eta_j,\qquad\alpha=\sum_j\nu_j(c_j+s_j).\label{eq:mixed}
\end{equation}
If $\gamma\simeq\Gamma\sum_j|y_j|^{2m_j}$ and $\Psi$ depends on $N\gamma$ alone, $\nu_j=1/2m_j$ and $\alpha=\sum_j(c_j+s_j)/2m_j$.\end{corollary}
\noindent\emph{Proof.} Substitute $y_j=N^{-\nu_j}\eta_j$: $\prod_j|y_j|^{s_j}\dd y_j=N^{-\alpha}\prod_j|\eta_j|^{s_j}\dd\eta_j$; dominated convergence gives the limit, and H3 the complement. $\square$ If the support of $\mu$ near $x_0$ is a cusp, as for the holes, whose band restricts $|E-E^*|\lesssim p^2$, the restriction becomes $|\eta_E|\lesssim N^{2}|\eta_p|^2$ in the scaled coordinates ($\nu_p=\frac14$, $\nu_E=\frac52$) and disappears in the limit, so the statement applies with $\Psi$ extended by 0 outside the band. The hypothesis is a joint scaling limit of the mean transmission, not the form of $\gamma$; it is what the hole transparent energy of black phosphorus requires (Sec.~\ref{sec:thicknessdft}), where the two directions are the residue $\gamma(E^*)\propto p^4$ and the heavy-tailed window $|E-E^*|^{1/2}/p$.
\begin{corollary}[several components]\label{cor:sum} The slowest component wins; a zero ring ($c=1$) beats the axis if $(1+s_1)/2m_1<(2+s_0)/2m_0$.\end{corollary}
\begin{remark}[band centre of bond disorder] At the band centre of pure off-diagonal disorder $\gamma\to0$ without a power law: $\langle\ln\mathcal T\rangle\propto-\sqrt N$, while rare, nearly transparent sequences make the mean transmission decay as $N^{-1/2}$~\cite{Steiner1999}, so that mean and typical transmission follow different laws at one energy; H2 and H4 fail. This point is excluded; it is irrelevant for band-edge carriers and relevant for mid-band Fermi levels (Sec.~\ref{sec:spacing}).\end{remark}

\subsection{Tests}
\paragraph{Discriminating tests.} A coupling that vanishes on the axis ($r=1$) changes the exponent from $1/\kappa$ to $(2+s)/2(\kappa-1)$: fitted slopes $-1.491$ and $-0.483$ against $-3/2$ and $-1/2$ for C$_2$ and C$_4$ layers (bandwidth measure). An added zero ring ($c=1$, $m=2$) makes $N^{-1/4}$ win over the axis law $N^{-1/2}$; the two-term form with both amplitudes from Eq.~\eqref{eq:general} matches the chain within 0--1.5\% at $N=1024$--4096. These tests use scalar chains.

\paragraph{Transparent energies: bare disorder of order one as a covered case.} A propagating orbital $a$ (interlayer hopping $t=1$) is side-coupled with random strength $c_l\in[0.6,1.6]$ to a local orbital $b$ at $e_b=2.4$; the bare on-site of $a$ is $v_l=-c_l^2/(E^*-e_b)$, so the eliminated on-site $v_l+c_l^2/(E-e_b)$ is exactly constant at $E^*=0.3$. The bare disorder is $O(1)$, yet Corollary~\ref{cor:transparent} predicts $g=A N^{-1/2}$ with
\begin{equation}
A=\frac{J_{1/2}}{|W|}(2\Gamma_2)^{-1/2},\qquad \Gamma_2=\frac{{\rm Var}(c^2)}{8t^2\sin^2k^*\,(E^*-e_b)^4}=2.689\times10^{-3},\qquad A=7.474,
\end{equation}
for the uniform energy window $|W|=3.98$. The example meets (i)--(iii): $E^*=0.3$ is inside the band and $e_b=2.4$ lies outside the tested interval, so the Schur term is analytic there. The simulation (16 sequences, $N\le8192$, 16001 energies; \texttt{transparent7.py}; the simulation uses the eliminated one-component chain, which is an exact reduction: the full two-orbital Green function agrees with it to $10^{-15}$) gives $g(8192)=0.0825$ against $AN^{-1/2}=0.0826$ (Fig.~\ref{fig:zero}a); the full Born curve $\int\Phi_1(2N\gamma_{\rm Born})$ follows the data within 2.2\% at every $N$ from 8 to 8192, and the local exponent reaches $-0.496$. A control with the same bare spread but no cancellation has decayed to $3.5\times10^{-104}$. The theorem thus predicts the exponent and the amplitude of a chain whose bare disorder is of order one; what it needs is the effective exponent, not the bare disorder.

\paragraph{Transparent energy at nodes: the heavy-tailed kernel.} The ideal factorized chain of Corollary~\ref{cor:heavy} ($t_l=t_0s_ls_{l+1}$, $e_l=E^*+2\Delta s_l^2$, continuous uniform twists, aligned leads) is computed exactly with the recursion of \texttt{vk\_channel\_fast.py} (no absorption, $\eta=0$) for the two positions of $E^*$ in the band that black phosphorus has at the coupling scales 1.53 and 2.15, $\cos k=-\Delta/t_0=-0.900$ and $-0.640$; $2^{14}$ twist sequences per point, each point with its own random streams, $N=64$, 256, 1024 and 4096, $q=N|\delta|^{1/2}$ from 0.03 to 32 on both sides of $E^*$ (\texttt{blackphosphorus/\allowbreak kernel14.py}; Fig.~\ref{fig:kernel}a). (i)~Matched leads: $T(E^*)=1$ to $10^{-13}$ for every sequence. (ii)~Collapse: at fixed $q\le8$ the means for $N=256$ and 1024 agree with those for $N=4096$ within their standard errors (rms of the normalized differences 1.21 and 1.08 over 48 comparisons for $\cos k=-0.90$ and $-0.64$, largest 2.6; the points are independent, but the $N=4096$ reference enters twice, so this is a description, not a test with 48 degrees of freedom); at $N=64$ finite-size deviations remain (rms 2.7 standard errors at $\cos k=-0.90$). (iii)~Mean and typical differ: at $q=1$, $\langle T\rangle=0.48$ (0.57) but $e^{\langle\ln T\rangle}=0.05$ (0.10), and at $q=8$ the largest 1\% of sequences carry 38--86\% of the mean. (iv)~Dilute limit ($q=0.03$--0.25, $N=1024$ and 4096): $(1-\langle T\rangle)/q=0.68$--0.76 against 0.757 of Eq.~\eqref{eq:dilute} and $-\langle\ln T\rangle/q=2.7$--3.1 against 3.03 for $\cos k=-0.900$; 0.53--0.58 against 0.570 and 2.1--2.4 against 2.28 for $\cos k=-0.640$. The ratio is 0.22--0.26, approaching $\frac14$ as $q\to0$; the weak-scattering kernel $\Phi_1$ gives 1. (v)~Tail: $\langle T\rangle\propto e^{-0.52q}$ ($e^{-0.44q}$) for $4\le q\le16$, consistent with an integrable exponential envelope over the sampled range; this is not a proof for all $q$, nor uniformly in $N$. (vi)~The integrated transmission of one channel, $N^2\int\langle T\rangle\,\dd\delta=\sum_\pm\int2q\,\Psi(q)\,\dd q$, is $J=8.73$--8.82 (15.42--15.50), with standard errors 0.05 (0.09), for $N=64$ to 4096, constant within 1\% (trapezoidal rule on the $q$ grid, which overestimates $J$ by about 2\%; the tail beyond $q=16$ adds at most 1.4\%): each channel transmits $Jt_0N^{-2}$, the law of Corollary~\ref{cor:heavy}, with a kernel whose amplitude is not that of $\Phi_1$. The scaling limit is thus supported numerically for $\beta=\frac12$ over $N=256$--4096; it is not proved.

\begin{figure}[!tbp]\centering
\includegraphics[width=\textwidth]{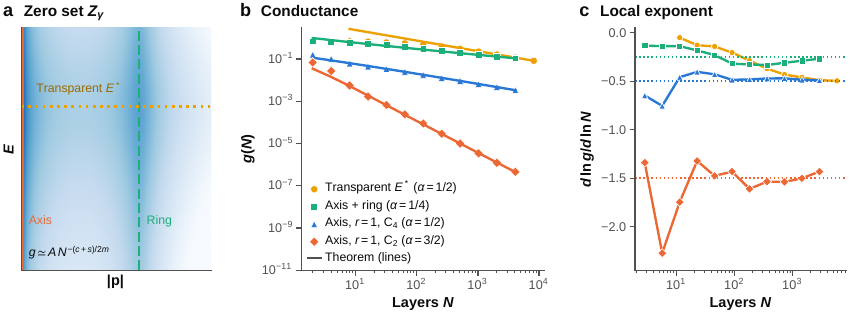}
\caption{Thickness law from the zero set. (a)~Zero sets of the effective Lyapunov exponent in momentum--energy space: the rotation axis, a ring, a transparent energy. (b)~Mean conductance of scalar chains (markers) against the theorem with closed-form amplitudes (lines): transparent energy (circles), axis plus ring (squares), and axis laws with $r=1$ in C$_4$ and C$_2$ layers (triangles, diamonds); panel (a) of letter Fig.~2. (c)~Local exponents approaching the predicted $-\alpha$ (dotted).}\label{fig:theoremtests}
\end{figure}
\begin{figure}[!tbp]\centering
\includegraphics[width=0.74\textwidth]{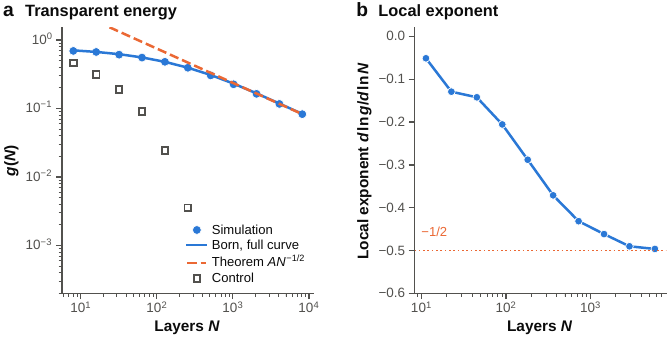}
\caption{Zero-scattering channels. (a)~Bare disorder of order one with a transparent energy: mean conductance of the side-coupled chain (markers, 16 sequences), the Born curve $\int\Phi_1(2N\gamma_{\rm Born})$ (line), the theorem's asymptote $AN^{-1/2}$ with $A=7.47$ computed in closed form (dashed), and a control with the same bare spread but no cancellation (open squares). (b)~Local exponent $\dd\ln g/\dd\ln N$, approaching $-1/2$.}\label{fig:zero}
\end{figure}

\paragraph{Physical models: phonon rotation has no transparent window, spacing noise and nodal edges have one.} Transparent energies require a cancellation that survives disorder averaging. For the rotation of phonons, eliminating the decoupled transverse component would cancel the in-plane disorder only if the rotated monolayer tensor kept its eigenvalues, which it does not in C$_2$ and C$_4$ layers. A scan of the full stress-free $3\times3$ model with rotational disorder (24 momenta $\times$ 160 frequencies per lattice, $N=4096$; Sec.~\ref{sec:phonon}, Fig.~\ref{fig:phonon}d) finds the transmission decreasing monotonically away from the protected region around the axis, with no isolated window at the scanned resolution: for rotation, $\Zg$ is the axis within that resolution (a narrower window between grid points is not excluded). Spacing noise is different. The bond and the sum-rule diagonal of an out-of-plane phonon respond to the same spacing, and their backscattering amplitudes cancel on a surface $\omega^*(\bp)$ inside the band (Sec.~\ref{sec:joint_asym}). This is a transparent window of the kind of Corollary~\ref{cor:transparent}, present in the stable model and visible in its transmission. For electrons, rotation itself produces one at nodal (B-type) edges: the factorized chain of Corollary~\ref{cor:heavy}, tested next and realized by the holes of black phosphorus (Sec.~\ref{sec:thicknessdft}).

%% file: sec_phonon.tex
\section{Phonons in the stable model}\label{sec:phonon}

\subsection{Stability}
Normal-force balance makes the out-of-plane sector stable at small $p$, but full stability also needs the in-plane condition $c_L^2>2|u_0|$ (Sec.~\ref{sec:stressfree}); it is checked numerically. Every production chain passes the gate of App.~\ref{app:stress}: 5 stiffness families $\times$ 657 samples (42 configurations, 16 sequences each except the noise-free aligned stack), all checkpoints $N=4$--256 and all 120 channels, plus every sequence of the branch-exponent runs. No negative eigenvalue occurs; the lowest eigenvalue of the dense Hessian of the smallest channel is $+3.52\times10^{-6}$\,eV\,\AA$^{-2}$\,amu$^{-1}$, the in-plane acoustic value $\simeq c_Tp^2$ at $|\bp|=0.01$\,\AA$^{-1}$. For comparison, pre-stressed chains (spacings imposed at a fixed potential) at 0.06\,\AA\ have negative modes in 2, 15, 15 and 14 of 16 samples (i.i.d., positional, curvature, alternating), down to $-8.2\times10^{-6}$ (Fig.~\ref{fig:phonon}a). On the square lattices the gate also finds negative modes at $|\bp|\approx|\mathbf G|$ on a grid extended beyond the zone disk (C$_4$ full model: 94 negative eigenvalues at $|\bp|=1.51$--1.60\,\AA$^{-1}$ over 4 sequences, none on the zone disk; \texttt{domain\_check7.py}), a further reason for the channel domain of Sec.~\ref{sec:framework}.

\subsection{Branch exponents}
Theorem~\ref{thm:general} at the axis ($c=2$) gives:
\begin{itemize}
\item \textbf{out-of-plane branch}: $r=0$, $s=0$, $G_z\propto N^{-1/\kappa_z}$, $\kappa_z=\max(4,\kappa_{\rm scalar})$ for a free layer;
\item \textbf{in-plane L branch}: $r=2$, $s=1$ (channel bandwidth $\propto p$), $G_L\propto N^{-3/2(\kappa_t-2)}$: $N^{-3/4}$ for C$_3$ and C$_6$ layers; C$_2$ and C$_4$ layers have $\kappa_t=r$: on-site disorder and coupling scale alike, $\gamma$ tends to a positive, $p$-independent limit on the axis, and the in-plane chains decay exponentially (Corollary~\ref{cor:exp});
\item \textbf{in-plane T branch}: decoupled for central forces; non-central forces make it behave like L.
\end{itemize}
In non-centrosymmetric C$_3$ layers (D$_{3h}$, e.g.\ MoS$_2$) Table~\ref{tab:kappa} also allows a gyrotropic in-plane term $i\epsilon_{ij}f(\bp)$ with $f\propto p^3\cos3\varphi$ ($\kappa_{\rm on}=3$). With central interlayer forces only the L projection couples, and $\hat p^{\top}\epsilon\,\hat p=0$, so this term enters the L channel only through L--T mixing: after the closed T component is eliminated it is of order $p^6/p^2=p^4$, and the L exponent stays $3/4$. With non-central forces ($F_{TT}\neq0$) a T band exists; the mixing scatters between open channels, and the law becomes $N^{-3/2}$, only if that band is open at the frequency of the L channel. A T band that does not overlap the L band in frequency leaves T closed, the $p^3$ term acts through it at second order, and the central-force exponent $3/4$ stands. The size of $f$ follows from ${\rm Im}\,D_{xy}(\bp)$ of the T0-b phonon calculation; the kit does not fit it, and the L-channel law does not need it for central forces.
The stable, zone-consistent chains confirm all of them (Fig.~\ref{fig:phonon}; 4 sequences, $N\le2048$). Local slopes at the largest $N$: C$_6$ $z$ $-0.10$ and full $-0.13$ (asymptote $-1/6$, reached as slowly as for C$_6$ electrons), in-plane $-0.65$ to $-0.72$ ($-3/4$); C$_4$ $z$ $-0.22$ and full $-0.21$ ($-1/4$), in-plane exponential (below $10^{-190}$ at $N=256$); C$_2$ with a tensioned-membrane $z$ block $-0.52$ ($-1/2$), full $-0.46$, in-plane exponential; flexural C$_2$ $z$ $-0.24$ to $-0.25$ ($-1/4$). Pre-stressed chains (spacings imposed at a fixed potential) give the same exponents within the sampling spread; their conductance ratios differ.

\subsection{Zero set of the effective exponent}
A transparent window would keep a channel transmitting at some $(\bp,\omega)$ away from the axis however thick the stack. We scanned the full stress-free $3\times3$ model on 24 momenta from 0.02\,\AA$^{-1}$ to the zone disk and 160 frequencies spanning each channel's band, with rotational disorder only and 4 sequences of $N=4096$ layers (\texttt{phys7.py scan}; Fig.~\ref{fig:phonon}d). The maximum over frequency of $\langle\mathcal T\rangle$ per open polarization is 1 within 1\% in the protected region near the axis ($|\bp|\le0.12$, 0.14, 0.30 and 0.55\,\AA$^{-1}$ for tensioned C$_2$, flexural C$_2$, C$_4$ and C$_6$ layers). Beyond it the maximum decreases monotonically to $3\times10^{-9}$, $3\times10^{-156}$, $10^{-11}$ and 0.2 at the zone disk, with no isolated maximum at finite $p$. The Lyapunov exponent obtained from the growth of $-\langle\ln\mathcal T\rangle$ between $N=1024$ and 4096 is below the resolution ($\sim10^{-4}$) in the protected region and grows monotonically beyond it. Within this resolution, $\Zg$ of rotational disorder is the axis. Spacing noise adds a zero surface of its own (Sec.~\ref{sec:joint_asym}).

\subsection{What the continuum filter does, and what registry does}\label{sec:registry}
In the zone-consistent C$_6$ model, rotation alone reduces the phonon conductance only modestly at moderate thickness: $g_{\rm ph}=0.95$, 0.86 and 0.78 at $N=4$, 64 and 256. The in-plane channels are removed quickly, but the out-of-plane channels near the axis, which carry most of the conductance, are protected to high order. On a grid extended to 1.6\,\AA$^{-1}$ the ratio would be 0.39 at $N=256$, because the re-counted extended-zone channels are strongly rotation disordered (Fig.~\ref{fig:phonon}b). For C$_2$ and C$_4$ layers the filter is stronger (full model 0.52 and 0.69 at $N=64$, 0.13 and 0.33 at $N=2048$).

A twisted interface also loses all $\mathbf G\neq0$ components of the interlayer coupling, which an aligned (registered) crystal keeps: above all its shear stiffness, which is zero at an incommensurate interface. This is an $N$-independent factor
\begin{equation}
g_{\rm reg}=\frac{G_{\rm ph}(\text{aligned, continuum coupling})}{G_{\rm ph}(\text{aligned, registered coupling})},
\end{equation}
which multiplies every twisted-stack phonon ratio when the reference is the real aligned crystal. It is a material quantity: the registered crystal's interlayer force constants at all in-plane $\bp$ follow from aligned bulk phonons (T1-d), the continuum ones from the registry scan. The sign is not universal. In the minimal model the registered crystal nests its atoms into hollow sites at a distance comparable to the Morse radius, with a corrugation of 170\,meV per atom, an order of magnitude above van der Waals crystals; its breathing stiffness is less than half the continuum value and $g_{\rm reg}=1.60$ (\texttt{registry7.py}). In van der Waals crystals the breathing stiffness changes little on twisting and the shear branches are lost, which gives $g_{\rm reg}<1$. The large reductions of cross-plane thermal conductivity measured in random-twist films~\cite{Kim2021} are therefore expected to come mainly from $g_{\rm reg}$ and from anharmonic scattering, with the continuum rotational filter adding a thickness-dependent factor on top.

\begin{figure}[!tbp]\centering
\includegraphics[width=\textwidth]{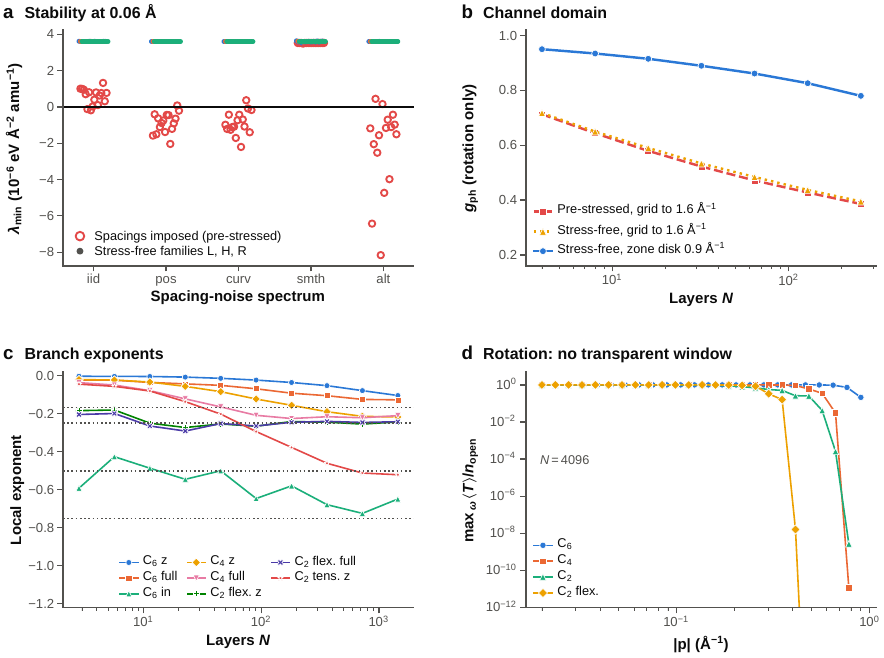}
\caption{The stable phonon model. (a)~Stability of the production chains at 0.06\,\AA: lowest eigenvalue of the smallest channel for each sequence and spacing-noise spectrum (i.i.d., positional, curvature, smooth, alternating), for chains with spacings imposed at a fixed potential (pre-stressed, open markers) and for the stress-free families (filled). (b)~Rotation-only phonon ratio of C$_6$ layers: pre-stressed chains on a grid extended to 1.6\,\AA$^{-1}$, the stress-free model on the same grid, and on the zone-consistent grid; the difference is the channel domain, not the stress. (c)~Local exponents of the branch conductances, C$_2$ (flexural and tensioned), C$_4$ and C$_6$; dotted lines: predicted asymptotes. (d)~Maximum over frequency of the mean transmission per open polarization at $N=4096$ versus channel momentum, full stress-free $3\times3$ model with rotational disorder only (4 sequences): 1 in the protected region around the axis, monotonic decrease beyond it, no isolated window.}\label{fig:phonon}
\end{figure}

%% file: sec_spacing.tex
\section{Spacing noise}\label{sec:spacing}

\subsection{Born exponents and the two sensitivities}
Let $t_l=\bar t(1+\delta_l)$ be the bond of channel $\bp$, with $\delta_l$ a stationary process of power spectrum $S_\delta(q)$. At weak disorder a state of wavenumber $k$ along the stack is backscattered at $2k$ (App.~\ref{app:born}):
\begin{equation}
\gamma_e(k)=\frac{S_\delta(2k)}{2\sin^2k},\qquad \gamma_{\rm ph}(k;\bp)=\frac{S_d(2k)\big[\beta_0F_0\cos k-\beta_pF_p\big]^2}{2F_p^2\sin^2k},\label{eq:gam_e}
\end{equation}
where for electrons $\delta_l$ follows from $\beta_e=-\partial_d\ln|\bar T|$, and for phonons the sum-rule diagonal $F_0$ and the bond $F_p$ vary together along the \emph{equilibrium family} of Eq.~\eqref{eq:family}, with $\beta_0=\beta_{\rm eq}$ and $\beta_p(\bp)=\beta_{\rm eq}+O(p^2)$. The phonon exponent at $\bp=0$ is $S_d(2k)\beta_{\rm eq}^2\tan^2(k/2)/2$: long-wavelength phonons are protected by the sum rule, and the overall scale is $\beta_{\rm eq}^2$. The phonon sensitivity is not the load slope of a fixed potential but a material parameter of its own, and it can be small: for a rigid shift of the interlayer potential $\beta_{\rm eq}=-0.25$\,\AA$^{-1}$ against $\beta_{\rm load}=3.05$\,\AA$^{-1}$ and $\beta_e=1.42$\,\AA$^{-1}$ in the model. The families confirm the scaling: at $\sigma_d=0.03$\,\AA\ (i.i.d., aligned twists, $N=256$) the phonon ratio is 0.39 (L), 0.63 (H) and 0.997 (R), while the electrons are the same in all three, 0.21.

\subsection{Where Born theory fails}\label{sec:anomalies}
\begin{enumerate}
\item \emph{Band centre of pure bond disorder.} At $E=0$, $\ln|\psi|$ performs a random walk with zero drift~\cite{EggarterRiedinger1978,Inui1994}; transfer matrices give $\gamma(0)=1.6\times10^{-4}$ against the Born value $5.0\times10^{-3}$ ($\sigma_\delta=0.1$), and Born is recovered by $|E|/t=0.05$ (Fig.~\ref{fig:spacing}a). Excluded from the theorem (Remark after Corollary~\ref{cor:sum}); irrelevant for band-edge carriers.
\item \emph{Band edge.} Born theory overestimates $\gamma$ within $\varepsilon_a\approx t\sigma_\delta^{4/3}$ of the edge for white noise~\cite{DerridaGardner1984,Deych2000}. For noise correlated beyond the thermal wavelength the scale is the landscape of local band edges, $\varepsilon_a\approx2t\beta_e\sigma_d$ (9--27\,meV at 0.03--0.09\,\AA), comparable to $k_BT$; carriers below it are reflected quasi-classically. This is the one failure of the parameter-free electron prediction (smooth spectrum, Sec.~\ref{sec:predict}). The measured deviation of the Born prediction of $g_e$ for smooth noise ranges from $-11\%$ to $+59\%$ (where $g_e>0.1$) and changes sign with $N$ and $\sigma_d$, so Born is not a bound.
\end{enumerate}

\subsection{Exact bond spectrum}
For $t\propto e^{-\beta d}$ and a Gaussian spacing process $d_l=d_{\rm eq}+\sigma_dx_l$ with correlation $C(m)$, the relative bond fluctuation has the exact spectrum
\begin{equation}
S_\delta(q)=\sum_m\big(e^{s^2C(m)}-1\big)e^{iqm},\qquad s=\beta\sigma_d.\label{eq:lognormal}
\end{equation}
Positional (first-kind) noise is protected only at leading order: $S_\delta(0)=e^{s^2}+2e^{-s^2/2}-3=\frac34s^4+O(s^6)$ (Fig.~\ref{fig:spacing}b). The same formula with $\beta=\beta_{\rm eq}$ gives the lognormal stiffness spectrum of the stress-free phonon model.

\begin{figure}[!tbp]\centering
\includegraphics[width=0.86\textwidth]{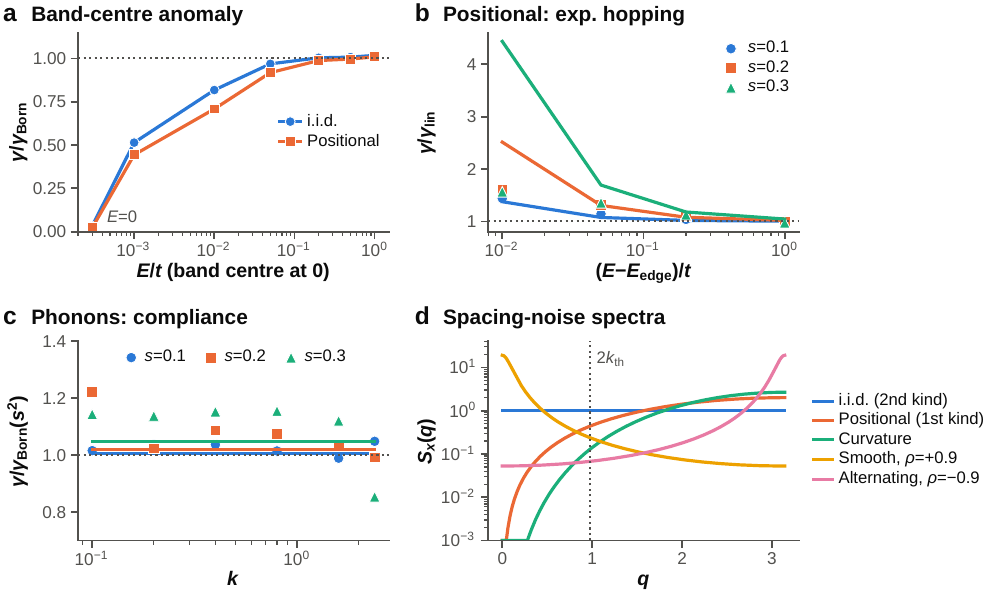}
\caption{Spacing noise beyond leading order. (a)~Band-centre anomaly of bond disorder. (b)~Positional noise near the band edge: the $O(s^4)$ residual of Eq.~\eqref{eq:lognormal}. (c)~Sum-rule chain with lognormal springs. (d)~The five spectra at equal rms; dotted line: $2k_{\rm th}$ of band-edge carriers at 300\,K.}\label{fig:spacing}
\end{figure}

\subsection{Which spectra protect electrons}\label{sec:spectra}
Band-edge electrons probe $S_\delta(q\to0)$ weighted by $1/k^2$; phonons probe all $q$~\cite{IzrailevRev2012}. For $S_\delta\propto q^{2a}$ the band-edge exponent scales as $\gamma_e\propto\varepsilon^{a-1}$: divergent for paracrystalline noise ($a=0$), flat for positional ($a=1$), vanishing for curvature-type ($a=2$). Spectra that vanish at $q\to0$ and keep weight at intermediate $q$ therefore favour phonon scattering. With random twists at $N=256$ (stable model, family L, $\beta_{\rm eq}=\beta_{\rm load}$) the selectivity $S=g_e/g_{\rm ph}$ is
\begin{center}\small
\begin{tabular}{lccccc}
\toprule
$\sigma_d$ (\AA) & i.i.d. & positional & curvature & smooth & alternating\\
\midrule
0.03 & 0.60 & 1.72 & 2.73 & 0.60 & 1.90\\
0.045 & 0.40 & 1.18 & 2.78 & 0.45 & 1.99\\
0.06 & 0.18 & 0.56 & 2.44 & 0.37 & 1.86\\
0.09 & 0.05 & 0.10 & 1.35 & 0.27 & 1.49\\
\bottomrule
\end{tabular}
\end{center}
against 1.28 for rotation alone. For family H ($\beta_{\rm eq}=\beta_{\rm load}/2$) curvature noise gives 1.53--1.71 and alternating noise 1.20--1.35 up to 0.06\,\AA; for family R ($\beta_{\rm eq}=0$) spacing noise lowers $S$ for every spectrum, because only the electrons respond. A paracrystal is phonon selective only if $\beta_{\rm eq}/\beta_e$ is large enough and its spectrum spares the band edge; Sec.~\ref{sec:thermo_cond} gives the thresholds.

\paragraph{Which spectra occur in practice?} Random twists make twist-driven spacings i.i.d.\ (a paracrystal of the second kind~\cite{Hosemann}): with uniform independent $\theta_l$, the steps $\Delta\theta_l$ mod $2\pi/n$ are i.i.d. Suppressing $S_\delta(q\to0)$ needs a reference, such as a substrate or intercalant template (positional) or compensating layer sequences (anti-correlated). The $00l$ diffraction line shapes identify the class: second-kind noise broadens lines $\propto l^2$, first-kind noise gives a Debye--Waller intensity loss.

\subsection{The spacing-noise transparency surface}\label{sec:joint_asym}
The bond $F_p(d_l)$ and the sum-rule diagonal $F_0(d_{l-1})+F_0(d_l)$ of an out-of-plane phonon respond to the same spacing, and in Eq.~\eqref{eq:gam_e} their backscattering amplitudes enter with opposite signs. In the scalar out-of-plane ($zz$) branch they cancel where
\begin{equation}
\cos k^*(\bp)=\frac{\beta_pF_p}{\beta_0F_0}=\frac{\partial_dF_{zz}(\bp;d)}{\partial_dF_{zz}(0;d)}\bigg|_{\rm family}.\label{eq:cone}
\end{equation}
Eq.~\eqref{eq:cone} is the scalar $z$-branch result. In the full $3\times3$ chain the normal--in-plane coupling enters only at higher order. Along the equilibrium family every interface is force free, so $F_{iz}=-ip_i[\partial_dU_l(\bp;d_l)-\partial_dU_l(0;d_l)]=O(p^3)$, and so is its variation along the family, $D_dF_{iz}=O(p^3)$ (local slopes 3.000 for families L, H and R; \texttt{models/\allowbreak checks10.py}). The partial derivative at a fixed potential, $\partial_dF_{iz}=-ip_i\partial_d^2U=O(p)$, describes a load, not the family, and is not the relevant quantity. The $O(p^3)$ terms and the mixing of branches move the cancellation slightly; an independent full-Born check at five momenta found it within 0.07\,rad/ps of Eq.~\eqref{eq:cone}, well inside the frequency grid (0.44\,rad/ps). At $\bp=0$ this is the acoustic point $k^*=0$. For $\bp\neq0$ the ratio is below 1 in families L and H and for the rigid shift, so every channel has a frequency $\omega^*(\bp)$ inside its out-of-plane band at which spacing noise does not backscatter at Born level: $k^*=3.9^\circ$, $12^\circ$ and $41^\circ$ at $|\bp|=0.1$, 0.3 and 0.9\,\AA$^{-1}$ in family L. (In family R, $\partial_dF_{zz}(0)=0$ and the zero lies instead at the momentum where $\partial_dF_{zz}(\bp)$ changes sign, near 0.36\,\AA$^{-1}$.) This is the phonon counterpart of the transparent energies of correlated diagonal and off-diagonal disorder~\cite{Tessieri2001}, and a case of Corollary~\ref{cor:transparent}: $\Zg$ of spacing noise has a component of codimension 1. The paired runs show it without any fit. For aligned stacks with i.i.d.\ noise of 0.06\,\AA\ (family L, $N=256$), every sampled channel (every third) with $|\bp|\ge0.56$\,\AA$^{-1}$ transmits 0.78--0.96 at a frequency within 0.22\,rad/ps of $\omega^*$ (frequency grid 0.44\,rad/ps), while its band bottom transmits as little as zero; below 0.56\,\AA$^{-1}$ the peaks lie within 0.32\,rad/ps of $\omega^*$, but the band bottom transmits nearly as well (\texttt{transparency8.py}). At four sample momenta the peaks are 0.97, 0.93, 0.93 and 0.80 at 9.1, 10.0, 10.9 and 11.8\,rad/ps, against $\omega^*=9.06$, 10.13, 11.03 and 11.87\,rad/ps from Eq.~\eqref{eq:cone}.

\begin{figure}[!tbp]\centering
\includegraphics[width=0.5\textwidth]{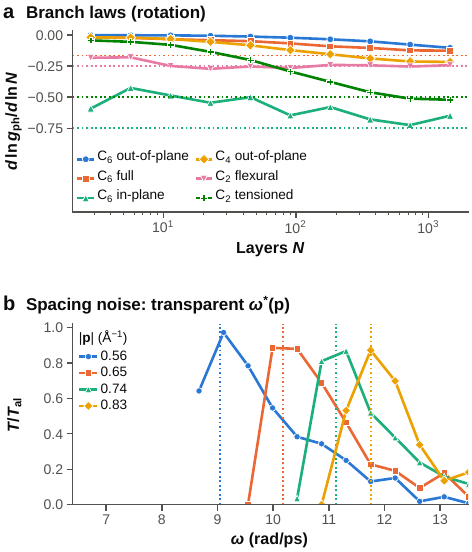}
\caption{Phonons in the stable model. (a)~Local exponents of branch conductances under rotation (four sequences, $N\le2048$; Fig.~\ref{fig:phonon}c); dotted: $-1/\kappa_z$, $-1/2$ (tensioned C$_2$) and $-3/4$. (b)~Spacing noise alone (aligned stack, i.i.d.\ noise of 0.06\,\AA, $N=256$): transmission of four channels relative to the noiseless stack, $T/T_{\rm al}$, against frequency; dotted: predicted transparent frequencies $\omega^*(\bp)$ of Eq.~\eqref{eq:cone}.}\label{fig:transparent}
\end{figure}
Three consequences follow.
\begin{enumerate}
\item \emph{Aligned stacks.} At Born level spacing noise gives $g_{\rm ph}\propto N^{-1/2}$ (Corollary~\ref{cor:transparent} with $s=0$). Terms of order $\sigma_d^4$ lift the surface, so this law holds at intermediate $N$. (The surface ends in a vertex at the acoustic point, where the tube hypothesis H1 fails; the vertex is a set of lower dimension and does not change the leading law.) Our runs are pre-asymptotic (local exponents between $-0.1$ and $-0.65$ at $N=128$--256, depending on spectrum and $\sigma_d$) and do not test it.
\item \emph{Random twists.} Rotation adds $\gamma_{\rm rot}>0$ at every $\bp\neq0$ and lifts the surface; the zero set shrinks to the acoustic point, but a valley of small $\gamma$ follows $\omega^*(\bp)$, lifted near the axis only as the rotational disorder $\propto p^{2\kappa_z}$. The asymptotic law of the joint problem depends on how the two disorders share this valley. It is not derived here, and all numerical statements about joint disorder are finite-$N$ held-out predictions (Sec.~\ref{sec:predict}).
\item \emph{What is unaffected.} The held-out predictions are per channel, and the measured spacing exponents contain the surface. An interpolation in $\beta_{\rm eq}$ of $\sqrt s$, linear between $\beta_{\rm eq}=0$ and $\beta_{\rm load}$, assumes that the bracket of Eq.~\eqref{eq:gam_e} keeps its sign; it fails near the surface and can flip threshold decisions there. We interpolate the signed amplitude instead and test it out of sample (Sec.~\ref{sec:betainterp}).
\end{enumerate}

%% file: sec_predict.tex
\section{Held-out predictions of joint disorder}\label{sec:predict}

Every channel $(\bp,k)$ receives a Lyapunov exponent from each independent disorder, and exponents of independent disorders add. The channel's mean transmission is $\Phi_1(2N\gamma)$ and the measure performs the channel sum. We test this against paired simulations of the stable model: triangular layers; $\Gamma$-valley electrons ($\bar T=0.106$\,eV at $d_{\rm eq}$, $\beta_e=1.42$\,\AA$^{-1}$); stress-free $3\times3$ phonons in five stiffness families ($\beta_{\rm eq}=0$, $\beta_{\rm load}/4$, $\beta_{\rm load}/2$, $3\beta_{\rm load}/4$, $\beta_{\rm load}$: R, Q$_1$, H, Q$_3$, L); five spacing spectra at 0.03, 0.045, 0.06 and 0.09\,\AA; aligned or random twists; 16 paired sequences each, $N=4$--256 (\texttt{pair\_v7.py}, \texttt{pair8.py}). Families share their random streams: sequence $i$ of a configuration has the same twists and unit spacing noise in every family.

\subsection{Electrons: no fitted parameter}
Born exponents with the exact lognormal spectrum, Eq.~\eqref{eq:lognormal}, plus the rotation variance, give $g_e(N)$ with no input from the three-dimensional runs. Where $g_e>0.1$ the prediction deviates by rms 2.1\% (alternating), 5.0\% (positional), 5.3\% (curvature) and 9.5\% (i.i.d.) over all $N$ and $\sigma_d$; for $N\ge32$ the rms deviation is 2.4, 5.2, 4.8 and 7.2\% at 0.03, 0.045, 0.06 and 0.09\,\AA. The largest errors are at $N\le16$, where single-parameter scaling is not yet reached, and grow with $\sigma_d$ (i.i.d.\ at 0.09\,\AA: $-23\%$ at $N=8$). Over all spectra but the smooth one the rms deviation is 5.8\%, against a sampling floor of 3.1\% from the 16 sequences (bootstrap, \texttt{boot7.py}). The smooth spectrum fails for the stated reason (Sec.~\ref{sec:anomalies}), with errors from $-11\%$ to $+59\%$. Rotation alone leaves the model electrons at $g_e=1.000$ ($\kappa=6$ at thermal momenta), so this test probes spacing noise; rotation enters the electron laws through the tests of Secs.~\ref{sec:index} and~\ref{sec:theorem}.

\subsection{Phonons: separate runs predict the joint runs}
For phonons we calibrate on separate runs: the per-$(\bp,\omega)$ Lyapunov variables $s=\Phi_1^{-1}(\mathcal T/\mathcal T_\al)$ of rotation-only and spacing-only stacks. The joint channel is then predicted as $\Phi_1(s_{\rm rot}+s_{\rm sp})$. Four tests, each on data not used for calibration, are summarized in Table~\ref{tab:holdout} (family L has the strongest phonon response).
\begin{table}[!htbp]\centering\small
\caption{Held-out prediction of the joint (rotation $+$ spacing) phonon conductance: relative error, rms (max) over all spectra, noise levels and $N$ with $g_{\rm ph}>0.05$, and the sampling floor of the rms from 300 bootstrap draws over calibration and test sequences (\texttt{boot7.py}).}\label{tab:holdout}
\begin{tabular}{L{5.4cm}cccc}
\toprule
 & \multicolumn{2}{c}{family L ($\beta_{\rm eq}=\beta_{\rm load}$)} & \multicolumn{2}{c}{family H ($\beta_{\rm load}/2$)}\\
test & rms (max) & floor & rms (max) & floor\\
\midrule
additivity, same $\sigma_d$ and $N$ & 3.9\% (17.5\%) & 3.8\% & 2.0\% (8.2\%) & 2.0\%\\
per-channel product $\tau_{\rm rot}\tau_{\rm sp}$ (baseline) & 3.8\% (16.5\%) & & 2.0\% (7.6\%) & \\
product of channel sums $g_{\rm rot}g_{\rm sp}$ (baseline) & 6.0\% (14.4\%) & & 3.8\% (12.4\%) & \\
$\sigma_d$ hold-out inside 0.03--0.06\,\AA\ ($s\propto\sigma_d^2$) & 3.3--6.4\% & 2.1--3.6\% & 1.4--2.4\% & 1.0--1.9\%\\
$\sigma_d$ hold-out, 0.06 $\to$ 0.09\,\AA & 7.6\% (17.6\%) & 5.6\% & 2.6\% (6.6\%) & 2.9\%\\
$\sigma_d$ hold-out, 0.03 $\to$ 0.09\,\AA & 14.0\% (27.0\%) & 5.4\% & 5.1\% (10.2\%) & 2.8\%\\
$N$ hold-out, $\gamma$ from $N_{\rm cal}=16$ only & 6.1\% (20.3\%) & 5.0\% & 3.6\% (17.8\%) & 2.9\%\\
\hspace{1em}$N_{\rm cal}=32$ & 4.2\% (12.3\%) & 4.2\% & 2.3\% (7.2\%) & 2.3\%\\
\hspace{1em}$N_{\rm cal}=64$ & 4.6\% (19.0\%) & 3.4\% & 2.2\% (7.0\%) & 1.9\%\\
$\beta_{\rm eq}$ hold-out: family H from families R and L ($\sqrt s$ rule; Table~\ref{tab:beta}) & & & 4.1\% (15.6\%) & 2.1\%\\
\bottomrule
\end{tabular}
\end{table}

\paragraph{Readings.} (i)~The additivity of Lyapunov exponents holds to a few percent. (ii)~Resolving channels matters little: in the zone-consistent grid rotation removes few channels, and the product of channel sums is within 6\%. (iii)~Born scaling $s\propto\sigma_d^2$ transfers the calibration across a factor 2 in $\sigma_d$ within 3--6\%. Extrapolating by a factor 3 to 0.09\,\AA\ costs accuracy in family L, where $\beta_{\rm eq}\sigma_d=0.27$ and the lognormal and beyond-Born terms grow. (iv)~One thickness suffices: exponents extracted at a single $N_{\rm cal}$ predict $N=4$--256 within 4.2--6.1\% rms ($N_{\rm cal}=16$, 32, 64; all three are reported, and none is selected). (v)~The material parameter $\beta_{\rm eq}$ needs an interpolation between families, the subject of Sec.~\ref{sec:betainterp}.

\paragraph{Uncertainty.} Resampling the calibration sequences and the test sequences (300 bootstrap draws) gives the sampling floor of each rms error in Table~\ref{tab:holdout}. Where calibration draws come from several families (the $\beta_{\rm eq}$ tests), one index vector per configuration is shared by all families, because they share random streams (\texttt{boot8.py}); configurations that do not share streams are resampled independently. Pairing changes the floors by at most 0.3 percentage points for the square-root and signed rules; for the quadratic rule, which takes a second difference of the families, it halves the floor on family Q$_1$ (2.48\% unpaired, 1.26\% paired: the shared noise cancels) and raises it slightly on Q$_3$ (2.62\% to 2.94\%); per point, the 68\% bands of the additivity error at $N=64$ and 256 are $[-4,+8]\%$ and $[-3,+8]\%$ (family L, i.i.d.\ and positional, 0.03 and 0.06\,\AA) and within $\pm4\%$ for family H. Inside the simulated range the rms errors are 1.0--2.0 times the floor: much of the deviation is the sampling noise of 16 sequences, and the systematic part is a few percent. The extrapolation from 0.03 to 0.09\,\AA\ in family L (14.0\% against 5.4\%) is clearly systematic; it is the strong-disorder end, where $\beta_{\rm eq}\sigma_d=0.27$.

\subsection{Interpolation in the stiffness slope \texorpdfstring{$\beta_{\rm eq}$}{beta\_eq}}\label{sec:betainterp}
The thermoelectric window (Sec.~\ref{sec:thermo_cond}) is a function of $\beta_{\rm eq}$, a material parameter, and is mapped by interpolating each channel's spacing variable $s(\beta_{\rm eq})$ between simulated families. At Born order the spacing amplitude of a channel is linear in $\beta_{\rm eq}$ along the equilibrium family (Eq.~\eqref{eq:gam_e}), so $s=A^2$ with $A$ linear, and $A$ can change sign between families. Taking $\sqrt s$ linear between R and L cannot follow a sign change and flips threshold decisions near it. The natural replacement, $s$ quadratic through R, H and L, is exact for Born-order exponents, but the exponents extracted from finite stacks at the strongest noise ($\beta_{\rm eq}\sigma_d$ up to 0.27) are not of Born order, and a quadratic through three points amplifies the difference between R and H.

\emph{Selection.} We therefore compared five rules on the two families Q$_1$ and Q$_3$, whose joint runs had not been used for any calibration: $\sqrt s$ linear between R and L, $s$ quadratic through R, H, L, the signed amplitude linear between R and L, $\sqrt s$ piecewise linear through R, H, L, and the signed amplitude piecewise linear through R, H, L, with the sign of each channel's amplitude chosen to make the calibration amplitudes most nearly collinear (\texttt{interp8\_test.py}). The last rule was best, and we adopt it (Table~\ref{tab:beta}). The quadratic rule fails on family Q$_1$ (10.9\% rms, 53\% at 0.09\,\AA), exactly where the strongest noise meets a sign change of the amplitude.

\emph{Out-of-sample test.} Because Q$_1$ and Q$_3$ selected the rule, their errors are validation errors, not test errors. We then ran the spacing-only calibration configurations of Q$_1$ and Q$_3$ as well (\texttt{pair8.py --aligned-only}), so that every family can serve as a calibration family, and tested the fixed rule on families it was not calibrated on (\texttt{loo8.py}): each interior family from the other four (leave one out), and pairs of families from only three. For $K$ calibration families the signs are chosen per channel to minimize the squared deviations of the interior amplitudes from the chords of their neighbours; for $K=3$ this is the rule above.
\begin{table}[!htbp]\centering\small
\caption{Interpolation in $\beta_{\rm eq}$: rms (max) relative error of the predicted phonon ratio of the joint runs (all five spectra, four noise levels, $N=4$--256, $g_{\rm ph}>0.05$) and the paired-bootstrap sampling floor (\texttt{boot8.py}, \texttt{loo8.py}).}\label{tab:beta}
\begin{tabular}{L{4.4cm}lccc}
\toprule
rule & calibration & target & rms (max) & floor\\
\midrule
$\sqrt s$ linear & R, L & H & 4.11\% (15.6\%) & 2.06\%\\
 & R, L & Q$_1$ & 4.68\% (24.6\%) & 1.11\%\\
 & R, L & Q$_3$ & 3.50\% (14.9\%) & 2.97\%\\
$s$ quadratic & R, H, L & Q$_1$ & 10.91\% (52.8\%) & 1.26\%\\
 & R, H, L & Q$_3$ & 2.97\% (13.4\%) & 2.94\%\\
signed piecewise (selection) & R, H, L & Q$_1$ & 1.39\% (6.2\%) & 0.98\%\\
 & R, H, L & Q$_3$ & 2.96\% (13.6\%) & 2.92\%\\
\midrule
signed piecewise (leave one out) & R, Q$_1$, Q$_3$, L & H & 2.17\% (8.9\%) & 1.97\%\\
 & R, H, Q$_3$, L & Q$_1$ & 1.41\% (6.3\%) & 0.98\%\\
 & R, Q$_1$, H, L & Q$_3$ & 2.96\% (13.6\%) & 2.93\%\\
signed piecewise (three families) & R, Q$_1$, L & H & 2.99\% (11.3\%) & 2.00\%\\
 & R, Q$_1$, L & Q$_3$ & 3.20\% (14.3\%) & 2.89\%\\
 & R, Q$_3$, L & Q$_1$ & 2.70\% (15.4\%) & 1.06\%\\
 & R, Q$_3$, L & H & 2.44\% (9.6\%) & 2.01\%\\
\bottomrule
\end{tabular}
\end{table}

\paragraph{Readings.} (i)~With calibration families $\beta_{\rm load}/2$ apart around the target, the rule predicts the left-out family at or near the sampling floor: 2.17\% against 1.97\% for H, 1.41\% against 0.98\% for Q$_1$, and 2.96\% against 2.93\% for Q$_3$. For Q$_1$ and Q$_3$ this nearly repeats the selection test, because the piecewise rule uses only the two neighbouring families (the fourth family changes a few sign choices). The prediction of H is out of sample in the sense that H was used neither to select nor to calibrate the rule; it is not a blind test, because the joint runs of H were known when the rule was selected, and it is not an extrapolation to another material. On the optimal-ZT ratio ($B=0.3$, 1; $N=128$, 256; four spectra, four noise levels) the leave-one-out errors are 1.9\%, 1.6\% and 1.9\% rms, with maxima of 4.7--5.6\%. (ii)~The error grows with the distance to the calibration families. With three families and gaps of $3\beta_{\rm load}/4$ around the target it is 2.4--3.2\% rms on $g_{\rm ph}$ and 2.2--3.9\% on the ZT ratio; the largest error is for Q$_1$ from R, Q$_3$ and L, where nine near-threshold discriminants flip (\texttt{loo8.log}). (iii)~The five-family map used for the thresholds (Sec.~\ref{sec:thermo_cond}) has gaps of $\beta_{\rm load}/4$. We use the leave-one-out rms of the ZT ratio, 1.8\%, as a sensitivity scale of this map, to widen the threshold intervals; it is an empirical scale from gaps twice as large, not an upper bound on the error. (iv)~The $\sqrt s$-linear rule is worse on every family (3.5--4.7\% rms, maxima 15--25\%); its largest systematic error is on Q$_1$ (4.7\% against a floor of 1.1\%).

\begin{figure}[!tbp]\centering
\includegraphics[width=\textwidth]{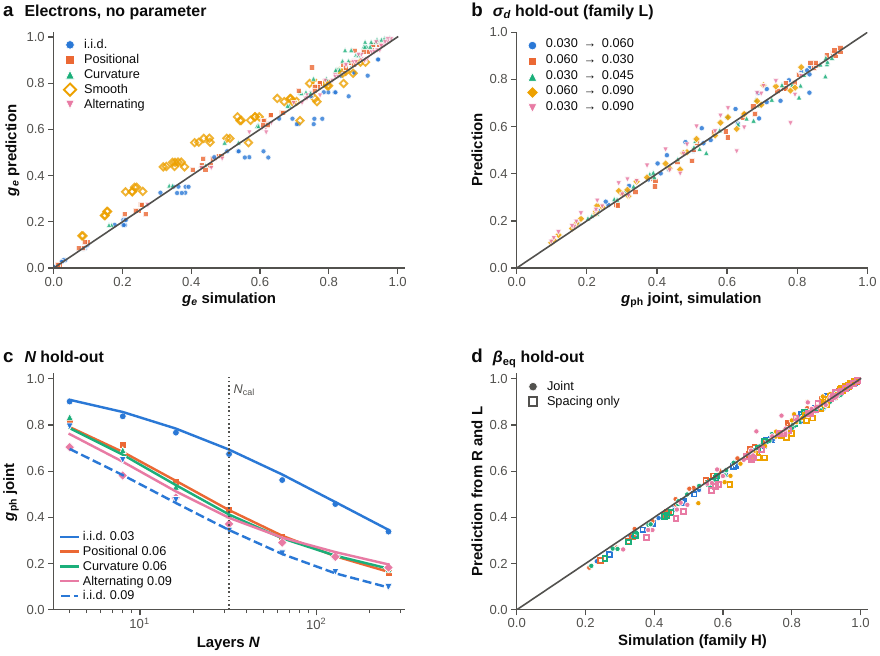}
\caption{Held-out predictions. (a)~Electrons: parameter-free prediction against simulation for all spectra, noise levels, twist states and $N$ (smooth spectrum open markers). (b)~Phonons, $\sigma_d$ hold-out: joint conductance predicted from spacing-only runs at another noise level. (c)~$N$ hold-out: curves predicted from exponents extracted at $N=32$ (lines) against the joint simulations (markers). (d)~$\beta_{\rm eq}$ hold-out with the $\sqrt s$-linear rule: family H predicted from families R and L. The tests of Table~\ref{tab:beta} with the signed rule are shown in Fig.~\ref{fig:outofsample}(a).}\label{fig:predict}
\end{figure}

%% file: sec_thermo.tex
\section{Thermoelectrics with Fermi--Dirac statistics}\label{sec:thermo}

All coefficients follow from the Onsager integrals $L_n(\mu)$ of the channel transmissions with the Fermi--Dirac window (App.~\ref{app:thermo}), at 300\,K and for chemical potentials from $-0.30$ to $+0.25$\,eV relative to the aligned band edge, i.e.\ from non-degenerate to degenerate doping (\texttt{thermo\_v7.py}, \texttt{thermo8.py}, \texttt{loo8.py}). A narrower window, $-0.15$ to $+0.10$\,eV, would change no optimal-ZT ratio at $B=1$ and 3, but up to 0.84\% at $B=0.3$ and 4.0\% at $B=0.1$ (strong i.i.d.\ noise, whose optimum lies above $+0.10$\,eV). All optima lie between $-70$ and $+125$\,meV, inside the window and at least 11\,$k_BT$ below the top of the electron energy grid; extending the window to $-0.45$\,eV changes nothing (\texttt{mu\_window8.py}). The electron transmissions are those of the paired runs of Sec.~\ref{sec:predict}; the phonon conductances are those of the five stiffness families. Two kinds of statement result. The Seebeck change is a parameter-free prediction. The ZT ratio is conditional on three material numbers, $\beta_{\rm eq}/\beta_e$, the quality factor $B$ and the registry factors, together with the spacing spectrum.

\subsection{Seebeck change: a parameter-free prediction}
The Born--$\Phi_1$ transmission of Sec.~\ref{sec:predict} gives every $L_n$, and therefore the Seebeck change $\Delta S(\mu)=S_{\rm stack}-S_\al$ at fixed $\mu$, with no input from the three-dimensional runs. Over the four spectra other than the smooth one, four noise levels, both twist states and all seven thicknesses (224 points per $\mu$), the rms deviation from the simulation is 0.059, 0.058, 0.051 and 0.039\,$k_B/e$ at $\mu=-0.10$, $-0.05$, 0 and $+0.05$\,eV (Fig.~\ref{fig:thermo}a), against sampling floors of 0.032, 0.031, 0.028 and 0.023\,$k_B/e$ from the 16 sequences (\texttt{boot7.py}).
\begin{itemize}
\item \emph{Sign by spectrum.} Paracrystalline (i.i.d.)\ noise raises $S$, by up to 2.4\,$k_B/e$ (0.09\,\AA, $N=256$), because its band-edge exponent diverges and the surviving carriers are the hotter ones. Curvature-type noise lowers $S$, by up to 0.5\,$k_B/e$, because it protects the band edge better than higher energies ($\gamma_e\propto\varepsilon^{a-1}$ with $a=2$). Positional noise ($a=1$) changes $S$ by at most 0.15\,$k_B/e$ up to 0.06\,\AA.
\item \emph{Degeneracy.} The change shrinks as the Fermi window moves into the band: for i.i.d.\ noise at 0.06\,\AA\ and $N=256$ it falls from 1.60 to 1.10\,$k_B/e$ between $\mu=-0.10$ and $+0.05$\,eV.
\item \emph{Smooth spectrum.} The prediction falls short by about 30\% (e.g.\ 0.63 against 0.88\,$k_B/e$ at 0.06\,\AA, $N=64$). This is the beyond-Born band-edge landscape of Sec.~\ref{sec:anomalies}: carriers below the local band edges are reflected quasi-classically, which raises $S$ further.
\end{itemize}

\paragraph{Fixed density.} At a fixed carrier density the disordered stack has a lower chemical potential, because disorder creates tail states below the band edge. From the eigenvalues of the finite stacks (localized tail included) the shift is $-0.4$ to $-0.5$, $-3.1$ to $-3.7$ and $-4.0$ to $-5.3$\,meV for i.i.d.\ noise of 0.03, 0.06 and 0.09\,\AA\ with random twists, $-1.5$ to $-1.6$\,meV for positional and $-1.2$ to $-3.4$\,meV for smooth noise at 0.06\,\AA, and zero for rotation alone. At fixed density $\Delta S$ is therefore larger than at fixed $\mu$ by at most 0.21\,$k_B/e$ (non-degenerate limit). The optimal-ZT comparison below does not depend on this choice, because each stack is optimized over $\mu$.

\subsection{Optimal ZT}
With the quality factor $B=\sigma_\al T(k_B/e)^2/\kappa_{{\rm ph},\al}$ of the aligned stack, $ZT(\mu)$ is Eq.~\eqref{eq:ztB}. It includes the electronic thermal conductance exactly. The aligned and the disordered stack are each taken at their own optimal $\mu$. For reference, the model's aligned stack has $ZT^*_\al=0.46$, 1.08, 2.37 and 4.29 at $B=0.1$, 0.3, 1 and 3, with the optimum moving from the band edge to $-55$\,meV. These values are not material predictions: the kit reports only the ratio, for a list of $B$. Table~\ref{tab:zt} gives the ratio at $B=1$.

\begin{table}[!htbp]\centering\small
\caption{Optimal-ZT ratio $ZT^*_{\rm stack}/ZT^*_\al$ for random twists at $B=1$; entries for $N=64/128/256$. Rotation alone gives 1.09/1.12/1.15 in every family.}\label{tab:zt}
\begin{tabular}{llccc}
\toprule
spectrum & $\sigma_d$ (\AA) & L ($\beta_{\rm eq}=\beta_{\rm load}$) & H ($\beta_{\rm load}/2$) & R (0)\\
\midrule
i.i.d. & 0.03 & 1.14/1.15/1.13 & 0.99/0.96/0.89 & 0.89/0.83/0.70\\
 & 0.06 & 1.03/1.00/0.84 & 0.81/0.73/0.57 & 0.64/0.52/0.34\\
 & 0.09 & 0.98/0.82/0.64 & 0.70/0.55/0.38 & 0.48/0.32/0.17\\
positional & 0.03 & 1.30/1.36/1.36 & 1.09/1.08/1.03 & 0.97/0.89/0.75\\
 & 0.06 & 1.24/1.09/0.76 & 0.94/0.80/0.53 & 0.68/0.49/0.25\\
 & 0.09 & 1.03/0.69/0.30 & 0.74/0.46/0.18 & 0.44/0.21/0.06\\
curvature & 0.03 & 1.40/1.54/1.70 & 1.16/1.22/1.30 & 1.04/1.01/0.97\\
 & 0.06 & 1.61/1.67/1.59 & 1.22/1.25/1.21 & 0.87/0.77/0.61\\
 & 0.09 & 1.62/1.45/1.12 & 1.19/1.04/0.74 & 0.71/0.51/0.27\\
alternating & 0.03 & 1.31/1.41/1.56 & 1.16/1.21/1.30 & 1.06/1.07/1.07\\
 & 0.06 & 1.54/1.64/1.70 & 1.24/1.30/1.35 & 1.01/0.97/0.91\\
 & 0.09 & 1.68/1.74/1.65 & 1.29/1.33/1.24 & 0.92/0.84/0.70\\
\bottomrule
\end{tabular}
\end{table}

\paragraph{Readings.} (i)~Rotation alone raises the optimal ZT by 7--23\% over $B=0.1$--3 and $N=64$--256: in the zone-consistent model the continuum filter is only mildly phonon selective. (ii)~Spacing noise adds most when the phonons respond to it ($\beta_{\rm eq}$) and the spectrum spares the band edge. Curvature and alternating spectra give the largest ratios, up to 1.74 at $B=1$ and 2.1 at $B=0.1$ in family L. i.i.d.\ noise adds at most 13\% to rotation alone, and only in family L at small $B$ or $N$; otherwise it lowers the ratio. (iii)~The gains shrink as $B$ grows, because the electronic thermal conductance then dominates the denominator: family L, curvature 0.03\,\AA, $N=256$ gives 2.01, 1.85, 1.70 and 1.59 at $B=0.1$, 0.3, 1 and 3.

\paragraph{Boundaries inside the simulated range.} For i.i.d.\ noise the ratio crosses 1 at $\sigma_d^*=0.080$, 0.060 and 0.047\,\AA\ (family L, $B=1$, $N=64$, 128, 256), 0.028--0.022\,\AA\ (H) and 0.020--0.016\,\AA\ (R). At $B=0.3$ the family-L values are $>0.09$, 0.076 and 0.055\,\AA. All are interpolated in $\sigma_d^2$ between simulated levels.

\paragraph{Numerical checks.} (i)~Bose--Einstein instead of classical weighting of the phonon channels at 300\,K ($\hbar\omega/k_BT\le0.54$ in the model) changes the phonon ratio by at most 0.54\% (rms 0.29\%) and the optimal-ZT ratio by at most 0.46\% over 340 ratios (\texttt{thermo8\_sens.py}). (ii)~The chemical-potential window is discussed above. (iii)~An independent spot check with a refined frequency grid changed one sample by about 1.5\%; we did not repeat it and count it as part of the sensitivity scale of the threshold intervals, which is of the same size.

\subsection{The material condition}\label{sec:thermo_cond}
How large must $\beta_{\rm eq}$ be? The electron data of a configuration do not depend on $\beta_{\rm eq}$; the phonon ratio does, through each channel's spacing variable, which is interpolated between the five simulated families by the signed piecewise rule of Sec.~\ref{sec:betainterp} and combined with the rotation-only variable. This gives the optimal-ZT ratio for any $0\le\beta_{\rm eq}\le\beta_{\rm load}$ (\texttt{loo8.py}). Two thresholds follow for each spectrum and noise level: break-even with the aligned stack, and net gain over rotation alone. Each is reported as follows (Table~\ref{tab:thr}).
\begin{itemize}
\item \emph{Point value} from the five-family map; ``$>2.14$'' if the level is not reached for $\beta_{\rm eq}\le\beta_{\rm load}$ (right-censored).
\item \emph{Interval}: the hull of the empirical 16th--84th percentiles of 200 bootstrap draws and of the crossings of the ratio curve scaled by $1\pm0.018$ (the sensitivity scale of the map, the leave-one-out rms of Sec.~\ref{sec:betainterp}). A draw that does not reach the level within the simulated range counts as $+\infty$; the percentiles are order statistics, not interpolations, so censored draws never produce an undefined endpoint. Configurations are resampled at the level at which they are independent: families that share random streams with a common index vector, the random-twist electron samples, the spacing-only calibration runs and the rotation-only reference each with their own indices.
\item $P_{\rm nr}$: the fraction of draws that do not reach the level.
\item \emph{Bracket}: the two simulated families between which the directly computed ratio crosses the level. It comes from the same simulated families as the map, so overlap of interval and bracket is an internal consistency check, not an independent coverage test.
\end{itemize}
Over $B=0.3$, 1, $N=128$, 256, four spectra, four noise levels and both thresholds (128 values), 37 point values are not reached within the simulated range; $P_{\rm nr}=0$ for 84 thresholds, 1 for 24, and in between for 20. All 128 intervals overlap their brackets. With the three families R, H and L alone (\texttt{thermo8.py}), 127 do; the exception is alternating noise of 0.09\,\AA\ at $B=1$, $N=256$, which breaks even at 0.50 [0.45, 0.53] against the bracket (0.54, 1.07). Where both maps reach the level, the five-family values differ from the three-family ones by at most 0.07 (median 0.003); every threshold is censored in both or in neither.

\begin{table}[!htbp]\centering\footnotesize
\caption{Thresholds of $\beta_{\rm eq}/\beta_e$ for $B=1$, $N=256$ ($\beta_e=1.42$\,\AA$^{-1}$), from the five-family interpolation: point value, right-censored bootstrap interval, and, in parentheses, the bracket from the directly simulated families (at $\beta_{\rm eq}/\beta_e=0$, 0.54, 1.07, 1.61, 2.14). ``$>2.14$'': not reached for $\beta_{\rm eq}\le\beta_{\rm load}$ (the interval is omitted when its lower end is not reached either, i.e.\ when more than 84\% of the draws do not reach the level); bracket (0.00): above the level already at $\beta_{\rm eq}=0$. Upper rows: spacing noise adds to the rotational gain; lower rows: the optimal ZT exceeds the aligned value. Fraction of draws not reaching the level where it is neither 0 nor 1: gain, curvature 0.09\,\AA: 0.83, positional 0.045\,\AA: 0.98, i.i.d.\ 0.03\,\AA: 0.78; break-even, curvature 0.09\,\AA: 0.02, i.i.d.\ 0.045\,\AA: 0.18.}\label{tab:thr}
\begin{tabular}{L{1.9cm}C{3.1cm}C{3.1cm}C{3.1cm}C{3.1cm}}
\toprule
spectrum & 0.03\,\AA & 0.045\,\AA & 0.06\,\AA & 0.09\,\AA\\
\midrule
\multicolumn{5}{l}{\emph{gain over rotation alone}}\\
alternating & 0.52 [0.43, 0.61]\newline (0.00, 0.54) & 0.56 [0.50, 0.62]\newline (0.54, 1.07) & 0.62 [0.58, 0.66]\newline (0.54, 1.07) & 0.86 [0.81, 0.91]\newline (0.54, 1.07)\\
curvature & 0.71 [0.66, 0.76]\newline (0.54, 1.07) & 0.82 [0.78, 0.85]\newline (0.54, 1.07) & 0.99 [0.94, 1.03]\newline (0.54, 1.07) & ${>}2.14$ [2.13, ${>}2.14$]\newline (2.14, ${>}2.14$)\\
positional & 1.45 [1.38, 1.52]\newline (1.07, 1.61) & ${>}2.14$\newline (2.14, ${>}2.14$) & ${>}2.14$\newline (2.14, ${>}2.14$) & ${>}2.14$\newline (2.14, ${>}2.14$)\\
i.i.d. & ${>}2.14$ [2.12, ${>}2.14$]\newline (2.14, ${>}2.14$) & ${>}2.14$\newline (2.14, ${>}2.14$) & ${>}2.14$\newline (2.14, ${>}2.14$) & ${>}2.14$\newline (2.14, ${>}2.14$)\\
\midrule
\multicolumn{5}{l}{\emph{break-even with the aligned stack}}\\
alternating & 0.00 [0.00, 0.00]\newline (0.00) & 0.04 [0.00, 0.15]\newline (0.00, 0.54) & 0.29 [0.25, 0.33]\newline (0.00, 0.54) & 0.57 [0.55, 0.60]\newline (0.54, 1.07)\\
curvature & 0.27 [0.18, 0.34]\newline (0.00, 0.54) & 0.56 [0.53, 0.59]\newline (0.54, 1.07) & 0.72 [0.68, 0.75]\newline (0.54, 1.07) & 1.84 [1.73, 1.98]\newline (1.61, 2.14)\\
positional & 0.98 [0.93, 1.03]\newline (0.54, 1.07) & 1.76 [1.66, 1.84]\newline (1.61, 2.14) & ${>}2.14$\newline (2.14, ${>}2.14$) & ${>}2.14$\newline (2.14, ${>}2.14$)\\
i.i.d. & 1.56 [1.48, 1.64]\newline (1.07, 1.61) & 2.01 [1.88, ${>}2.14$]\newline (1.61, 2.14) & ${>}2.14$\newline (2.14, ${>}2.14$) & ${>}2.14$\newline (2.14, ${>}2.14$)\\
\bottomrule
\end{tabular}
\end{table}

The thermoelectric statement is therefore conditional, with named inputs. Given the spectrum class and $\sigma_d$ (from $00l$ line shapes, Sec.~\ref{sec:spectra}), $\beta_e$ and $\beta_{\rm eq}$ (T1-a, T1-b) and $B$ of the aligned crystal, the optimal-ZT ratio is predicted. Over $B=0.3$--1 and $N=128$--256, spacing disorder adds to the rotational gain only if $\beta_{\rm eq}/\beta_e$ exceeds 0.5--0.6 for alternating noise and 0.7--1.1 for curvature-type noise (both up to 0.06\,\AA), 1.2--1.5 for positional noise of 0.03\,\AA\ (1.6 or more at 0.045\,\AA), and 1.8 or more for i.i.d.\ noise of 0.03\,\AA; each value carries the interval of Table~\ref{tab:thr}. The thresholds refer to the model's $\beta_e\sigma_d$; for a material, \texttt{vk\_predict} evaluates the optimal-ZT ratio directly from its parameters and spectrum.

\paragraph{The real aligned crystal.} All ratios above are relative to an aligned stack with the continuum (registry-averaged) coupling. The real aligned crystal differs in two identifiable ways: its phonon conductance carries the registry factor $g_{\rm reg}$ (Sec.~\ref{sec:registry}; T1-d), and its electrons couple with $t(\bm\tau_{\min})$ at the registry minimum instead of $\bar T$ (both from T1-a). A registered crystal with $g_{\rm reg}<1$, as expected for van der Waals crystals, lowers every phonon ratio by that factor. The electron reference can move either way, depending on whether $|t(\bm\tau_{\min})|$ is larger or smaller than $|\bar T|$. Both corrections are material outputs, not model results.

\begin{figure}[!tbp]\centering
\includegraphics[width=\textwidth]{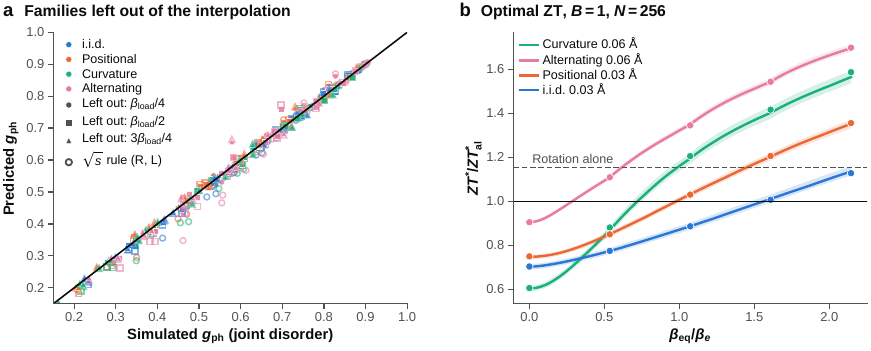}
\caption{Out-of-sample predictions and the conditional window. (a)~Predicted against simulated phonon conductance of joint disorder ($N=16$--256) for the three interior stiffness families, each predicted from the other four (filled) or by square-root interpolation between the outer two (open). (b)~Optimal-ZT ratio against $\beta_{\rm eq}/\beta_e$ ($B=1$, $N=256$; each stack at its optimal chemical potential): five-family interpolation with its 68\% bootstrap band, and the simulated families (markers); lines: break-even and rotation alone.}\label{fig:outofsample}
\end{figure}
\begin{figure}[!tbp]\centering
\includegraphics[width=\textwidth]{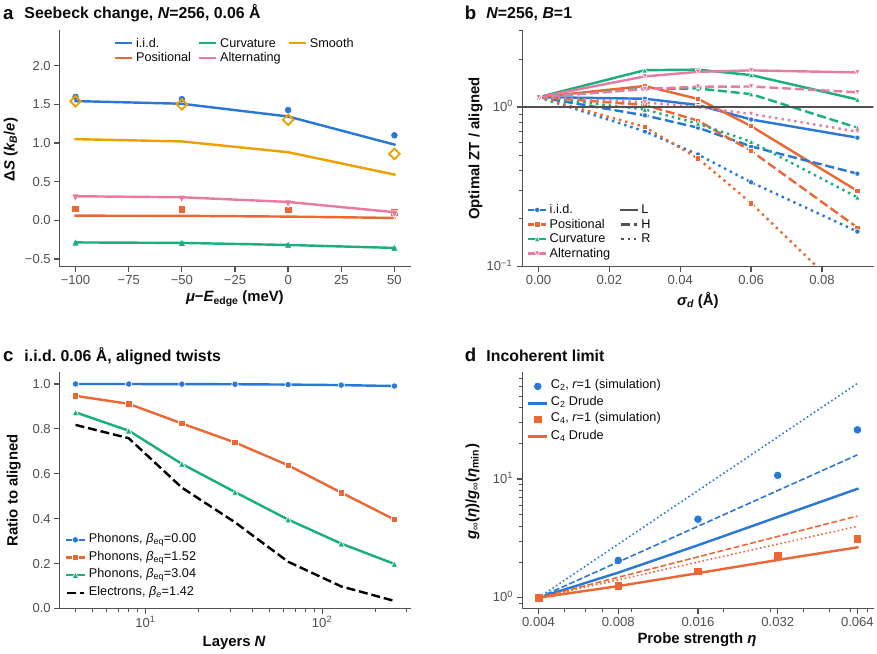}
\caption{A conditional window. (a)~Seebeck change at $N=256$ and 0.06\,\AA\ (random twists) against the chemical potential: simulation (markers; smooth spectrum open) and parameter-free Fermi--Dirac prediction (lines). (b)~Optimal-ZT ratio at $N=256$ and $B=1$ against the noise level, for four spectra and the stiffness families L (solid), H (dashed) and R (dotted). (c)~Spacing response of phonons for $\beta_{\rm eq}=0$, $\beta_{\rm load}/2$ and $\beta_{\rm load}$, and of electrons (i.i.d.\ noise, 0.06\,\AA, aligned twists). (d)~Incoherent limit of $r=1$ chains (Sec.~\ref{sec:coherence}): frozen ratio against the probe strength, normalized at the smallest $\eta$; B\"uttiker-probe simulations (circles C$_2$, squares C$_4$), Born--Drude curves (solid), predicted exponent $\alpha_{\rm inc}$ (dashed) and naive exponent $(c+s)/2(\kappa-r)$ (dotted).}\label{fig:thermo}
\end{figure}

%% file: sec_beyond.tex
\section{Beyond the coherent continuum}\label{sec:beyond}

\subsection{Phase breaking}\label{sec:coherence}
We model phase breaking with B\"uttiker voltage probes on every layer: strength $\eta$, zero net current at each energy, so the dephasing is elastic and conserves the channel $\bp$. The aligned reference carries the same probes. Beyond a few phase-breaking lengths the ratio $g=G/G_\al$ freezes at $g_\infty(\eta)$.

\begin{proposition}[incoherent limit, Drude scaling]\label{prop:incoh}
Let the dephasing rate per layer be the same for all channels, so that the dephasing resistivity of channel $\bp$ is $\rho_\phi\propto\eta/v(\bp)$ with $v\propto\rho^r$ near the axis, and let the rotational resistivity be $\rho_d=2\gamma\propto\rho^{2(\kappa-r)}$. In the Drude regime $g_\infty=\langle1/(1+\rho_d/\rho_\phi)\rangle$, weighted by the aligned conductance, and
\begin{equation}
g_\infty\propto\eta^{\alpha_{\rm inc}},\qquad \alpha_{\rm inc}=\min\!\left(\frac{c+s+r}{2\kappa-r},\,1\right),\label{eq:incoh}
\end{equation}
with a logarithmic factor when the first argument equals 1. For $r=0$, $\alpha_{\rm inc}=\min(\alpha,1)$ with $\alpha=(c+s)/2\kappa$.
\end{proposition}
\noindent\emph{Proof.} Channels with $\rho_d<\rho_\phi$, i.e.\ $\rho<\rho_c\propto\eta^{1/(2\kappa-r)}$, conduct as in the aligned stack; their share of the aligned conductance, whose weight is $\propto\rho^{c+s+r-1}/\eta$, is $\propto\rho_c^{c+s+r}$. Channels outside conduct $\propto\rho_\phi/\rho_d$; their contribution $\propto\eta\int_{\rho_c}\rho^{c+s-1-2(\kappa-r)}\dd\rho$ is of the same order $\rho_c^{c+s+r}$ if $c+s<2(\kappa-r)$, and otherwise converges at small $\rho$ and is $\propto\eta$. $\square$

The law $\ell_\phi^{-\alpha}$ holds for $r=0$, $\alpha<1$; for $\alpha\ge1$ the Drude integral diverges, strongly scattered channels still conduct $\propto1/\ell_\phi$, and the exponent is capped at 1. The derivation is a Drude (Matthiessen) scaling argument, so the proposition is conditional on that description of channel-conserving dephasing.

\emph{Tests.} Chains with $r=0$ give $0.48$ and $0.22$ for C$_2$ and C$_4$ electrons (predictions $1/2$ and $1/4$). Further tests use chains whose coupling vanishes on the axis ($T=T_1p$, $r=1$, bandwidth measure $s=1$), where the prediction differs from the naive one: for C$_2$ layers $\alpha_{\rm inc}=1$ (naive $3/2$), for C$_4$ layers $4/7$ (naive $1/2$). B\"uttiker-probe simulations ($N\le256$, 4 sequences, $\eta=0.004$--0.064; \texttt{incoherent7.py}) give local slopes of $g(256)$ in $\eta$ of 1.04, 1.16, 1.23 and 1.28 for C$_2$ and 0.33, 0.40, 0.44 and 0.47 for C$_4$, from small to large $\eta$ (Fig.~\ref{fig:thermo}d). These tests are not conclusive.
\begin{itemize}
\item Freezing is incomplete at small $\eta$ ($g(256)/g(128)=0.77$ for C$_2$ and 0.85 for C$_4$ at $\eta=0.004$), which biases the small-$\eta$ slopes low. The C$_2$ slopes lie between the prediction and the naive value, and their approach to 1 may partly be this bias; the C$_4$ slopes move away from $4/7$ as $\eta$ decreases.
\item The Born--Drude curves are themselves pre-asymptotic in this range (slopes 0.70--0.79 for C$_2$, 0.33--0.37 for C$_4$). Their amplitude overestimates $g_\infty$ by a factor 4--12 (C$_2$) and 1.4--1.6 (C$_4$), because strongly scattered channels localize within a coherence length (\texttt{incoherent7\_drude.py}).
\end{itemize}
The exponents for $r>0$ therefore rest on the derivation. The cap at 1 does not depend on the Drude form: strongly scattered channels conduct $\propto\eta$ also in the sequential-tunnelling picture.

\emph{Scope.} The probes are elastic and channel conserving. If in-plane scattering randomizes $\bp$ faster than carriers cross a layer (sequential tunnelling, the incoherent limit of interlayer transport in twisted graphene~\cite{Bistritzer2010,Kim2013}), each channel's interlayer rate is a Lorentzian overlap of the rotated bands, and the ratio scales as $\Gamma_\parallel^{\min(2\alpha,2)}$ in the in-plane broadening, a different law. Inelastic energy relaxation is not included. The incoherent-limit exponent is therefore a statement about a named dephasing mechanism, not a universal constant.

\paragraph{Consequence: an Ohmic stack with a symmetry-indexed resistivity.} Beyond the phase-breaking length the resistance is linear in $N$. Within the Drude argument, for channel-conserving dephasing, the resistivity ratio to the aligned stack is $1/g_\infty\propto\ell_\phi^{\alpha_{\rm inc}}$: $\alpha_{\rm inc}=1/2$ for Fermi-level carriers of black phosphorus, $1/3$ ($1/6$) for the MoS$_2$ $\Gamma$ valence band when its spin--orbit (scalar) term dominates. Each species freezes at its own $\ell_\phi$. At room temperature phonon mean free paths are about one layer~\cite{Eriksson2023}, and the phonon side is then governed by $g_{\rm reg}$ (Sec.~\ref{sec:registry}) rather than by the coherent laws.

\subsection{Umklapp and nearly aligned interfaces (conditional)}\label{sec:umklapp}
\paragraph{Off-shell floor.} The $\mathbf G\neq0$ couplings take a carrier at $\bp\approx0$ to partner states $\bp+\mathbf G_l$ in the neighbouring, rotated layers. Second-order elimination is valid when these partners are off shell, $|E-\varepsilon_Q|\ge\Delta>0$ with $\|H_{PQ}\|/\Delta\ll1$; for a band edge that is a global extremum and an interface far from alignment they lie far from $E$, but the spectral condition, not the angle, is what must be checked. It gives the correlated on-site shift
\begin{equation}
\Sigma_l(\bp,E)=\sum_{\mathbf G\in{\rm shell}(l)}T(|\bp+\mathbf G|)^2\left[\frac{1}{E-\varepsilon_{l+1}(\bp+\mathbf G)}+\frac{1}{E-\varepsilon_{l-1}(\bp+\mathbf G)}\right],
\end{equation}
with $\Sigma_l=a(\Delta_{l-1})+b(\Delta_l)$, $C_1/C_0=1/2$ by mirror symmetry (0.48 measured), and band-edge protection up to
\begin{equation}
N_U=\frac{4\,\bar T\,k_BT}{S_V(0)},\qquad S_V(0)=C_0+2C_1 .\label{eq:NU}
\end{equation}
In the C$_6$ test model ($u=T(|\mathbf G|)/\bar T=0.31$, 0.44, 0.59) the effective chain gives $N_U=4009$, 943 and 288, and the parameter-free prediction follows it within 4.1\% down to $g\approx0.4$ (Fig.~\ref{fig:umklapp}).

\paragraph{Status.} This test validates the second-order effective chain, with interfaces within $\theta_c=10^\circ$ of alignment excluded. The floor is dominated by the interfaces just outside the cutoff (${\rm Var}\,\Sigma\propto\theta_c^{-3}$ above about $4^\circ$, where the shift saturates), and near-aligned interfaces are resonant. We treat them as registered interfaces with the T1 registry spread, a fraction $f=n\theta_c/\pi$, which gives bond disorder $\sigma_{\rm near}^2$ and a scale $N_{\rm near}=k_BT/\bar T\sigma_{\rm near}^2$. The umklapp floor is therefore a conditional result: valid in the off-shell regime, with the near-aligned treatment as a stated model assumption; the main-text claims do not use it. Away from the band edge, umklapp partners are on shell and scatter between channels; the stack turns diffusive beyond an umklapp length $\ell_U$, with $G\propto\ell_U^{1-\min(\alpha,1)}/N$.

\paragraph{High shells and commensurate angles (tail estimate).} Away from alignment the lowest shells are off shell, but an incommensurate twist brings some rotated vector $R_\theta\mathbf G'$ arbitrarily close to a vector $\mathbf G$ of the other layer (a Diophantine small mismatch), at the price of a large $|\mathbf G|$: for phosphorene lattices and a mismatch below 0.05\,\AA$^{-1}$, 66\% of random twists have no such pair with $|\mathbf G|\le16$\,\AA$^{-1}$. A partner is on shell when the mismatch lies in the thermal window, $\delta<\delta_{\rm th}$ (0.03--0.1\,\AA$^{-1}$ for band edges at 300\,K), and its coupling carries the evanescent decay of both planar components across the gap, $|T_{\mathbf G}|/|\bar T|\approx c\,e^{-(\kappa_G-\kappa_0)d_{\rm gap}}$ with $\kappa_G=(G^2+\kappa_0^2)^{1/2}$. For phosphorene the PBE planar components give $\kappa_0=0.80$\,\AA$^{-1}$ and 1.56--2.03\,\AA$^{-1}$ for the first shells, consistent with this form, and $d_{\rm gap}=3.1$\,\AA. With $c=1$ this form underestimates the computed armchair-shell coupling (below) fivefold and overestimates the zigzag one, so the ratios in this paragraph locate the on-shell interfaces rather than fix their amplitudes; the statement that 66\% of interfaces have no on-shell vector within 16\,\AA$^{-1}$ is purely geometric. The largest on-shell ratio at a given twist (all shells with $|\mathbf G|\le16$\,\AA$^{-1}$; \texttt{blackphosphorus/\allowbreak umklapp\_tail.py}) is shown in Fig.~\ref{fig:umklapptail}. It is set by the lowest shell that reaches the window. Near alignment the first shells are on shell, with ratios up to 0.07 (armchair) and 0.02 (zigzag); at the other commensurate angles, where two shells of equal length coincide (such as $2\arctan(b^*/a^*)\approx74^\circ$ for the $(1,1)$ shells), the ratio is at most $5\times10^{-3}$; high shells are on shell at many angles but with ratios below $10^{-7}$. For uniform random twists and $\delta_{\rm th}=0.05$\,\AA$^{-1}$ (a grid of $2\times10^5$ angles), the ratio exceeds $10^{-2}$ at 2.2\% of the interfaces, all within $2^\circ$ of alignment, $10^{-4}$ at 6.7\% and $10^{-6}$ at 9.7\%; at 90\% of them it is below $2\times10^{-7}$, and 66\% have no on-shell vector with $|\mathbf G|\le16$\,\AA$^{-1}$ at all. The first common vector of the $90^\circ$ cell, 5.716\,\AA$^{-1}$, gives $2\times10^{-7}$, the size of the forward-dominance truncation of Sec.~\ref{sec:t90}. The umklapp problem of a random stack is therefore one of rare nearly aligned interfaces, not of a dense set of small denominators. The twist average $\langle r^2\rangle$ is dominated by interfaces within $2^\circ$ of alignment, where the first armchair shell carries more than 99\% of it; its amplitude is therefore taken from the computed components, below.

\paragraph{First shell from the computed components.} The first armchair shell ($|\mathbf G|=1.436$\,\AA$^{-1}$) decides the umklapp risk, so we take its amplitude from the PBE planar components of the $\Gamma$ conduction edge itself (\texttt{blackphosphorus/\allowbreak umklapp12.py}). In Bardeen form with both faces on one plane, the ratio of the umklapp to the forward coupling is $\rho_U=|F_{\mathbf G}F'_{\mathbf G}|/|F_0F'_0|=0.341$ on the midplane and 0.367 on the pair of admissible planes 1.30/1.80\,\AA\ (the planar value ratio on the midplane is 0.355; the empirical form gave 0.073). The partner of a lab channel $\bp$ across an interface of twist $\theta$ is $\bp+\bm\delta$, $|\bm\delta|=2|\mathbf G|\sin(\theta/2)$, open for thermal carriers when $E(\bm\delta)\lesssim k_BT$, i.e.\ $|\theta|\lesssim\theta_U=\delta_{\rm th}/|\mathbf G|$ with $\delta_{\rm th}=(k_BT/c_{zz})^{1/2}$ and $c_{zz}=3.19$\,eV\,\AA$^2$ the zigzag curvature of the PBE band (the mismatch lies nearly along the zigzag direction). The C$_2$ shell coincides near $\theta=0$ and $\pi$, so a fraction $f_U=4\theta_U/2\pi$ of random interfaces has an open partner: $\theta_U=3.6^\circ$ and $f_U=0.040$ at 300\,K, $1.8^\circ$ and 0.020 at 77\,K. With a momentum-relaxation probability $\Gamma_U\sim f_U\rho_U^2$ per layer (transmission factors of order one not computed), $\Gamma_U=(4.6$--$5.4)\times10^{-3}$ at 300\,K and $(2.3$--$2.7)\times10^{-3}$ at 77\,K, and a rough crossover estimate for the forward-channel laws is $N_U\sim1/\Gamma_U\approx190$--220 layers at 300\,K and 370--430 at 77\,K. This is an order-of-magnitude estimate: $\theta_U$ is the characteristic on-shell angle of low-energy channels near the axis of a parabolic band, not a hard boundary for every thermal channel, and $f_U\rho_U^2$ does not resolve the velocities and fluxes of the partners; it places the crossover at a few hundred layers, above the window of 16--128 layers. Beyond $N_U$ umklapp moves carriers between thermal channels, many of which are themselves protected; that regime calls for a multichannel chain with these couplings.

\begin{figure}[!tbp]\centering
\includegraphics[width=\textwidth]{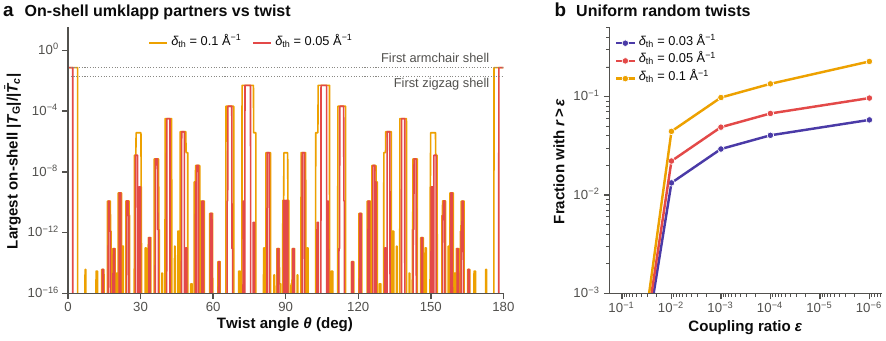}
\caption{Umklapp tail of a random-twist phosphorene interface. (a)~Largest on-shell umklapp coupling of a $\Gamma$-valley state relative to the forward one against the twist angle, for thermal windows 0.05 and 0.1\,\AA$^{-1}$ (all shells to 16\,\AA$^{-1}$; $c=1$); dotted: the first shells. (b)~Fraction of uniform random interfaces whose ratio exceeds $\epsilon$.}\label{fig:umklapptail}
\end{figure}

\begin{figure}[!tbp]\centering
\includegraphics[width=0.72\textwidth]{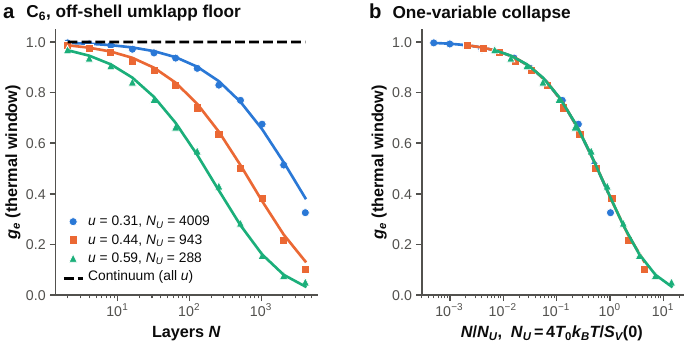}
\caption{Umklapp floor for band-edge carriers in the second-order effective chain: C$_6$ model, thermal window at 300\,K, interfaces within $10^\circ$ of alignment excluded. (a)~Simulation (markers) and parameter-free prediction (lines); dashed: continuum. (b)~Collapse against $N/N_U$.}\label{fig:umklapp}
\end{figure}

\subsection{Off-axis band edges}
Band edges away from the rotation axis (K valleys of dichalcogenides, L-derived edges of rock-salt slabs) have rotational disorder of the order of the valley depth; their carriers are blocked after a few layers. The protected carriers are those at $\Gamma$ with $r=0$.

%% file: sec_validation.tex
\section{Material validation: what is computed and what it identifies}\label{sec:validation}

Every material input of the theory is an identifiable quantity of a primitive-cell or small-cell calculation, with stated exceptions. The protocol (\texttt{kit/\allowbreak validation-protocol.md}) gives settings and file formats; this section lists each calculation with the quantities it identifies. Every kit output carries one of three categories: \emph{identified} (determined by the data alone), \emph{conditional} (determined under a named assumption, which is printed with the value), or \emph{consistency test} (a check that can fail, not an identification). Nothing that the data cannot determine is promised: in particular the material ZT is \emph{not} an output, because it needs the aligned crystal's absolute conductivities. The regression tests of the kit include seven inputs for which a simpler analysis reports a pass or an identified value that the data do not support (the tools' self-tests): five in the parameter tools, a spin--orbit-mixed hole pair whose weak $p$-independent forward amplitude competes with a $p^2$ term (a crossover, not an order), and face amplitudes taken as coefficient sums of Wannier functions without their physical forward weights, which a constant re-phasing of the basis turns from A- into B-type. All seven are rejected or labelled conditional.

\subsection{Materials}
\textbf{Black phosphorus} (D$_{2h}$, C$_2$; first-principles values in Sec.~\ref{sec:dft}): in phosphorene models the $\Gamma$ conduction edge is A-type ($r=0$, $\kappa\approx2$ at 0.02--0.1\,\AA$^{-1}$; spin--orbit coupling drives it toward the Kramers index 1 below $p^*\approx0.007$\,\AA$^{-1}$ in the $sp^3$ model) and the valence edge B-type ($r=2$, not protected; a spin--orbit window $p\lesssim p_h\approx0.02$\,\AA$^{-1}$ with a coupling $3\times10^{-4}$ of the electrons'), an electron--hole contrast in the forward-channel model whose parities T0-m identifies ($p_h$ and the hole/electron coupling ratio are ratios of face data on one plane, quoted with their spread over the admissible planes; $p^*$ depends in addition on the scale of $T$, which T2-a constrains for the bilayer but does not fix for a multilayer); flexural $\kappa_z=4$; in-plane phonons not protected ($\kappa_t=r=2$). The showcase for the thickness laws; its large mass anisotropy limits band-edge protection. \textbf{MoS$_2$} (D$_{3h}$, C$_3$): $\Gamma$ valence band $J_z=\pm1/2$, $r=0$, $\kappa=3$ (on-site spin--orbit $p^3$ term and face amplitudes; 6 where both are negligible); in-plane phonons $N^{-3/4}$ (central interlayer forces); K-valley conduction band off-axis, a built-in negative control. SnSe is optional (T0, T1 only).

\subsection{Calculations and identified quantities}
{\small
\begin{longtable}{L{0.9cm}L{5.0cm}L{2.9cm}L{5.8cm}}
\caption{Calculations of the validation protocol, the scripts that read them and the quantities they identify.}\label{tab:calcs}\\
\toprule
tier & calculation (cell) & script & identified quantities\\
\midrule
\endfirsthead
\toprule
tier & calculation (cell) & script & identified quantities\\
\midrule
\endhead
\bottomrule
\endfoot
T0-a & monolayer bands on a polar $k$ grid around $\Gamma$ (radii 0.005--0.25\,\AA$^{-1}$ $\times$ 48 angles); spin--orbit and spin textures for MoS$_2$; projections of the $\Gamma$ multiplet on $(l,m_l,s_z)$ or $(j,m_j)$ & \texttt{vk\_fit\_kappa} (scalar, spinor, \texttt{--rep}) & C$_n$ class, $r$, $\kappa_{\rm on}$ (checked against the rule), $v_V$, harmonics, masses; ring check if $\kappa=1$; admixture weights (size of $\kappa_T$ terms)\\
T0-m & the same polar grid, plus a fine grid bracketing the spin--orbit crossovers (black phosphorus: radii $10^{-4}$--$5\times10^{-3}$\,\AA$^{-1}$): wavefunctions of the $\Gamma$-derived multiplet at every $\mathbf q$ (spinors along the lab $z$ axis for spin--orbit runs) and its energies. Face vectors: the planar Fourier component $\mathbf G_\parallel=0$ of each Bloch state and its $z$ derivative at one gap plane per face (a height $h$ above the top atoms; $d_{\rm gap}-h$ below the bottom atoms of the next layer), normalized over the cell volume. Planes must lie outside the augmentation spheres of both layers (\texttt{rcore}); use several planes of that band. Interface coupling in eV from the Bardeen transfer form, $T=(\hbar^2/2m)A\,[F_{\rm top}^\dagger F'_{\rm bot}-F'^\dagger_{\rm top}F_{\rm bot}]$, written with the data. The reader is validated on the same calculation: planar density from the wavefunctions against the code's charge density in the vacuum & \texttt{vk\_face\_from\_pw} (VASP \texttt{WAVECAR}, Quantum ESPRESSO, generic arrays; \texttt{--check-density}), then \texttt{vk\_matrix\_index} & $\kappa$ and $r$ (identified if the gauge-invariant C$_n$ check passes, the Born expansion is valid, the local orders are stable over the fit radii and the projection is physical; otherwise conditional, with the reason, e.g.\ a crossover; C$_1$ Kramers: transfer matrices); the plane dependence of the Bardeen $T$ (printed spread; $\kappa$ and $r$ must not depend on the plane); the coupling scale in eV, constrained by T2-a for the bilayer (a multilayer constant is not identified); $p_h$ and the hole/electron coupling ratio as ratios of face data, with their spread over the admissible planes; $p^*$ with its dependence on the coupling scale. Planar Fourier ratios $\mathbf G_\parallel\neq0$ at the gap planes are printed as a diagnostic only: an umklapp coupling needs the coincidence set, the matrix elements and the partner energies\\
T0-b & monolayer phonons with rotational sum rules, on the polar $q$ grid & \texttt{vk\_fit\_kappa} (scalar ZA, \texttt{--tensor}) & $\kappa_z$ (ZA quadratic), $v_z$, $\kappa_t$\\
T0-c & monolayer band on the circles $|\mathbf k|=|\mathbf G|$ & \texttt{vk\_predict} & umklapp partner energies and the off-shell margin $\Delta$; with umklapp matrix elements (not the planar ratios of T0-m alone), the floor $N_U$ (conditional)\\
T1-a & aligned bilayer registry scan: uniform $6\times6$ shifts $\times$ $\ge5$ spacings, frozen layers; $\Gamma$ splitting, layer polarization $P$ and bonding character of each band-edge pair; total energy & \texttt{vk\_bilayer\_T} & with the bonding character (or phase-resolved layer projections): signed $t(\bm\tau)$, $\bar T=\langle t\rangle_{\bm\tau}$, $\beta_e$, umklapp $c_{\mathbf G}$ (identified); $\sigma_{\rm reg}$. Without it $\bar T$ is not identified, whatever the amplitudes: a sign change between registry samples is invisible in $|t|$ (example: $0.03+0.10\cos6\pi\tau_1$\,eV on $6\times6$ reads $-100$\,meV, true $+30$\,meV); a conditional value is printed\\
T1-b & same scan, energies & \texttt{vk\_interlayer\_\allowbreak phonon} & $d_{\rm tw}$, $F_{zz}(0)/M$ in (rad/ps)$^2$, $\omega_B$, $\beta_{\rm load}$, aligned shear stiffness; $\beta_{\rm eq}$ from the registry family (a proxy, checked in T2-b); the finite-$p$ continuum $F(\bp)$ and $\ell_{\rm ph}$ (\AA) are conditional on the assumed pair form, since the scan fixes only $U(0;d)$ and $U(\mathbf G;d)$\\
T1-c & registry-averaged splitting at $|\mathbf k|=0.05$, 0.1\,\AA$^{-1}$ in 4 directions (same scan, extra $k$ points) & \texttt{vk\_twist\_check --bands} & coupling range $\ell_T$ (conditional fit of a Gaussian range model); it does not determine $\kappa$ or the matrix coupling, which need T0-m\\
T1-d & aligned bulk phonons (primitive cell, DFPT, $6\times6\times2$ $q$ grid) & (direct) & registry-loss factor $g_{\rm reg}$ (Sec.~\ref{sec:registry})\\
T2-a & relaxed commensurate twisted bilayers: MoS$_2$ $\sqrt7$ at 21.79$^\circ$ and 38.21$^\circ$ (42 atoms) and the held-out $\sqrt{13}$ at 32.20$^\circ$ (78 atoms); black phosphorus 90$^\circ$ ($4\times3$ on $3\times4$ cells, 96 atoms, $\pm0.48\%$ strain at the experimental lattice: done, Sec.~\ref{sec:dft}; or $7\times5$ on $5\times7$ cells, 280 atoms), with the aligned bilayer at the same spacing as reference; $\Gamma$ splitting with layer polarization, unfolded as in Ref.~\cite{Tan2016} & \texttt{vk\_twist\_check --gamma} & parameter-free test $R=|t(0;\theta)|/|\bar T(d)|=1$; spread of $d(\theta)$ (``not identified'' from one angle); for black phosphorus the $\Gamma$ splittings of the conduction and of the valence pair (contrast test) and the scale of the forward coupling from the projection of the bilayer $\Gamma$ states on those of each layer computed alone (target subspace, $2\times2$ $H/S$, leakage), not from a selected level pair\\
T2-b & $\Gamma$ phonons of the twisted bilayers~\cite{Maity2020} & \texttt{vk\_twist\_check --shear}, \texttt{--breathing} & consistency tests: breathing frequency vs $\omega_B$ and, separately, the stiffness--spacing slope across twisted interfaces vs the registry $\beta_{\rm eq}$ (``not identified'' if the cells spread $d$ by less than 0.02\,\AA); shear pinning only from $\ge3$ cells approximating one target angle, strictly decreasing with cell size (equal or single values are not a convergence statement)\\
T2-c & MoS$_2$ trilayer $0^\circ/21.79^\circ/0^\circ$ ($\sqrt7$, 63 atoms); $\Gamma$ levels of the band-edge triplet and their layer-resolved amplitudes & \texttt{vk\_twist\_check --trilayer} & with the amplitudes: the $3\times3$ layer Hamiltonian, $t_{12}$, $t_{23}$, $t_{13}$ and the shifts (identified), and the held-out test $t_{12}=t_{23}=\bar T$, $t_{13}\approx0$; with levels only, $t$ is conditional on $t_{12}=t_{23}$, $t_{13}=0$ (three levels fit that chain exactly whatever the truth: $t_{12}=0.08$, $t_{23}=0.12$\,eV returns 0.102\,eV)\\
T3 & thickness curves with the DFT band (no new DFT) & \texttt{vk\_chain} (\texttt{--sz}, requires the $\sigma_z$ splitting column); \texttt{vk\_face\_chain} (full face Hamiltonian: PBE on-site band and random bonds from the T0-m faces) & exponents, $N_e$; mean, median and typical conductance at finite $N$\\
all & parameters $\to$ predictions & \texttt{vk\_predict} & $g_e(N)$ (Fermi--Dirac), $\Delta S$, $g_{\rm ph,z}(N)$, optimal-ZT ratio for a given quality factor $B$\\
\end{longtable}}
Every script reads named columns with units and fails on a missing column; every output carries its unit; under-determined inputs return a ``not identified'' marker that propagates to every prediction that needs it. An end-to-end synthetic test (\texttt{vk\_e2e.py}) generates T0--T2 data from the stable model (with the bonding column), recovers every parameter ($\bar T$ to $10^{-6}$, $\beta_e$ to 0.3\%, $d_{\rm tw}$ to $10^{-4}$\,\AA, $F_{zz}(0)/M$ to 0.1\%, $\beta_{\rm load}$ to 0.1\%, $\beta_{\rm eq}$ within 0.021\,\AA$^{-1}$ (8\%) of its exact value, $\ell_{\rm ph}$ and the Morse parameters exactly, the last two because the synthetic data obey the assumed pair form) and predicts $g_e(N)$ within 1.5\% of the direct Born prediction; both lie within the Born--$\Phi_1$ accuracy of the channel simulation of the same model. It also checks that a coupling without bonding information stops the pipeline. \texttt{vk\_matrix\_index} recovers $\kappa$ of random symmetric models in a smooth and in a completely scrambled multiplet gauge (7 classes; C$_1$ Kramers flagged and treated by transfer matrices). \texttt{vk\_face\_from\_pw} reads VASP (\texttt{vasp\_std}, \texttt{vasp\_ncl}) and Quantum ESPRESSO wavefunctions, or a generic array file, and writes its own $k$-point lists. Its self-test uses synthetic files in both layouts: Gaussian-orbital phosphorene states on VASP's plane-wave sphere (planar components and slopes reproduced to $1.4\times10^{-5}$ and $5.6\times10^{-5}$ at a 250\,eV cutoff; the VASP route equals the QE route to $4\times10^{-8}$); the invariances of the coupling (vacuum 20$\to$40\,\AA: face vectors unchanged to $7\times10^{-6}$; gap plane: the Bardeen $T$ of states with exact exponential tails is constant to $7\times10^{-8}$ over 0.5--2.5\,\AA\ and equals its closed form; a real change of the profile basis with the contragredient kernel leaves the Born rates unchanged to $3\times10^{-12}$); and $r=0$, $\kappa=2$ (conduction) and $r=\kappa=2$ (valence; the Born index of this scale-invariant chain is masked because its rates exceed 0.1 per layer, and $\kappa=r$ is read from the diagnostic slope) under random phases and random U(2) rotations of the Kramers pairs, with the kernel in eV. Real first-principles output is checked by the density test above.

\paragraph{Scope of the index tool and the transport modules.} (i)~\emph{Born applicability.} A Born rate is a small-disorder quantity; \texttt{vk\_matrix\_index} reports rates of 0.1 per layer or more as ``not applicable'' (\texttt{BORN\_MAX}), excludes them from the open-channel diagnostics and prints a flag. It also masks the index: where the rates at the fit radii reach the threshold, $\kappa$ is not returned (NaN, printed ``n/a''; category ``not applicable (Born)''), and the unmasked slope survives only as the diagnostic \texttt{kappa\_raw}. A self-test forces the masking, and a crossover model with a weak forward face falls into this category. The black-phosphorus valence rates (2--6 per layer on the midplane, up to 15 over the five planes) and the spin--orbit hole rates from $5.7\times10^{-4}$\,\AA$^{-1}$ up fall in this class (at $1.357\times10^{-3}$\,\AA$^{-1}$ the rate is not even defined on four of the five planes). A profile-basis counterexample (a real change of basis that keeps the bonds to $10^{-16}$ but changes the common open-channel count from 0 to 2 at $p=1.357\times10^{-3}$\,\AA$^{-1}$) lies there; the three-start gauge alignment is a heuristic that repairs the observed failures; basis invariance is checked on the data (Sec.~\ref{sec:dft}). (ii)~\emph{\texttt{vk\_face\_chain}} computes exact thermal or band-averaged conductances of scalar random-twist chains with the full face Hamiltonian (polynomial faces with the parity of the edge, continuous twists, aligned leads, recursive Green functions), with the mean, median and quantiles over twist sequences and the typical channel, and an optional quadrature refinement around the transparent energy of B-type edges; its self-test checks aligned stacks ($g=1$), dense Caroli inversions for one and two layers, the quadrature (weights, smooth integrands, a cusp at $E^*$) and $\langle T\rangle\ge e^{\langle\ln T\rangle}$. (iii)~\emph{\texttt{vk\_channel\_fast}} computes one channel of the same model in a C kernel (about 50\,ns per layer and sequence on two cores), with energies measured from $E^*$ so that the node region keeps full relative precision; it returns $\ln T$ of every sequence, and builds the ideal factorized chain. Every pair (seed, sequence) starts its own stream at a hashed state, and the calculations hash each grid point's parameters into its seed, so that different points are statistically independent. Its self-test checks the C kernel against a numpy reference to $10^{-14}$ (ideal, B- and A-type models), aligned films ($T=1$), $T(E^*)=1$ for every sequence of the ideal chain between matched leads, a positive exponent off $E^*$, and that the kernel draws exactly the angles of its stream, with neighbouring streams sharing none.

\subsection{What supports the theory, and what would falsify it}
\begin{center}\small
\begin{tabular}{L{3.2cm}L{5.8cm}L{5.4cm}}
\toprule
claim & supporting outcome & falsifying outcome\\
\midrule
matrix index (Prop.~\ref{prop:matrix}, Table~\ref{tab:kappa}) & \texttt{vk\_matrix\_index} on T0-m data gives $\kappa$ within $\pm0.15$ of the value for the computed class and face relation, over a factor $\ge3$ in $p$; on-site fit (T0-a) consistent & a different integer: the multiplet class, the face channels or the genericity assumption is wrong for the material\\
continuum coupling (Prop.~\ref{prop:coupling}) & $|R-1|<0.05$ at every twisted cell, including the held-out 32.20$^\circ$; trilayer $t/\bar T=1$ within 5\% & a twist-dependent $R$: moir\'e or commensurate terms are not negligible and $\bar T$ is not the twisted coupling\\
superlubric phonon coupling & shear frequency decreasing with cell size; breathing mode within 5\% of $\sqrt{2F_{zz}(0)}$ & a size-independent shear stiffness at generic angles\\
phonon spacing response & $\beta_{\rm eq}$ from the registry family agrees with the slope across twisted cells & the two slopes disagree (the registry family is not the relevant one), or registries do not spread $d_{\rm eq}$ ($\beta_{\rm eq}$ then needs another family, e.g.\ intercalation)\\
electron--hole contrast (black phosphorus) & T0-m: conduction pair A-type, valence pair B-type with a small spin--orbit forward amplitude; 90$^\circ$ bilayer: valence $\Gamma$ splitting $\ll$ conduction splitting & a valence pair with an $O(1)$ forward amplitude, or comparable splittings in the 90$^\circ$ bilayer\\
Bardeen coupling scale & the projected $90^\circ$ coupling within a factor of about 2 of the midplane $|\bar T(0)|$ (found: 1.5 two-state, 1.7--2.2 downfolded; 2.7--3.8 on the outermost admissible planes); $\kappa$, $r$ plane independent & a larger discrepancy, or no dominant pair in the projection: the scale stays conditional (found: the pair carries 81\% of the reference weight; the bilayer is consistent with the scanned range, but a multilayer scale is not identified)\\
umklapp (conditional) & $u\ll1$ and $N_U\gg$ target thickness (black phosphorus: first-shell $\rho_U=0.34$--0.37 on shell mainly near alignment, order-of-magnitude estimate $N_U\sim2$--$4\times10^2$ layers) & $N_U$ of a few layers\\
thickness law (real band) & local exponent $\to(c+s)/2(\kappa-r)$ in \texttt{vk\_face\_chain} with the full face Hamiltonian (black phosphorus electrons: $-0.47$ at $2.3\times10^4$ layers) & exponential decay at $N\ll N_U$\\
\bottomrule
\end{tabular}
\end{center}

\paragraph{Expected magnitudes (guidance only).} DFT masses of monolayer black phosphorus give strongly anisotropic edges, so $N_e$ for band-edge carriers is tens to about a hundred layers at 300\,K: black phosphorus shows the laws, not thermoelectric protection. The MoS$_2$ $\Gamma$ valence band has negligible $p^6$ disorder at thermal momenta; its protection is limited by the spin--orbit $p^3$ term and, conditionally, by umklapp. With Proposition~\ref{prop:coupling} the T2 test is parameter free and should pass unless moir\'e or umklapp terms are large.

\paragraph{Experimental route.} The coupling of a single interface can be constrained before thickness curves of random stacks are measured: angle- and momentum-resolved tunnelling between rotatable layers, as in the quantum twisting microscope~\cite{Inbar2023}, probes the twist and momentum dependence of $T$ (its scale and harmonics), and vertical transport through gated twisted black phosphorus~\cite{Srivastava2021} probes the coupling scale. Thickness curves of random stacks then test the laws; contact and area effects have to be separated from the stacking disorder, and a finite-thickness fit alone does not identify an asymptotic exponent.

%% file: sec_dft.tex
\section{First-principles test on black phosphorus}\label{sec:dft}

The monolayer face data, the index and the coupling scale were computed for black phosphorus with VASP 6.1.0 (PAW, PBE) on the Huadong A cluster. This section gives the methods, the checks of the toolchain and the results; the run-by-run record (job IDs, hashes, parallel shapes, timings, memory) is kept with the first-principles results (\texttt{results/bench.csv} and the run reports). Figure~\ref{fig:struct} shows every cell used.

\begin{figure}[!tbp]\centering
\includegraphics[width=\textwidth]{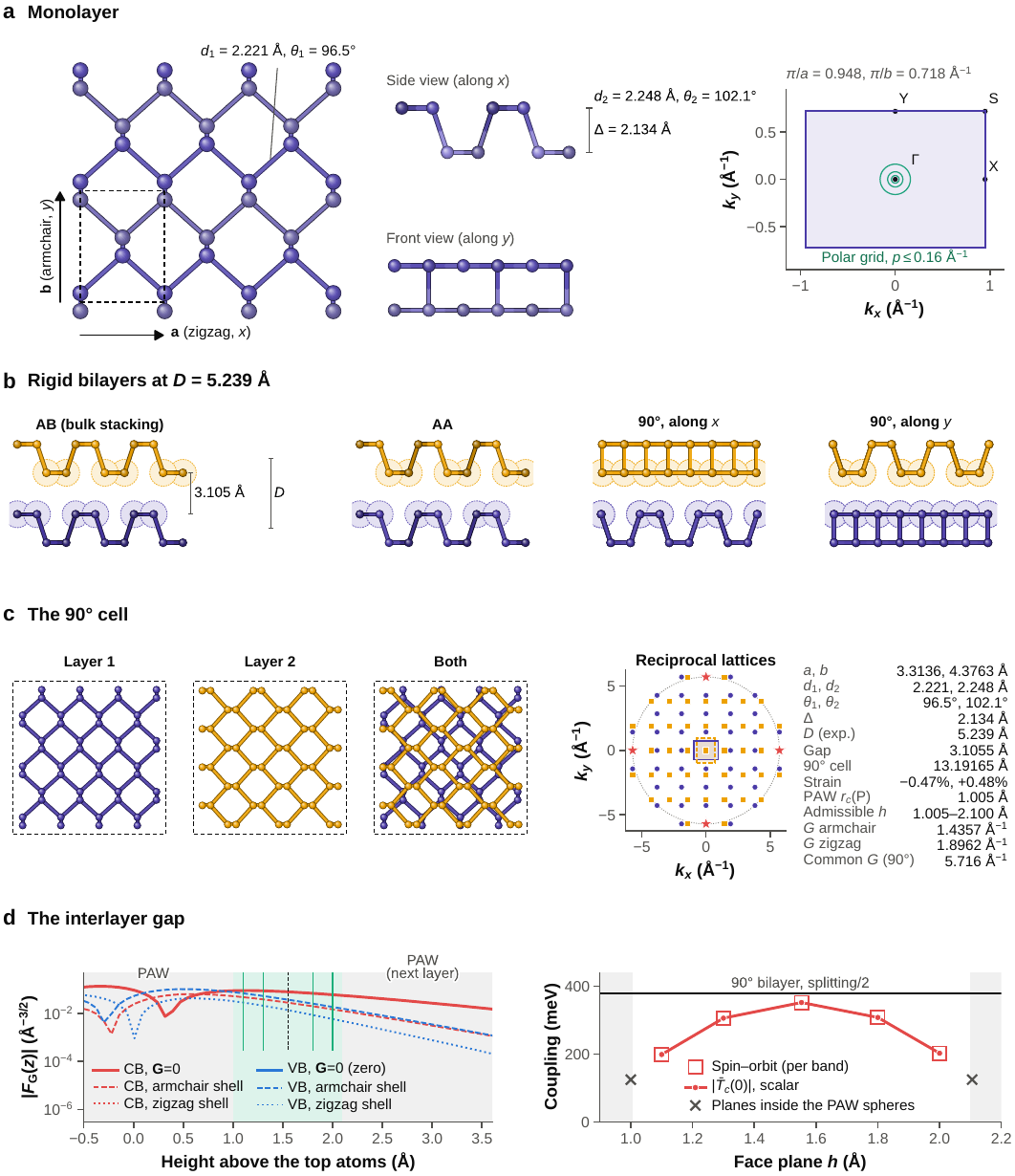}
\caption{Structures and zones of the first-principles calculations. (a)~Monolayer at the experimental lattice with PBE internal coordinates: top view (unit cell dashed), side view along zigzag (armchair profile) and front view along armchair (zigzag chains); first zone with the polar $\mathbf p$ grid. (b)~Rigid bilayers at $D=5.239$\,\AA: AB, AA and the $90^\circ$ cell viewed along both axes, with the P augmentation spheres (1.005\,\AA) of the facing atoms; markers in AB: the gap between the facing atoms and the layer distance $D$. (c)~The $90^\circ$ cell ($4\times3$ cells of layer 1 on $3\times4$ cells of layer 2; 96 atoms), its reciprocal lattices (stars: the first common vectors, 5.716\,\AA$^{-1}$) and the structure parameters. (d)~Planar Fourier components of the $\Gamma$ edges across the gap (left; grey: inside the augmentation spheres of either layer; green: the admissible band and the planes used) and the Bardeen coupling of the conduction edge on these planes against the $90^\circ$ bilayer (right; crosses: planes at 1.00 and 2.11\,\AA, inside the spheres).}\label{fig:struct}
\end{figure}

\subsection{Methods}
\paragraph{Settings.} PAW data set \texttt{PAW\_PBE P 06Sep2000}; 500\,eV cutoff (twice ENMAX; 650\,eV as a check), PREC = Accurate, LASPH, no real-space projection, Gaussian smearing 0.01\,eV, EDIFF $10^{-9}$\,eV (96-atom cells $10^{-8}$), no symmetrization of wavefunction runs (every $k$ point computed). Monolayer: $16\times12$ $\Gamma$-centred mesh (0.12\,\AA$^{-1}$), 20\,\AA\ vacuum (40\,\AA\ as a check), 24 bands for 10 occupied. Polar grids are non-self-consistent runs on the SCF density: $\Gamma$ and 12 radii from 0.001 to 0.16\,\AA$^{-1}$ $\times$ 48 angles (577 $k$ points), and, for spin--orbit coupling, 10 radii from $10^{-4}$ to $5\times10^{-3}$\,\AA$^{-1}$ $\times$ 48 angles, which bracket both spin--orbit crossovers.
\paragraph{Geometry.} A PBE+D3(BJ) relaxation of bulk black phosphorus compresses it (armchair 4.183\,\AA\ against 4.376, stacking 10.200\,\AA\ against 10.478), and the phosphorene bands are strain sensitive. We therefore use the experimental lattice ($a=3.3136$, armchair 4.3763\,\AA)~\cite{BrownRundqvist1965} and relax only the internal coordinates of the monolayer in PBE (forces below $10^{-3}$\,eV/\AA): bonds 2.221 and 2.248\,\AA\ (experiment 2.224 and 2.244), angles 96.5 and 102.1$^\circ$, puckering 2.1335\,\AA\ (2.1308). Bilayers are rigid pairs of these layers at the experimental layer distance $D=b/2=5.239$\,\AA\ (gap between the facing atomic planes 3.105\,\AA) and at $D\pm0.1$\,\AA. The $90^\circ$ bilayer is a $4\times3$ on $3\times4$ cell (96 atoms, 13.19\,\AA\ square) with a common C$_{2z}$ axis and strains $-0.47\%$ (zigzag) and $+0.48\%$ (armchair) in both layers; a monolayer with these strains isolates their effect, and each layer of the cell was also computed alone (48 atoms, same cell, positions, cutoff and mesh) as the reference of the projections below.
\paragraph{Face data, planes and coupling.} For every band-edge state the tool reads the planar Fourier component $\mathbf G_\parallel=0$ of the Bloch state and its $z$ derivative on a plane in the gap (Sec.~\ref{sec:validation}, T0-m), normalized as $\Omega^{-1/2}$, and forms the Bardeen coupling $T=(\hbar^2/2m)A[F_{\rm t}^\dagger F'_{\rm b}-F'^\dagger_{\rm t}F_{\rm b}]$ between the top face of one layer and the bottom face of the next at the same plane. The pseudo-wavefunction equals the all-electron one only outside the augmentation spheres (P: RCORE 1.90\,bohr $=1.005$\,\AA), so a face plane must lie more than 1.005\,\AA\ from the atoms of both layers: at the 3.1055\,\AA\ gap the admissible band is $1.005<h<2.100$\,\AA. We use five planes in it, $h=1.10$, 1.30, 1.55275 (the midplane), 1.80 and 2.00\,\AA; planes at 1.00 and 2.11\,\AA\ lie at or inside the spheres and are not used. The index and the coupling order follow from \texttt{vk\_matrix\_index} with this kernel in eV (gauge alignment: bond alignment from three starts, keeping the largest coherent bond, a profile-basis invariant).
\paragraph{Target subspace of the $90^\circ$ bilayer.} All 300 computed $\Gamma$ states of the bilayer are projected on the $\Gamma$ conduction and valence edge states of each layer computed alone. Both runs use the same cell, cutoff and plane-wave set, so the overlaps $o_{na}=\langle n|a\rangle$ are sums over common coefficients (pseudo-wavefunctions; the augmentation overlap is not included). The reference states are identified by their weight on their own layer's reciprocal lattice (unfolding weight above 0.9; in the $4\times3$ on $3\times4$ cell a $\Gamma$ state of layer 1 has Miller indices $m_1\equiv0$ mod 4, $m_2\equiv0$ mod 3, and one of layer 2 the reverse); the computed bilayer states carry 98\% of each. From the overlaps we find the bilayer states that carry the bonding and antibonding combinations of the reference pair and measure their coupling in four ways: (i) the $2\times2$ $H/S$ model of the pair restricted to these two states, $S_{ab}=\sum_n o_{na}^*o_{nb}$, $H_{ab}=\sum_nE_no_{na}^*o_{nb}$, with the orthogonalized hopping $t=(h_{ab}-\bar\varepsilon s_{ab})/(1-|s_{ab}|^2)$, which is half the splitting of the two states when they carry the pair symmetrically; (ii) the L\"owdin-downfolded effective Hamiltonian of the orthonormalized pair, $H_{\rm eff}(E)=E-G_{PP}(E)^{-1}$ with $G_{PP}(E)=\sum_nP|n\rangle\langle n|P/(E-E_n)$, at the energies $E_\pm$ of the two states, whose energy dependence measures how far the pair is from an isolated two-level system; (iii) the direct matrix element of the bilayer Hamiltonian between the orthonormalized reference states, over all computed states; (iv) half the distance between the spectral centroids of the bonding and antibonding combinations within $\pm2$\,eV. Separate runs have different energy zeros, so only bilayer energies enter.
\paragraph{Spin--orbit coupling.} We use the noncollinear binary, \texttt{vasp\_ncl}, with spin--orbit coupling on, the spin quantization axis along $z$, zero magnetic moments and the GGA compatibility switch off. The binary was not validated on the cluster; we validated it before use (below).

\subsection{Checks of the toolchain}
\begin{itemize}
\item \emph{Wavefunction reader on real data.} The planar density built from the wavefunction file agrees with the code's charge density on the 114 vacuum planes at least 1.5\,\AA\ from the atoms to $8.7\times10^{-4}$ (tolerance $2\times10^{-3}$); the charge density integrates to 20.0000 electrons.
\item \emph{\texttt{vasp\_ncl} without spin--orbit coupling against \texttt{vasp\_std}} (same inputs): total energies agree to $5\times10^{-5}$\,meV/atom, forces to $5\times10^{-6}$\,eV/\AA, the 16 lowest $\Gamma$ levels to 0.23\,meV, each scalar level appearing as a spinor pair degenerate to 0.007\,meV at all 192 $k$ points.
\item \emph{One against four ranks and a real spin--orbit reference} (MoS$_2$ monolayer, $6\times6$ mesh): identical energies, forces and levels; valence spin--orbit splitting at K of 148.0\,meV (PBE literature 146--148\,meV); Kramers pairs at the time-reversal-invariant momenta degenerate to 0.04\,meV.
\item \emph{Parallel shapes.} After a fixed number of electronic steps the energy is identical for 4 to 64 ranks in the monolayer (scalar and spin--orbit) and 8-atom bilayer classes. Converged runs on different rank counts agree to $2\times10^{-8}$\,eV in energy and $10^{-6}$\,eV in the $\Gamma$ levels (monolayer on 4 and 16 ranks; 96-atom bilayer on 8, 16 and 128 ranks). The repeat of the $90^\circ$ bilayer with wavefunctions kept follows the first run step by step (identical energy changes at every electronic step).
\item \emph{Profile-basis invariance on real data.} A non-orthogonal real change of the (value, slope) basis of the spin--orbit hole faces (condition number 5.2) with the contragredient kernel changes the Born rates by at most $|\Delta\ln\gamma|=4\times10^{-9}$ with the bond alignment, against 0.33 with a Euclidean alignment (up to a factor 5.4 for another basis change); the conduction data are insensitive to the choice.
\end{itemize}

\subsection{Monolayer: representations, index and coupling}
The PBE gap at $\Gamma$ is 0.664\,eV at the experimental lattice. Of the eight $\Gamma$ bands between $-5.9$ and $+1.2$\,eV only two have a forward face amplitude: the conduction edge ($|F|=0.0800$\,\AA$^{-3/2}$ on the midplane) and a valence band 1.9\,eV below the edge (0.054\,\AA$^{-3/2}$). The valence edge has $|F|<10^{-10}$: it is odd under the C$_{2z}$ rotation, as in the models of Sec.~\ref{sec:tbexamples}; its weight sits in the first armchair and zigzag shells (Fig.~\ref{fig:struct}d). On the polar grid, on each of the five admissible planes:
\begin{center}\small
\begin{tabular}{lccccc}
\toprule
plane $h$ (\AA) & 1.10 & 1.30 & 1.553 (mid) & 1.80 & 2.00\\
\midrule
conduction (A): $r$, $\kappa$ & 0.00, 2.00 & 0.00, 2.00 & 0.00, 2.00 & 0.00, 2.00 & 0.00, 2.00\\
valence (B): $r$, $\kappa$ & 2.00, 2.00 & 2.00, 2.00 & 2.00, 2.00 & 2.00, 2.00 & 2.00, 2.00\\
$|\bar T_c(0)|$ (eV) & 0.199 & 0.306 & 0.352 & 0.308 & 0.203\\
\bottomrule
\end{tabular}
\end{center}
The conduction index is identified (local $\kappa$ 1.99--2.00 below 0.008\,\AA$^{-1}$, 1.94 at 0.02 and 1.87 at 0.03\,\AA$^{-1}$, gauge-invariant C$_2$ check $6\times10^{-8}$) and the same on every plane to $10^{-3}$. The valence forward amplitude grows as $0.31\,p$\,\AA$^{-3/2}$ per \AA$^{-1}$ on the midplane (0.36--0.22 from 1.10 to 2.00\,\AA), with a node along zigzag; $r=\kappa=2$ is the coupling order of a B-type edge whose Born expansion has an $O(1)$ twist spread at $p\to0$, so this edge is not protected, a symmetry statement, while its transport is computed exactly (Sec.~\ref{sec:thicknessdft}). The contrast between the two edges at twisted interfaces was found in GW calculations of 90$^\circ$ bilayers and traced to the sign structure of the $\Gamma$ wavefunctions, whose values on a plane 1.6\,\AA\ above the atoms keep one sign along armchair for the conduction edge and alternate for the valence edge~\cite{Cao2016}; twisted bilayers and heterostructures show it in gating and optics~\cite{Sevik2017,Zhao2021}, and twist controls vertical tunnelling~\cite{Srivastava2021}. The first-principles monolayer thus gives the black-phosphorus classes of Table~I of the letter: conduction electrons $r=0$, $\kappa=2$ ($g\propto N^{-1/2}$); valence holes $r=\kappa=2$, no protected axis.

On the midplane the conduction coupling is $|\bar T_c(0)|=0.352$\,eV (valence: zero). Within the admissible band it falls to 0.199 and 0.203\,eV at 1.10 and 2.00\,\AA\ (symmetric about the midplane to 2\%): Bardeen's form is plane independent only for states that obey the same vacuum equation in the gap, and at a 3.1\,\AA\ gap they do not. The index, a symmetry property, does not depend on the plane; the scale does, by a factor of 1.77 over the admissible band, and is fixed by the $90^\circ$ bilayer below. At the midplane the first-shell armchair umklapp component of the conduction edge ($|\mathbf G|=1.436$\,\AA$^{-1}$) is 0.35 of the forward one, the zigzag one ($1.896$\,\AA$^{-1}$) below $10^{-3}$ at small $p$.

\paragraph{Convergence.} The $\Gamma$ face data change by less than 0.3\% between the baseline and runs with a 650\,eV cutoff ($|\bar T_c(0)|$ 0.3530 against 0.3521\,eV), a $22\times16$ mesh (0.3522\,eV) and a 40\,\AA\ vacuum (0.3521\,eV); the gap stays at 0.6644\,eV and the valence amplitude at zero. The strain of the $90^\circ$ cell changes the coupling by 0.7\% (0.3495\,eV) and the gap by 7\,meV. The same SCF on 4 and on 16 ranks gives identical energies, forces and levels. With spin--orbit coupling at 650\,eV (SCF and the fine polar grid on the midplane), the gap changes by 0.1\,meV, $|\bar T_c(0)|$ by 0.25\%, $|\bar T_v(0)|$ by 0.13\%, $p^*$ from $0.943$ to $0.939\times10^{-3}$\,\AA$^{-1}$, while $p_h$ ($1.263\times10^{-3}$\,\AA$^{-1}$) and the hole/electron ratio ($3.20\times10^{-5}$) do not change in the digits quoted; the conduction index at the smallest radius is 1.02 and 1.01.

\subsection{Spin--orbit coupling}\label{sec:socdft}
The spin--orbit polar grids resolve both crossovers from both sides (Fig.~\ref{fig:planes}). The conduction Kramers pair keeps $r=0$; its local index is 1.02--1.03 at $1.2\times10^{-4}$\,\AA$^{-1}$, 1.18--1.27 at $4.6\times10^{-4}$, 1.57--1.69 at $1.1\times10^{-3}$ and 1.94--1.97 at $4\times10^{-3}$\,\AA$^{-1}$ (ranges over the five planes), and 1.92--1.98 over 0.003--0.02\,\AA$^{-1}$ on the coarser grid. Its Born rate is the Kramers $p^2$ term plus the scalar $p^4$ term: a two-term fit over all ten radii (maximum relative residual 1.5--2.9\%) puts the crossover at $p^*=0.94\times10^{-3}$\,\AA$^{-1}$ on the midplane and $(0.73$--$0.94)\times10^{-3}$\,\AA$^{-1}$ over the admissible band; the coarse grid gives the same values to 0.2\%. The asymptotic index is the Kramers value 1 (Proposition), reached below about $2\times10^{-4}$\,\AA$^{-1}$; the scalar value 2 holds above about $3\times10^{-3}$\,\AA$^{-1}$, and in between the local slope is a crossover, not an index. The kit flags such a monotone drift as a crossover and reports the two-term fit.

The valence pair acquires a forward amplitude at $\Gamma$: a fit $|F|^2=F_{\rm so}^2+s^2p^2$ over the smallest radii gives $F_{\rm so}=3.9\times10^{-4}$\,\AA$^{-3/2}$ per band on the midplane ($(2.8$--$5.0)\times10^{-4}$ over the band) against the linear term $s=0.308$ (scalar: 0.307), so the spin--orbit window is $p\lesssim p_h=F_{\rm so}/s=1.26\times10^{-3}$\,\AA$^{-1}$ ($(1.24$--$1.36)\times10^{-3}$ over the band). Inside it the holes form a C$_2$ Kramers pair with $r=0.0$ and $\kappa=1.01$--1.02 from the four smallest radii (local index 1.01 at the smallest, 1.07--1.08 at the third): protected, $\kappa=1$, but with a Born rate $\gamma_B=5.6\times10^{-3}$ per layer already at $p=10^{-4}$\,\AA$^{-1}$, which reaches the Born limit near $p_h$. The hole coupling at $\Gamma$ is $|\bar T_v(0)|=1.13\times10^{-5}$\,eV on the midplane (the same from the SCF and from the grid), $3.2\times10^{-5}$ of the electrons'; over the admissible band the ratio is $(3.2$--$4.5)\times10^{-5}$. Kramers pairs are degenerate at every $k$ point to 0.013\,meV (inversion and time reversal). Both windows lie far below thermal momenta (a few $10^{-2}$\,\AA$^{-1}$ at 300\,K); how far below matters only through the thickness, since the transmitting momenta shrink with $N$ (Sec.~\ref{sec:thicknessdft}). These are properties of the monolayer faces with Bardeen's kernel. $p_h$ and the hole/electron ratio are ratios in which the kernel scale cancels, and their spread over the admissible band ($\pm5\%$ and $\pm17\%$) measures what the choice of plane leaves. $p^*$ is not a pure ratio: it balances on-site against bond terms and rises with the kernel scale. With $|\bar T_c(0)|$ on each admissible plane set to the couplings of the $90^\circ$ bilayer (0.352, 0.538 and 0.756\,eV: the monolayer midplane value, the two-state coupling and the largest downfolded one; Sec.~\ref{sec:t90}), $p^*=(0.94$--$1.65)\times10^{-3}$\,\AA$^{-1}$, and $1.22\times10^{-3}$ on the midplane with the two-state coupling; each plane's own kernel gives $(0.73$--$0.94)\times10^{-3}$ (\texttt{pstar\_scale} in \texttt{analyse\_dft11.py}). The classes ($r=0$, $\kappa\to1$ below $p^*$) do not depend on the scale.

\begin{figure}[!tbp]\centering
\includegraphics[width=\textwidth]{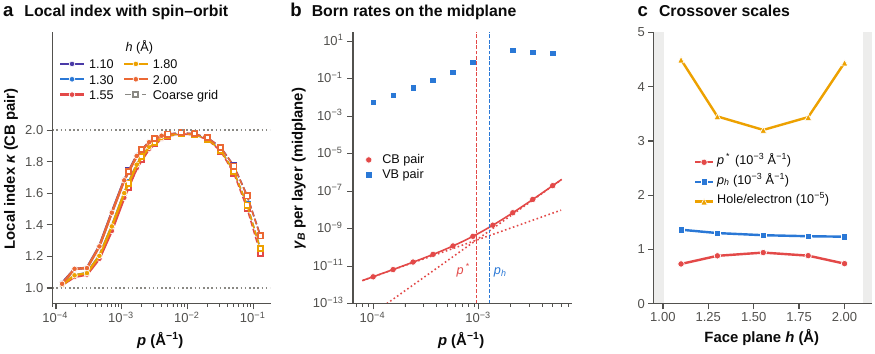}
\caption{Admissible face planes and the spin--orbit crossovers. (a)~Local index of the conduction pair with spin--orbit coupling on the fine and coarse polar grids, five planes; dotted: Kramers value 1 and scalar value 2. (b)~Born rates of the conduction pair (circles) and the valence pair (squares) on the midplane; lines: two-term fit $ap^2+bp^4$ and the hole window. (c)~$p^*$, $p_h$ and the hole/electron coupling ratio over the admissible band.}\label{fig:planes}
\end{figure}

\subsection{The \texorpdfstring{$90^\circ$}{90-degree} bilayer: target subspace and the scale of the coupling}\label{sec:t90}
In the $90^\circ$ bilayer the two layers share no reciprocal-lattice vector shorter than 5.716\,\AA$^{-1}$ (Fig.~\ref{fig:struct}c). On the midplane the planar components of the conduction edges on these vectors are $(2.0$--$3.3)\times10^{-4}$ of their forward ones (three layer distances), and those of the valence edges ($(4.5$--$8)\times10^{-6}$\,\AA$^{-3/2}$, which have no forward component) are equally small, so within this truncation the $\Gamma$ states couple through the forward channel or indirectly through other states; the statement is ``forward-dominated to this accuracy'', not an exact zero.

\paragraph{Identity of the pairs.} At $D=5.239$\,\AA\ the conduction reference pair lies on two bilayer states, 241 and 254, 1.075\,eV apart, which carry 45\% and 36\% of it (Fig.~\ref{fig:t90}a); the bonding combination has 85\% of its weight on 241, the antibonding one 79\% on 254. The other 17\% sits on states 1--16\,eV away (4--5\% each on 214 and 233), and 2\% lies above the 300 computed states. Between 241 and 254 lie twelve states with weights below $5\times10^{-5}$: states folded from other momenta of the two layers, with unfolding weights 0.002--0.06 on either layer's reciprocal lattice (241: 0.78). The lowest of them, 242, 0.760\,eV above 241, is the second of the two lowest empty states, which a selection by equal layer weights would pick; its weight is $2\times10^{-9}$ (Fig.~\ref{fig:t90}d). Equal layer weights cannot tell it apart, since every state of this cell has them. At 5.139 and 5.339\,\AA\ the conduction pairs are 241/256 (1.133\,eV; 78\%) and 241/254 (1.018\,eV; 84\%), and state 242 again carries $10^{-9}$--$10^{-8}$. The valence pair is the two highest occupied states: at 5.239\,\AA, 239 and 240, 3.9\,meV apart, carry 71\% of the reference pair, and 231 and 232, 0.80\,eV lower, another 24\% (5.139\,\AA: 5.9\,meV, 66\% and 25\%; 5.339\,\AA: 2.5\,meV, 75\% and 21\%).

\paragraph{Coupling.} The four measures of the Methods (Fig.~\ref{fig:t90}b), in meV:
\begin{center}\small
\begin{tabular}{lcccccc}
\toprule
$D$ (\AA), pair & $\Delta E$ & two-state $t$ & downfolded at $E_+$, $E_-$ & direct & centroids/2 & Bardeen midplane\\
\midrule
5.139, CB 241/256 & 1133.3 & 566.6 & 858.8, 637.8 & 182.1 & 384.2 & 340.4\\
5.239, CB 241/254 & 1075.0 & 537.5 & 756.1, 596.0 & 201.5 & 386.3 & 349.5\\
5.339, CB 241/254 & 1018.1 & 509.1 & 674.1, 556.9 & 214.5 & 383.4 & 353.8\\
5.139, VB 239/240 & 5.9 & 2.97 & 4.44, 4.34 & 0.12 & -- & 0\\
5.239, VB 239/240 & 3.9 & 1.93 & 2.71, 2.67 & 0.17 & -- & 0\\
5.339, VB 239/240 & 2.5 & 1.24 & 1.64, 1.62 & 0.06 & -- & 0\\
\bottomrule
\end{tabular}
\end{center}
The measures differ because the pair is not an isolated two-level system: a fifth of the reference weight lies on other states, and the reference states are not eigenstates of the bilayer. The direct matrix element between the isolated-layer states (0.20\,eV at 5.239\,\AA) is the smallest; the admixture of other states raises the effective coupling at the pair's energies to 0.60--0.76\,eV (downfolded); the two-state value, 0.54\,eV, is half the splitting of the dominant pair. Whether any of these values is the coupling of a multilayer is examined in Sec.~\ref{sec:effmodel}; the transport scans the coupling scale from the monolayer kernel (0.35\,eV) to the largest downfolded value (0.76\,eV) (Sec.~\ref{sec:thicknessdft}). The monolayer Bardeen kernel on the midplane (strained layers) underestimates the coupling by $\lambda=1.54$ (two-state; 1.71--2.16 downfolded) at 5.239\,\AA, by 1.66 (1.87--2.52) at 5.139\,\AA\ and by 1.44 (1.57--1.91) at 5.339\,\AA. Relative to the kernel of the unstrained monolayer, 0.352\,eV, which is $\lambda=1$ in the transport, the ratios at 5.239\,\AA\ are 1.53 and 1.69--2.15. This scale enters $p^*$ (Sec.~\ref{sec:socdft}) and the transport, not the classes. The valence couplings are smaller by more than two orders of magnitude, and their direct element (0.06--0.17\,meV) is a small part of them: the valence residue is an indirect coupling through other states, as the absence of a forward component requires.

\paragraph{Distance dependence.} The two-state coupling falls from 566.6 to 537.5 and 509.1\,meV over 5.139--5.339\,\AA, on one exponential to 0.1\% with $\beta_e=+0.54$\,\AA$^{-1}$ (downfolded at $E_+$: $+1.21$\,\AA$^{-1}$), the sign of tunnelling (Fig.~\ref{fig:t90}c). The folded state 242 would give the opposite sign, $\beta_e=-0.52$\,\AA$^{-1}$, because its distance from 241 grows with $D$ (715, 760 and 794\,meV at 5.139, 5.239 and 5.339\,\AA). The monolayer Bardeen value rises slightly with $D$ ($-0.19$\,\AA$^{-1}$): the planar average $|F_0(z)|$ of the conduction edge turns over about 1.1\,\AA\ above the outer atoms, and the product $|F_0||F_0'|$ that sets the midplane coupling is largest 1.6--1.7\,\AA\ above them, at the midplane of the 3.1\,\AA\ gap. The faces are there not in the tunnelling regime in which Bardeen's form is plane independent; the kernel sets a scale, not a distance law.

\begin{figure}[!tbp]\centering
\includegraphics[width=\textwidth]{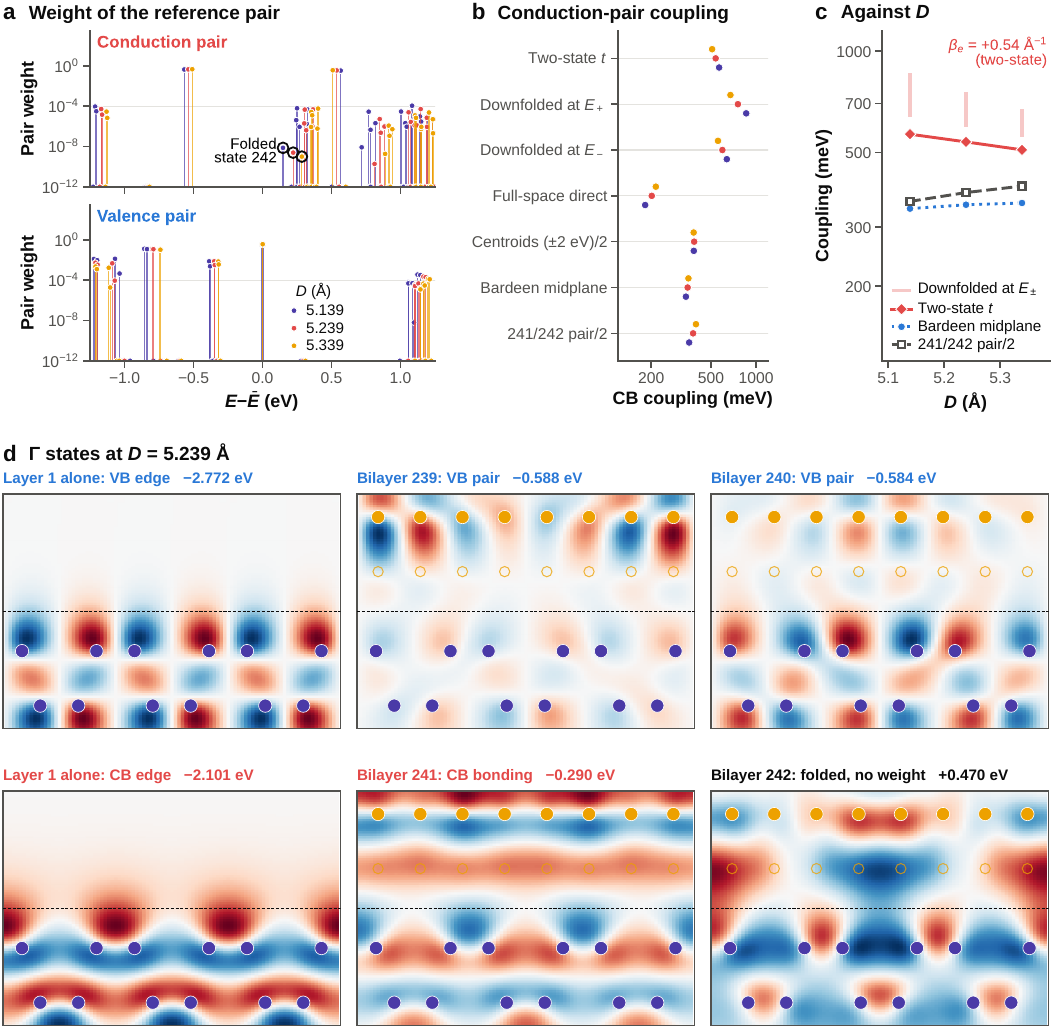}
\caption{Target subspace of the $90^\circ$ bilayer. (a)~Weight of the reference pair (the $\Gamma$ conduction or valence edge states of the two layers, each computed alone) on every bilayer $\Gamma$ state (log scale), against the energy from the centroid of the dominant pair, at the three layer distances; circles: band 242 at 5.239\,\AA, a zone-folded state without weight that equal layer weights would select. (b)~The coupling of the conduction pair measured four ways (two-state, downfolded at the energies $E_\pm$ of the bonding and antibonding states, direct, centroids), the monolayer Bardeen midplane value and half the 241/242 splitting. (c)~Two-state coupling (band: downfolded), Bardeen midplane value and half the 241/242 splitting against the layer distance. (d)~$\Gamma$ states on the $yz$ plane at $D=5.239$\,\AA\ (Re\,$\psi$; violet: layer 1, amber: layer 2; filled atoms within 0.9\,\AA\ of the plane, open within 2.2\,\AA; dashed: midplane): the valence and conduction edge states of layer 1 computed alone, the bilayer valence pair (239, 240), the conduction bonding state (241) and the folded state 242.}\label{fig:t90}
\end{figure}

\subsection{One coupling in the transport window? Effective model and PAW metric}\label{sec:effmodel}
\paragraph{Energy dependence.} The downfolded pair Hamiltonian $H_{\rm eff}(E)=E-G_{PP}(E)^{-1}$ is exact for the computed spectrum at every energy, so it can be evaluated where a multilayer's carriers sit (\texttt{blackphosphorus/\allowbreak effmodel12.py}; Fig.~\ref{fig:effmodel}). Between the two pair states it changes by 0.16\,eV in its hopping ($|t|=0.596$--0.756\,eV) and by 0.28\,eV in its on-site energy ($\varepsilon=0.19$--0.47\,eV on the bilayer's energy scale, 0.48--0.76\,eV above state 241) at 5.239\,\AA; the two-state model ($\varepsilon=0.25$\,eV on that scale, $|t|=0.538$\,eV) lies in between. A stack built by repeating this pair would have at $\bp=0$ the band $\varepsilon\pm2|t|$, with its bottom at $E_b=\varepsilon(E_b)-2|t(E_b)|$: 0.90, 0.91 and 0.67\,eV below the bonding state at 5.139, 5.239 and 5.339\,\AA, outside the spectrum that fixes the pair and 0.36--0.62\,eV below the bilayer's PBE valence maximum, among occupied states of the same symmetry that carry pair weight (band 233: 3.0--4.2\%; band 234: 0.4--0.5\%). There $H_{\rm eff}$ has a pole inside the window $E_b\ldots E_b+4k_BT$ (300\,K) at all three distances, and even away from the pole the two-state model misses the on-site energy by 0.17--0.62\,eV and the hopping by 0.25--0.44\,eV at $E_b$ (0.48--0.62 and 0.32--0.39\,eV across the window at 5.339\,\AA, more than 10\,meV from its pole). This periodic extrapolation is itself not controlled: a bilayer self-energy counts one neighbour of each layer where a stack has two, and it misses the paths through remote states of the next layer but one. In a two-orbital stack whose bilayer is reproduced exactly by its downfolded pair (remote level 1.5 below the target, couplings 0.5, 0.4 and 0.3 in the same units; \texttt{blackphosphorus/\allowbreak effmodel13\_toy.py}), the replicated chain puts the band bottom at $-1.31$, among the remote levels, whereas the target band of the true stack starts at $-0.58$. What the data establish is therefore negative: no single constant from the bilayer calibrates the coupling of a stack, and a unified multilayer model (a trilayer, or a periodic stack computed directly) is needed; whether a quasiparticle treatment is also needed cannot be decided from this extrapolation. The transport therefore treats the coupling scale $\lambda$ as a sensitivity parameter, $\lambda=1$, 1.526 and 2.147 (0.352, 0.538 and 0.756\,eV at $\bp=0$), not as a calibrated value with an error bar. The classes ($r$, $\kappa$) and the transparent-energy threshold of the holes (Sec.~\ref{sec:thicknessdft}) are statements about $\lambda$, not about one value.

\paragraph{Pseudo spectral surrogates.} All projection weights, $H/S$ models, downfolded and direct couplings of Secs.~\ref{sec:t90} and~\ref{sec:effmodel} are built from normalized pseudo-wavefunction overlaps in the common plane-wave set, without the PAW augmentation overlap between the two runs. The pseudo norms of the pair states are 1.052--1.058 (1.05--1.10 for all states in the window), so the augmentation changes norms by 5--10\%. The identity of the pairs, which rests on weight ratios of $10^8$, is not affected. The weights (81\%), the completeness (98\%) and the coupling values are those of the pseudo-wavefunction surrogate; the physical cross-overlap needs the projector coefficients of both runs. The monolayer face data are read outside the augmentation spheres, where pseudo and all-electron wavefunctions coincide (Sec.~\ref{sec:dft}).

\begin{figure}[!tbp]\centering
\includegraphics[width=\textwidth]{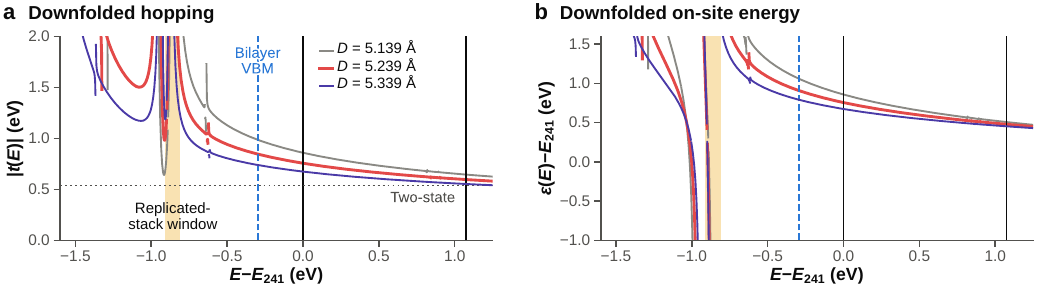}
\caption{Energy dependence of the downfolded conduction pair of the $90^\circ$ bilayer (\texttt{effmodel12.py}; pseudo-wavefunction overlaps). (a)~Hopping $|t(E)|$ and (b)~on-site energy $\varepsilon(E)$, relative to the bonding state 241, at the three layer distances; vertical lines: the pair states 241 and 254 and the bilayer PBE valence maximum (dashed; 5.239\,\AA); amber: the window $E_b\ldots E_b+4k_BT$ (300\,K) of the replicated-pair surrogate at 5.239\,\AA\ (not a calibrated multilayer); dotted: the two-state coupling. Spikes: poles of $H_{\rm eff}$ from occupied states that carry pair weight (bands 233, 234).}\label{fig:effmodel}
\end{figure}

\paragraph{Aligned bilayers.} Rigid bilayers at $D=5.239$\,\AA, with the $\Gamma$ pair taken as the two highest occupied and two lowest empty states (every state has weight 0.50 on each layer, since each stacking exchanges the layers by a symmetry):
\begin{center}\small
\begin{tabular}{lccc}
\toprule
bilayer & valence pair (meV) & conduction pair (meV) & coupled components\\
\midrule
AB (bulk stacking) & 817.8 & 570.6 & all $\mathbf G_\parallel$ (registry phases)\\
AA & 1100.0 & 601.0 & all $\mathbf G_\parallel$\\
$90^\circ$ (96 atoms) & 3.86 & 759.9 (projected pair: 1075.0) & forward-dominated\\
\bottomrule
\end{tabular}
\end{center}
Equal layer weights do not identify a pair: in the folded $90^\circ$ cell states from other momenta have them too (state 242 is one of them), so the pairs are identified by projection. (The AB and AA pairs were not projected; their cells are primitive, so no state is folded to $\Gamma$.) The valence splitting falls from 5.9 to 2.5\,meV over $D=5.139$--5.339\,\AA, a decay rate of 4.4\,\AA$^{-1}$, more than three times that of the first-order valence splitting of the AB bilayer (1.3\,\AA$^{-1}$) and eight times that of the forward coupling of the conduction pair ($\beta_e=0.54$\,\AA$^{-1}$). The residue is therefore compatible with the neglected reciprocal channels and with interface-potential corrections. The projection shows how it arises: the direct element between the isolated-layer valence states is only 3--9\% of the two-state and downfolded couplings, so the residue is carried by the admixture of other states (24\% of the reference weight lies on 231 and 232); which interlayer path dominates that admixture, and at which order, is not decomposed. In a model of isolated layers joined by interlayer hopping alone, two hops return to the same layer and give a diagonal shift, so an off-diagonal coupling needs odd orders or intralayer scattering by the interface potential. Half the residue, 1.93\,meV, is two orders of magnitude above the spin--orbit forward coupling of the holes (0.011\,meV with the monolayer kernel, 0.017--0.024\,meV with the bilayer scale): the hole/electron ratio of Sec.~\ref{sec:socdft} bounds the forward sector only, not all leakage channels of a real interface.
Aligned bilayers split the valence pair strongly (it decays with $D$ at 1.3\,\AA$^{-1}$ in AB, from 930 to 716\,meV over $D=5.139$--5.339\,\AA) because its first-shell armchair component at the gap midplane, 0.038\,\AA$^{-3/2}$, is half the conduction edge's forward amplitude. Their conduction pair (568--601\,meV) lies within 25--160\,meV of the next conduction states and is not a clean two-level splitting. Sampling: the $3\times3$ mesh of the 96-atom cell samples each layer like a $12\times9$ mesh of the primitive cell; the AB bilayer on that mesh reproduces its $16\times12$ splittings to 0.1\,meV.

\subsection{Random-twist films: exact transport of the full face Hamiltonian}\label{sec:thicknessdft}
The thickness dependence follows from one exact transport calculation of the full face Hamiltonian (\texttt{kit/\allowbreak vk\_face\_chain.py}, driver \texttt{blackphosphorus/\allowbreak thickness12.py}), not from Lyapunov exponents: $e^{-2\gamma N}$ is the typical decay, not the mean; continuous twists keep the two mirror nodes of the hole faces at zero probability, which a grid of angles would not; and the bond is random.

\paragraph{Model.} Each channel $\bp$ is a scalar chain (one state per layer; spin--orbit coupling is treated separately below) with on-site energies $E(R_{\theta_l}^{-1}\bp)$ from the PBE band and bonds
\[
t_l=\lambda\,\big|T(R_{\theta_l}^{-1}\bp,R_{\theta_{l+1}}^{-1}\bp)\big|,\qquad T(\mathbf q,\mathbf q')=\tfrac{\hbar^2}{2m}A\big[F_{\rm t}(\mathbf q)F'_{\rm b}(\mathbf q')-F'_{\rm t}(\mathbf q)F_{\rm b}(\mathbf q')\big],
\]
between the top face of layer $l$ and the bottom face of layer $l+1$, from the midplane face data of Sec.~\ref{sec:dft} (value and slope; top and bottom faces related by the mirror, $F_{\rm t}=-F_{\rm b}$, $F'_{\rm t}=F'_{\rm b}$). A scalar chain is gauge invariant, so $|T|$ suffices. At every grid momentum the faces are brought to a real gauge (phase of the larger bottom component; remainder $6\times10^{-8}$), with the sign fixed by the symmetry of the edge (A-type: the sign of the bottom value; B-type: that sign times ${\rm sgn}\,q_y$), and represented, like the band, by even polynomials of degree 10 in $(q_x,q_y)$, times $q_y$ for the B-type edge (degree 11, odd in $q_y$). Fit residuals: $(1.5$--$1.9)\times10^{-4}$ (electrons) and $(3.4$--$4.2)\times10^{-4}$ (holes) of the largest face amplitude, 0.06\,meV for the band. Twists are continuous: no symmetry node carries a finite probability, which removes the node artefact of the 48-angle grid. The leads are semi-infinite aligned stacks; twists are i.i.d.\ uniform. The coupling scale $\lambda$ takes the three values of Sec.~\ref{sec:effmodel}.

\paragraph{Measure and statistics.} Channels $|\bp|\le0.16$\,\AA$^{-1}$ (48 radii dense near 0, 16 angles over $[0,\pi)$, using the C$_2$ equivalence of $\pm\bp$). Beyond the grid lies about 1.4\% of the two-dimensional Boltzmann occupation of the electron band at 300\,K and none at 77\,K; the zigzag hole band is nearly flat (it rises 18\,meV to 0.16\,\AA$^{-1}$), so for holes an estimated 12\% (300\,K) and 0.7\% (77\,K) lies beyond it (\texttt{truncation12.py}; power-law extrapolation of the band per angle). These are occupation weights of the in-plane band, without the lead bandwidths that enter the conductance measure, not conductance errors. For the channels beyond the grid $E^*$ lies outside the averaged mapped band at $\lambda\le2.15$ (it enters only above $\lambda=2.75$ at 0.16\,\AA$^{-1}$; \texttt{threshold13.py}) and localize within a few layers. The polynomial faces are not used beyond the grid: extrapolated to 0.20\,\AA$^{-1}$ they destroy the conductance (\texttt{conv12.py}). Energies: 32 per channel at equal Boltzmann probability inside the aligned lead band (capped at $12k_BT$), normalized to the aligned stack ($T=1$ in every open channel). For holes a geometric grid of 40 points per side around the transparent energy $E^*$ (from $2\times10^{-2}$ to $10^{-10}$ of the band width, with exact Boltzmann weights) is added, because above $\lambda_c$ the hole mean is carried by a window around $E^*$ that shrinks as $N^{-2}$ and that the equal-probability points do not resolve (Table~\ref{tab:conv12}: 32, 128 and 256 such points give $g(128)=4.6$, 3.4 and $1.5\times10^{-5}$; the refined grid gives $1.99\times10^{-5}$ for every refinement). One recursive pass per sample gives $N=1,2,4,\dots$; a sample is one twist sequence summed over channels and energies. We report the mean over samples (electrons 32 sequences to $N=4096$ and 16 to $N=32768$; holes 512 to $N=128$), the median and the typical channel $e^{\langle\ln T\rangle}$ over channels and energies. The mean and the typical value are different statistics ($\langle T\rangle\ge e^{\langle\ln T\rangle}$), and $e^{-2\gamma N}$ is the typical decay, not a bound on the mean. The recursion reproduces dense Caroli inversions for one and two layers to $10^{-8}$, gives $g=1$ for aligned stacks, and the refined quadrature conserves every channel's Boltzmann weight and integrates a cusp at $E^*$ to 1\% (self-tests of \texttt{vk\_face\_chain.py}).

\input{tab_transport25}

\paragraph{Electrons.} Table~\ref{tab:e12} and Fig.~\ref{fig:transport}(a,b). At 300\,K the mean falls from 0.82--0.83 at one layer to 0.32--0.36 at 16 and 0.23--0.24 at 128, with little dependence on $\lambda$ (4\% at 128 layers), and its local exponent grows through $-0.2$ at 16--128 layers to $-0.40$ to $-0.43$ at 2900 layers and $-0.47$ at $2.3\times10^4$ layers (16 sequences): the approach to the axis law $N^{-1/2}$ of Eq.~(2) of the letter ($\kappa=2$, $r=s=0$). At 77\,K the decay is slower, with local exponents $-0.30$ to $-0.35$ at 2900 layers. The early drop (local exponents $-0.37$ to $-0.48$ at 3--6 layers) is the localization of the thermal channels at larger $p$; afterwards the conductance is carried by channels near the axis, which is why the typical channel transmits only $(1.6$--$2.5)\times10^{-3}$ at 16 layers (300\,K) while the mean is 0.32--0.36. The sample distribution is narrow (median/mean 0.98--0.99 at 16 layers; 16--84\% range within about $\pm10\%$ of the mean), because each sample sums many channels.

\paragraph{Holes and the transparent energy.} Table~\ref{tab:h12}, Figs.~\ref{fig:transport}(c,d), \ref{fig:holes} and~\ref{fig:kernel}. The valence faces are odd in $q_y$ (node along the layer's zigzag axis $x$), so both the bond and the rotational on-site disorder scale as $p^2$ and the chain has no protected axis ($\kappa=r$). For $\lambda=1$ the mean falls exponentially (typical and mean alike); at $N\ge32$ it is carried by single rare sequences and is not resolved (Table~\ref{tab:h12}; \texttt{lam1tail13.py}: with 64 sequences the largest holds 90--100\% of the mean at $N\ge32$, and $g(64)$ changes by two orders of magnitude between 112, 224 and 512 energies per channel). For $\lambda=1.53$ and 2.15 it decays algebraically instead, because of a second zero set. At small $p$ the bond factorizes, $t_l=t_0s_ls_{l+1}$ with $s_l=|\sin\alpha_l|$ ($\alpha_l$ the angle of $\bp$ in the frame of layer $l$), and the parabolic on-site energy is $e_l=E^*+2\Delta s_l^2$, where $E^*$ is the on-site energy of a layer whose node faces $\bp$ and $2\Delta$ the anisotropy of the band at $|\bp|$. With $\phi_l=s_l\psi_l$ the chain becomes $[(E-E^*)/s_l^2-2\Delta]\phi_l=t_0(\phi_{l+1}+\phi_{l-1})$: at $E=E^*$ the twist disorder of on-site energy and bond cancels exactly, through a local transformation of the kind known for correlated diagonal and off-diagonal disorder~\cite{Flores1989,Tessieri2001}, a relative of the random dimer~\cite{Dunlap1990}, and the chain is ordered; it transmits when its band contains $E^*$, $\Delta<t_0$, and between matched (factorized) leads $T(E^*)=1$ for every sequence. With $t_0\propto\lambda$ this defines a threshold: from the polynomial coefficients, $\lambda_c=|e_{01}-e_{10}|/2|K_0|=1.373$ as $p\to0$ (\texttt{residue14.py}). At finite $p$ the averaged mapped chain (below) contains $E^*$ for $\lambda$ above 1.405, 1.50, 1.56 and 1.81 at 0.02, 0.04, 0.05 and 0.08\,\AA$^{-1}$ (\texttt{threshold13.py}; Fig.~\ref{fig:holes}d): at $\lambda=1.53$, $E^*$ lies inside the propagating band of the averaged mapped chain for $p<0.045$\,\AA$^{-1}$, at 2.15 for $p<0.11$\,\AA$^{-1}$. This band is the reference of a near-transparent resonance, not an exact transparency: at finite $p$ the random chain keeps the residue below. The ratio $\Delta/t_0(\lambda{=}1)$ of the band extrema and the largest self-bond (1.38--1.44 for $p\le0.02$\,\AA$^{-1}$, 1.62 and 2.26 at 0.04 and 0.08) is a simpler diagnostic of the leading factorization, which overestimates the threshold at larger $p$. Near $E^*$ the disorder $V_l=(E-E^*)/t_0s_l^2$ has the exact tail $P(V>v)=(2/\pi)\arcsin(\delta/v)^{1/2}$: a heavy tail in the domain of attraction of the one-sided stable law of index $\frac12$ (Corollary~\ref{cor:heavy}), so $\gamma\propto|E-E^*|^{1/2}$: fitted exponents 0.50 and 0.51 (Fig.~\ref{fig:holes}b). Its mean-transmission kernel is not $\Phi_1$ (the ratio of $1-\langle T\rangle$ to $-\langle\ln T\rangle$ is $\frac14$ in the dilute limit), and for the ideal factorized chain the means collapse in $q=N|\delta|^{1/2}$ from 256 to 4096 layers with a constant integral, so that one channel transmits $\propto N^{-2}$ (Sec.~\ref{sec:theorem}, Tests; Fig.~\ref{fig:kernel}a).

\paragraph{The residue.} The cancellation is exact only for the parabolic band and the linear node. At $E^*$ the mapped chain is nevertheless exact for the full model: its on-site energy $(e_l-E^*)/s_l^2$ and bond $t_l/s_ls_{l+1}$ are smooth in the angles, because the band is even in $q_y$ and the face odd. Its disorder is correlated (each bond depends on two angles) and of relative order $p^2$, and its Born exponent~\cite{Tessieri2001} is $\gamma(E^*)=C_\lambda p^4[1+O(p^2)]$, with $C_\lambda=|B|^2/64\sin^2k$ from the coefficients of the band and face polynomials (letter, Eq.~(8)): $C_\lambda=3546$\,\AA$^4$ ($\lambda=1.53$, $\cos k=-0.900$) and 576\,\AA$^4$ ($\lambda=2.15$, $\cos k=-0.640$). The exact Born sum at finite $p$ and the Lyapunov exponent of chains long enough that $N\gamma=190$--630 (up to $5.4\times10^8$ layers, 16 sequences with independent random streams; \texttt{residue14.py}) agree with each other within 2.1\% for $p\le0.02$\,\AA$^{-1}$ at $\lambda=1.53$ (4.7\% at 0.028, where $E^*$ approaches the band edge) and within 1.3\% for $p\le0.057$\,\AA$^{-1}$ at $\lambda=2.15$, and $\gamma/p^4$ tends to $C_\lambda$ as $p\to0$ (Table~\ref{tab:residue}, Fig.~\ref{fig:holes}c). The corrections are strong at $\lambda=1.53$, where $E^*$ lies near the band edge: $\gamma/p^4$ rises to $1.2C_\lambda$ at 0.02, $1.6C_\lambda$ at 0.028 and 4--4.7$C_\lambda$ at 0.04\,\AA$^{-1}$ (chain and Born sum differ by 13\% there), and beyond 0.045\,\AA$^{-1}$ $E^*$ leaves the averaged mapped band; at $\lambda=2.15$ they stay within 3\% to 0.03\,\AA$^{-1}$.

\paragraph{The thermal law.} The zero set of the thermal problem is the point $p=0$, $E=E^*$, with the residue $\propto p^4$ in one direction and the heavy-tailed window $|E-E^*|^{1/2}/p$ in the other: a mixed-order zero set (Corollary~\ref{cor:mixed}). With $u=C_\lambda Np^4$ and $v=N|\delta|^{1/2}$, $\delta=(E-E^*)/Cp^2$, $C=\lambda|K_0|$, the scaling limit $\langle T_N\rangle\to\Psi(u,v,\pm)$ gives, through $\dd^2p\,\dd E=p\,\dd p\,\dd\varphi\,Cp^2\dd\delta$, $g\simeq A_{\rm th}N^{-3}$ with $A_{\rm th}=\pi Cw^*I/2C_\lambda Z$, $I=\sum_\pm\int\dd u\int2v\,\dd v\,\Psi$ (the azimuthal integral assumes that $\Psi$ does not depend on the direction of $\bp$ relative to the leads, which the check at a second lead azimuth below supports), $w^*$ the Boltzmann density at $E^*$ and $Z$ the aligned normalization of Sec.~\ref{sec:thicknessdft}. We computed $\langle T_N\rangle$ of the full face model at fixed $(u,v)$ with the exact recursion (Table~\ref{tab:collapse}, Fig.~\ref{fig:kernel}b--e): 9 values of $u$ from 0.01 to 8, 12 of $v$ from 0 to 16 and both signs, for $N=256$, 1024 and 4096 ($2^{14}$ sequences per point) and 16384 (4096 sequences; \texttt{collapse14.py}), and 15 points with $u\le4$ and $v\le4$ at $N=65536$ (4096 sequences; \texttt{collapse14b.py}).

\emph{Convergence in $N$.} Every point has its own random streams, so the means are independent: the standard error of a linear functional such as $I$ follows from theirs, which 16 independent batch means per point confirm (within about 30\%, the scatter of such an estimate being about 18\%), and the normalized differences between thicknesses are independent. At $\lambda=2.15$ the means change by a median of $+4.7$, $+2.2$ and $+0.6\%$ per factor 4 in $N$ from 256 to 16384 (points with $\langle T\rangle>10^{-4}$; normalized differences are approximately normal, less so where rare sequences dominate). Between 4096 and 16384 they agree point by point within their errors ($\chi^2=210$ for 207 points), with a small coherent rise (mean normalized difference $+0.23\pm0.07$). The agreement is not carried by poorly sampled points. Measuring how strongly a mean rests on few sequences by $n_{\rm eff}=(\sum T)^2/\sum T^2$ (a concentration measure, not an autocorrelation-based sample size; 55 points at $N=16384$ have $n_{\rm eff}<30$ and 28 have $n_{\rm eff}<10$), the 152 points with $n_{\rm eff}\ge30$ at both thicknesses give $\chi^2=153$, and the central region $u\le1$, $v\le4$ gives $\chi^2=91$ for 90 points; the rise is not carried by poorly sampled points either ($+0.28\pm0.08$ among the well-sampled points, $+0.17\pm0.10$ in the central region; \texttt{diag15.py}). 16384 and 65536 agree as well ($\chi^2=17.9$ for 15 points, all with $n_{\rm eff}>42$; median $+2.8\%$, against $+0.3\%$ from 4096 to 16384 at the same points; mean normalized difference $+0.8\pm0.3$). These $\chi^2$ values are consistency diagnostics against an approximately Gaussian reference, not proofs of a scaling limit: point-by-point agreement and a coherent finite-size drift coexist, the approach to the limit is slow, not monotone at this precision, and continuing, and we find no sign of a different limit. At $\lambda=1.53$, where $E^*$ lies near the band edge, it is slower (median changes $+11.6$, $+9.6$ and $+3.1\%$ from 256 to 16384; $\chi^2=830$ for 207 points between 4096 and 16384 and 45 for 15 between 16384 and 65536), because the $O(p^2)$ corrections of the residue and of the position of $E^*$ in the band decay only as $N^{-1/2}$. The lead azimuth does not matter within the errors ($\varphi=30^\circ$ against $73^\circ$: $\chi^2=8.5$ for 10 points at $N=1024$, $\lambda=2.15$).

\emph{The line $v=0$.} At $v=0$ only the residue scatters: at $\lambda=2.15$ the means lie within 0.4\% of the finite-variance kernel $\Phi_1(2u)$ for $u\le0.3$ and 2--9\% below it at $u=1$--2 for $N=4096$--65536 (at $\lambda=1.53$ the finite-$p$ deviations are larger, 2--3\% already at $u=0.3$); replacing $u$ by $u_{\rm eff}=N\gamma_{\rm Born}(E^*,p)$, the exact residue at the actual $p$, lowers the rms deviation from 0.011 to 0.005 ($\lambda=2.15$) and from 0.057 to 0.013 ($\lambda=1.53$), absolute. The corollary needs the existence of $\Psi$, not its form. As $u\to0$ the kernel tends to the ideal one, $I(u)\to J$.

\emph{Integral and amplitude.} The integral $I$ by the trapezoidal rule on this base grid has sampling errors of 0.5\% ($N\le4096$) and 1\% ($N=16384$). We measured its quadrature error by adding all midpoints of the grid (4096 sequences per added point at $N=4096$, both $\lambda$; 2048 at 16384, $\lambda=2.15$): the refined grid lowers $I$ by $1.1\pm0.5\%$ ($\lambda=2.15$) and $2.4\pm0.5\%$ ($\lambda=1.53$) at $N=4096$ and by $2.7\pm0.9\%$ at 16384, and the Richardson value $I_{\rm fine}+(I_{\rm fine}-I_{\rm base})/3$ removes the leading error. The tails beyond $v=16$ and $u=8$, extrapolated exponentially, would add another 1.1--1.4\% ($\lambda=2.15$) and 0.1--1\% ($\lambda=1.53$) of the base integral. On the base grid the amplitude of the kernel integrated at thickness $N$ is $A_{\rm th}=689$, 804, 846 and 868 at $N=256$, 1024, 4096 and 16384 for $\lambda=2.15$, still rising by $2.6\pm1.1\%$ from 4096 to 16384; the quadrature correction grows with $N$, and the Richardson values, $834\pm5$ and $837\pm6$, agree. Other rules on the refined grid differ from them by at most 0.5\% (tensor-product Simpson $-0.5$ and $+0.3\%$, monotone cubic interpolation $-0.1$ and $+0.5\%$ at $N=4096$ and 16384), and the strip $0<u<0.01$ below the base grid is taken as $0.01\,I(0.01)$, 1.0--1.2\% of $I$. Assuming that the decrease of $I(u)$ observed on the base grid extends to $0<u<0.01$, the endpoint values $I(0.01)$ and the ideal-chain moment $J$ at $u=0$ suggest a sensitivity of 0.02--0.03\% of $I$ at $N=4096$ (\texttt{diag15.py}). This estimate is conditional on that monotonic continuation and is not a rigorous error bound on the strip; the endpoint values also carry sampling and quadrature uncertainty (at $N=16384$, $I(0.01)=15.35\pm0.18$ and $J=15.31$ agree within their errors). Interior samples at $u=0.001$, 0.0025, 0.005 and 0.0075 (both $\lambda$, $N=4096$ and 16384, independent streams; \texttt{strip16.py}) lie within two standard errors of the endpoint range but do not decrease monotonically; the reversals (at most $1.6\sigma$) are within their errors (0.6--1.2\%), and integrating through them changes $I$ by $-0.01$ to $+0.02\%$. The sampled strip agrees with the estimate, which remains a sensitivity check, not a bound. The extrapolated amplitude is thus stable to a few per cent over the tested thicknesses, with residual finite-size and tail-extrapolation uncertainty: the Richardson value removes the grid error, not the finite-thickness drift found above, and the tails rest on an exponential extrapolation, so these numbers are not an error bound on the $N\to\infty$ amplitude. For $\lambda=1.53$ it is 54, 89, $109\pm1$ (refined) and 128 (base grid), still rising. Integrated at the same thickness, the kernel reproduces the directly computed thermal means (\texttt{thick13\_h.py}; 256 sequences to 1024 layers, 300\,K) within 5\%: $N^3g=698$ and $841\pm26$ against 689 and 804 at 256 and 1024 layers ($\lambda=2.15$), 55 and $90\pm3$ against 54 and 89 ($\lambda=1.53$), both on the base grid. The energy quadrature of these means (32 equal-probability energies plus 40 geometric points per side of $E^*$) is adequate to about 3\% at 1024 layers: on the same films, doubling the points around $E^*$ lowers $g(1024)$ by $3.1\pm1.4\%$ ($\lambda=2.15$, 64 films) and $2.2\pm0.9\%$ ($\lambda=1.53$, 32 films), a second doubling by a further $0.6\pm0.6\%$ ($\lambda=2.15$), and $g(N\le512)$ by at most 1.5\% (\texttt{thermref14.py}); with this correction $N^3g(1024)\approx815$ at $\lambda=2.15$, within 3\% of the kernel at the same thickness (base grid, or with the quadrature correction measured at 4096). At these thicknesses the thermal mean is thus carried by the neighbourhood of the mixed-order zero set. Accordingly the local exponents of the thermal means reach $-2.6$ ($\lambda=1.53$) and $-2.8$ to $-2.9$ ($\lambda=2.15$) at 360--720 layers (Table~\ref{tab:h12}, Fig.~\ref{fig:transport}d), approaching $-3$. The law needs the cancellation, and anything that breaks the factorization, such as hole umklapp or the non-forward residue below, cuts it off.

\begin{figure}[!tbp]\centering
\includegraphics[width=\textwidth]{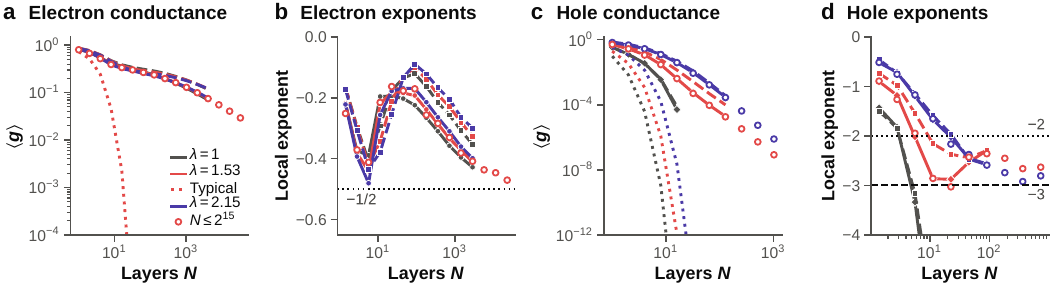}
\caption{Transport of random-twist black-phosphorus films with the full face Hamiltonian (\texttt{thickness12.py}, \texttt{thick13\_h.py}). (a)~Electron mean conductance for the three coupling scales at 300\,K (solid) and 77\,K (dashed), median (thin) and typical channel (dotted) for $\lambda=1.53$ at 300\,K; open circles: 16 sequences to $N=2^{15}$. (b)~Local exponents; dotted: $-1/2$. (c)~Hole mean conductance (diamonds: 300\,K; dashed: 77\,K; $\lambda=1$ only to 16 layers, where it is resolved) and typical channel (dotted); open circles: 256 sequences to 1024 layers (\texttt{thick13\_h.py}, 300\,K). (d)~Hole local exponents; dotted: $-2$ (one channel of the ideal chain), dashed: $-3$ (thermal law of the mixed-order zero set).}\label{fig:transport}
\end{figure}

\begin{figure}[!tbp]\centering
\includegraphics[width=\textwidth]{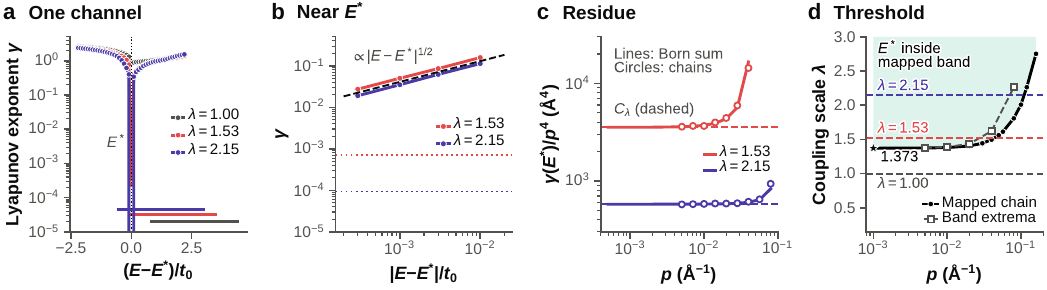}
\caption{Transparent energy of the hole chain (\texttt{holes12.py}, \texttt{residue14.py}, \texttt{threshold13.py}). (a)~Lyapunov exponent of one channel ($p=0.02$\,\AA$^{-1}$, 6000 layers) across the band, against $(E-E^*)/t_0$ for three coupling scales; bars: the aligned lead band. (b)~Near $E^*$ ($4\times10^4$ layers): $\gamma\propto|E-E^*|^{1/2}$; dotted: the converged residue $\gamma(E^*)$ at this $p$ (\texttt{residue14.py}). (c)~Residue $\gamma(E^*)/p^4$ against $p$: exact Born sum of the mapped chain (lines), Lyapunov exponents of chains with $N\gamma=190$--630 (circles; up to $5.4\times10^8$ layers) and the leading coefficient $C_\lambda$ (dashed). (d)~Threshold: the averaged mapped chain contains $E^*$ above the solid line (\texttt{threshold13.py}), with the limit 1.373 (star); squares: the band-extrema ratio $\Delta/t_0(\lambda{=}1)$; dashed: the three scales used.}\label{fig:holes}
\end{figure}
\begin{figure}[!tbp]\centering
\includegraphics[width=\textwidth]{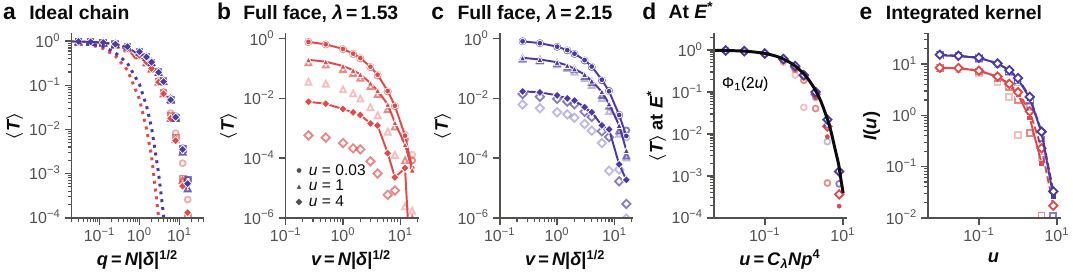}
\caption{Kernels of the transparent energy. (a)~Ideal factorized chain: mean transmission against $q=N|\delta|^{1/2}$ for $N=64$, 256, 1024 and 4096 (markers, light to dark), with $E^*$ at $\cos k=-0.90$ (red) and $-0.64$ (violet); dashed: the dilute limit $1-q/2\sqrt{\sin k}$; dotted: typical transmission ($N=4096$). (b,c)~Full face model: mean transmission at fixed $u=C_\lambda Np^4$ against $v=N|\delta|^{1/2}$ for $N=256$ (open, light), 1024 and 4096 (filled, line). (d)~At $E^*$ ($v=0$) against $u$, with the finite-variance kernel $\Phi_1(2u)$. (e)~The $v$-integrated kernel $I(u)=\sum_\pm\int2v\langle T\rangle\,\dd v$ for the three thicknesses; in (d,e) open diamonds: $N=16384$, 4096 sequences (\texttt{kernel14.py}, \texttt{collapse14.py}; independent random streams for every point).}\label{fig:kernel}
\end{figure}

\paragraph{Scope of the forward-channel results.} (i)~Umklapp: first-shell partners couple with $\rho_U=0.34$--0.37 of the forward strength and are on shell mainly within a characteristic angle of $\pm3.6^\circ$ (300\,K) or $\pm1.8^\circ$ (77\,K) of alignment; a rough estimate of the thickness over which they relax momentum is $N_U\approx190$--220 (300\,K) and 370--430 layers (77\,K) (Sec.~\ref{sec:umklapp}). This is an order-of-magnitude estimate: the on-shell window is the characteristic angle of low-energy channels near the axis of a parabolic band, not a hard boundary for every thermal channel, and the per-layer rate $f_U\rho_U^2$, whose inverse is $N_U$, does not resolve velocities and fluxes. The electron curves beyond $N_U$ are forward-channel results, not predictions for real films, and the window 16--128 layers is below the estimate rather than a controlled bound. (ii)~Spin--orbit coupling, electrons: the Kramers law $N^{-1}$ takes over when the transmitting momenta, which shrink as $N^{-1/4}$, reach $p^*$, i.e.\ at $N^*\approx1/\gamma_B(p^*)$. Evaluated at thermal energies $E_b+uk_BT$ ($u=0.5$--4) above the bottom of the averaged chain band with the two-state scale on the midplane (\texttt{soc12.py}), $p^*=(1.7$--$4.9)\times10^{-4}$\,\AA$^{-1}$ and $N^*=(0.8$--$6)\times10^{10}$ at 300\,K, and $N^*=(0.3$--$2.5)\times10^{11}$ at 77\,K. (iii)~Spin--orbit coupling, holes: channels $p<2p_h$, where the scalar faces vanish and the spin--orbit forward coupling sets the bandwidth, carry at most $5.0\times10^{-6}$ (300\,K) and $2.5\times10^{-5}$ (77\,K) of the aligned hole conductance; each channel transmits at most its aligned weight, so these are bounds on what the Kramers window can add, whatever its disorder. (iv)~Non-forward residue: half the valence splitting of the $90^\circ$ bilayer, 1.93\,meV, is a coupling that the forward faces do not contain. A momentum-independent coupling of this size would exceed the forward hole bond in channels holding 8.1\% (300\,K) and 32\% (77\,K) of the aligned hole conductance. This is a weight estimate for that replacement model, not a bound on what the residue does to the whole hole conductance; it would also break the factorization behind $E^*$. (v)~Umklapp, holes: no perturbative bound exists. The valence forward face vanishes on the axis while its first-shell component does not: the ratio of the first armchair shell to the forward face is 6.6 at 0.02, 3.4 at 0.04 and 1.7 at 0.16\,\AA$^{-1}$ (\texttt{umklapp12.py}). The nearly flat zigzag band (it rises only $0.68\,k_BT$ to 0.16\,\AA$^{-1}$ at 300\,K) opens partners within about $\pm7.7^\circ$ of alignment at 300\,K (at least $\pm6.4^\circ$; mean-curvature extrapolation beyond the grid) and $\pm4.7^\circ$ at 77\,K, i.e.\ at 8.6\% (5.2\%) of random interfaces. At those interfaces the hole coupling is dominated by umklapp; it mixes channels and can break the cancellation behind $E^*$ over a thickness that needs a multichannel treatment (8.6\% is a density of such interfaces, not a decay length). The hole results are therefore forward-channel results. (vi)~Coherence and rigid layers: phase breaking (Sec.~\ref{sec:coherence}) and relaxation at twisted interfaces are not included.

\paragraph{What the transport shows.} Within the forward channel, exactly for the full face Hamiltonian and for every coupling scale scanned, electrons keep an algebraic conductance approaching $N^{-1/2}$, while the hole conductance decays exponentially for $\lambda<\lambda_c$ and algebraically above, approaching $N^{-3}$. Relative to the aligned stacks, twist disorder retains at least 7 times more of the electron than of the hole conductance at 16 layers and $10^2$ times more at 64 (Table~\ref{tab:h12}): these are ratios $g_e/g_h$ of retained fractions, each normalized to its own aligned stack and band edge; an absolute ratio $G_e/G_h=(g_e/g_h)(G_{e,{\rm al}}/G_{h,{\rm al}})$ needs a common chemical-potential or density protocol and the aligned conductances, which we do not claim. The prediction for real films is limited to $N\lesssim N_U$ for electrons; for holes only the absence of a protected axis is a statement about real films, and their magnitudes are forward-channel results (iv, v).

\subsection{What the calculations identify}
The levels of evidence are:
\begin{itemize}
\item \emph{Identified from the monolayer} (symmetry statements, the same on every admissible plane): the representations of the $\Gamma$ edges (A-type conduction, B-type valence); $r=0$ and $\kappa=2$ of the conduction edge (Born regime, controlled); $\kappa=r=2$ of the valence edge, a scaling statement (bond and rotational on-site disorder both $\propto p^2$) that holds beyond Born order, while the Born rates themselves exceed 0.1 per layer and are not used (the kit reports them as not applicable); the spin--orbit classes (conduction $\kappa\to1$ below $p^*$; holes a Kramers pair below $p_h$).
\item \emph{Estimated with Bardeen's kernel on the admissible band} (spread over the band given): $p_h=(1.24$--$1.36)\times10^{-3}$\,\AA$^{-1}$ and the forward hole/electron ratio $(3.2$--$4.5)\times10^{-5}$ (ratios of monolayer data); $p^*=(0.73$--$0.94)\times10^{-3}$\,\AA$^{-1}$ with each plane's own kernel and $(0.94$--$1.65)\times10^{-3}$ with the couplings of the $90^\circ$ bilayer at an energy 23\% of the averaged bandwidth above the band bottom; at thermal energies $(1.7$--$4.9)\times10^{-4}$\,\AA$^{-1}$ (300\,K, two-state scale), with $N^*\approx10^{10}$--$10^{11}$.
\item \emph{Fixed by the $90^\circ$ bilayer} (projected target subspace; pseudo spectral surrogates): the identity of the $\Gamma$ pairs (conduction: the states split by 1.075\,eV at 5.239\,\AA, not the equal-weight pair 241/242; valence: the two highest occupied states), the two-state coupling (0.54\,eV, 1.5 times the monolayer kernel), its decay with $D$ ($\beta_e=+0.54$\,\AA$^{-1}$), the energy dependence of the downfolded pair and the indirect character of the valence residue. \emph{Not fixed}: a constant coupling of a multilayer (no single constant calibrates a stack; the periodic extrapolation of the bilayer pair is uncontrolled, Sec.~\ref{sec:effmodel}; the layer-dependent optical subbands of aligned few-layer phosphorene~\cite{Li2017} constrain a chain of aligned layers, although excitonic effects keep them from fixing electron and hole couplings separately), the PAW cross-overlap, and the path of the valence residue.
\item \emph{Computed exactly for the model} (forward channel, coherent, rigid layers, coupling scale scanned over $\lambda=1$--2.15): finite-$N$ mean conductances of electrons and holes with their medians, typical values and local exponents (Tables~\ref{tab:e12} and~\ref{tab:h12}); the hole transparent energy with its threshold ($\lambda_c=1.373$ as $p\to0$), its residue $C_\lambda p^4$ and its kernels; the thermal law $N^{-3}$ follows under a scaling limit that is supported numerically (at $\lambda=2.15$) but not proved (Corollary~\ref{cor:mixed}).
\item \emph{Estimated, order of magnitude}: the electrons' first-shell umklapp thickness $N_U\sim2\times10^2$ (300\,K) to $4\times10^2$ (77\,K), from computed first-shell components and an order-one transmission factor (for holes no such estimate exists); a bound on what the hole Kramers window can add to the aligned hole conductance ($\le5\times10^{-6}$ at 300\,K, $2.5\times10^{-5}$ at 77\,K), and a weight estimate for a 1.93\,meV non-forward residue (8\% at 300\,K, 32\% at 77\,K; a replacement model, not a bound).
\item \emph{Not identified}: the umklapp transmission factors and a multichannel calculation beyond $N_U$, relaxed twisted interfaces, the phonon inputs $\beta_{\rm eq}$, $F_{zz}(0)$ and $\ell_{\rm ph}$ (T0-b, T1-b), and measured thickness curves.
\end{itemize}

%% file: tab_transport25.tex
\begin{table}[!htbp]\centering\footnotesize
\caption{Electrons: mean conductance $\langle g(N)\rangle$ of random-twist black-phosphorus films relative to aligned stacks (exact transport of the full face Hamiltonian, 32 twist sequences; \texttt{thickness12.py}), local exponent near 2900 layers, median/mean and typical channel $e^{\langle\ln T\rangle}$ at 16 layers. Sampling errors of the mean are below 0.007 for $N\ge16$. With 16 sequences to $N=2^{15}$ ($\lambda=1.53$, 300\,K): $g(8192)=0.0553$, $g(16384)=0.0406$, $g(32768)=0.0293$; local exponents $-0.437$, $-0.446$, $-0.470$ at 5793, 11585, 23170 layers.}\label{tab:e12}
\begin{tabular}{llccccccc}
\toprule
$\lambda$ & $T$ & $g(1)$ & $g(16)$ & $g(128)$ & $g(1024)$ & $g(4096)$ & slope (2900) & median/mean, typical (16)\\
\midrule
1.00 & 300\,K & 0.822 & 0.363 & 0.236 & 0.124 & 0.070 & $-0.428$ & 0.990, $1.6\times10^{-3}$\\
1.00 & 77\,K & 0.864 & 0.386 & 0.286 & 0.184 & 0.117 & $-0.352$ & 0.991, 0.011\\
1.53 & 300\,K & 0.827 & 0.345 & 0.234 & 0.129 & 0.075 & $-0.410$ & 0.985, $2.2\times10^{-3}$\\
1.53 & 77\,K & 0.863 & 0.367 & 0.270 & 0.184 & 0.121 & $-0.327$ & 0.988, 0.013\\
2.15 & 300\,K & 0.825 & 0.323 & 0.226 & 0.131 & 0.077 & $-0.398$ & 0.990, $2.5\times10^{-3}$\\
2.15 & 77\,K & 0.857 & 0.350 & 0.252 & 0.179 & 0.121 & $-0.300$ & 0.985, 0.013\\
\bottomrule
\end{tabular}
\end{table}

\begin{table}[!htbp]\centering\footnotesize
\caption{Holes: mean conductance relative to aligned stacks (512 twist sequences; energies refined around the transparent energy $E^*$), local exponent between 64 and 128 layers, typical channel at 16 layers, and the ratio $g_e/g_h$ of the retained fractions of electrons and holes at 16 and 64 layers (same $\lambda$ and $T$; each normalized to its own aligned stack, not an absolute conductance ratio). At $\lambda=1$ the mean at $N\ge32$ is carried by single sequences (the largest of 64 holds 90--97\% of it at $N=32$ on the three energy grids, and $g(64)$ changes from $1.6\times10^{-17}$ to $2.7\times10^{-15}$ with the energy grid; \texttt{lam1tail13.py}): n.r., not resolved. Last two columns: 256 sequences to 1024 layers (\texttt{thick13\_h.py}, 300\,K); a finer energy grid around $E^*$ lowers $g(1024)$ by 2--3\% (\texttt{thermref14.py}).}\label{tab:h12}
\setlength{\tabcolsep}{2.2pt}
\begin{tabular}{llcccccccccc}
\toprule
$\lambda$ & $T$ & $g(1)$ & $g(16)$ & $g(64)$ & $g(128)$ & slope (90) & typ.\ (16) & $g_e/g_h$ (16) & $g_e/g_h$ (64) & $g(1024)$ & slope (720)\\
\midrule
1.00 & 300\,K & 0.342 & $5.0\times10^{-5}$ & n.r. & n.r. & n.r. & $8.1\times10^{-19}$ & $7.2\times10^{3}$ & n.r. & -- & --\\
1.00 & 77\,K & 0.348 & $7.7\times10^{-5}$ & n.r. & n.r. & n.r. & $4.9\times10^{-17}$ & $5.0\times10^{3}$ & n.r. & -- & --\\
1.53 & 300\,K & 0.504 & $4.0\times10^{-3}$ & $9.5\times10^{-5}$ & $1.9\times10^{-5}$ & $-2.29$ & $1.2\times10^{-12}$ & $8.6\times10^{1}$ & $2.8\times10^{3}$ & $8.4\times10^{-8}$ & $-2.63$\\
1.53 & 77\,K & 0.583 & 0.014 & $4.8\times10^{-4}$ & $9.9\times10^{-5}$ & $-2.30$ & $9.1\times10^{-10}$ & $2.7\times10^{1}$ & $6.0\times10^{2}$ & -- & --\\
2.15 & 300\,K & 0.666 & 0.039 & $1.7\times10^{-3}$ & $2.9\times10^{-4}$ & $-2.56$ & $2.2\times10^{-8}$ & 8.4 & $1.5\times10^{2}$ & $7.8\times10^{-7}$ & $-2.81$\\
2.15 & 77\,K & 0.722 & 0.047 & $2.2\times10^{-3}$ & $3.6\times10^{-4}$ & $-2.58$ & $1.1\times10^{-8}$ & 7.4 & $1.2\times10^{2}$ & -- & --\\
\bottomrule
\end{tabular}
\end{table}

\begin{table}[!htbp]\centering\footnotesize
\caption{Convergence of the transport (two-state scale, 300\,K; the same twist sequences within each block). Electrons: 8 sequences, production settings (48 radial points, 16 angles, 32 energies, polynomial degree 10) against refinements (\texttt{conv12\_electrons.py}). Holes: 64 sequences; equal-probability energies only, and with the geometric refinement around $E^*$ (\texttt{conv12\_holes.py}).}\label{tab:conv12}
\begin{tabular}{llccc}
\toprule
carrier & setting & \multicolumn{3}{c}{$\langle g(N)\rangle$}\\
\midrule
electrons & production & $N=16$: 0.3357 & $N=128$: 0.2360 & $N=1024$: 0.1275\\
electrons & 32 angles & $N=16$: 0.3355 & $N=128$: 0.2360 & $N=1024$: 0.1274\\
electrons & 64 energies & $N=16$: 0.3357 & $N=128$: 0.2363 & $N=1024$: 0.1277\\
electrons & 72 radial points & $N=16$: 0.3356 & $N=128$: 0.2362 & $N=1024$: 0.1276\\
electrons & degree 12 & $N=16$: 0.3357 & $N=128$: 0.2360 & $N=1024$: 0.1275\\
electrons & degree 8 & $N=16$: 0.3210 & $N=128$: 0.2260 & $N=1024$: 0.1223\\
\midrule
holes & equal probability, 32 energies & $N=16$: $4.4\times10^{-3}$ & $N=64$: $1.7\times10^{-4}$ & $N=128$: $4.6\times10^{-5}$\\
holes & equal probability, 128 energies & $N=16$: $4.2\times10^{-3}$ & $N=64$: $1.2\times10^{-4}$ & $N=128$: $3.4\times10^{-5}$\\
holes & equal probability, 256 energies & $N=16$: $4.1\times10^{-3}$ & $N=64$: $8.8\times10^{-5}$ & $N=128$: $1.5\times10^{-5}$\\
holes & refined, 40 per side (production) & $N=16$: $4.1\times10^{-3}$ & $N=64$: $9.5\times10^{-5}$ & $N=128$: $2.0\times10^{-5}$\\
holes & refined, 60 per side & $N=16$: $4.1\times10^{-3}$ & $N=64$: $9.4\times10^{-5}$ & $N=128$: $2.0\times10^{-5}$\\
holes & refined, 64 energies & $N=16$: $4.1\times10^{-3}$ & $N=64$: $9.5\times10^{-5}$ & $N=128$: $2.0\times10^{-5}$\\
holes & refined, 32 angles & $N=16$: $4.1\times10^{-3}$ & $N=64$: $9.5\times10^{-5}$ & $N=128$: $2.0\times10^{-5}$\\
holes & refined, 72 radial points & $N=16$: $4.1\times10^{-3}$ & $N=64$: $9.5\times10^{-5}$ & $N=128$: $2.0\times10^{-5}$\\
\bottomrule
\end{tabular}
\end{table}

\begin{table}[!htbp]\centering\footnotesize
\caption{Residue of the hole transparent energy, $\gamma(E^*)/p^4$ in \AA$^4$: exact Born sum of the mapped chain (correlated on-site and bond disorder) and the Lyapunov exponent of chains of $N$ layers (16 sequences with independent random streams, $N\gamma=193$--627; \texttt{residue14.py}). Leading coefficient $C_\lambda=|B|^2/64\sin^2k$: 3545.7\,\AA$^4$ ($\lambda=1.53$) and 576.4\,\AA$^4$ ($\lambda=2.15$); threshold $\lambda_c(p\to0)=1.3733$. --: $E^*$ outside the band of the averaged mapped chain.}\label{tab:residue}
\begin{tabular}{lcccccc}
\toprule
 & \multicolumn{3}{c}{$\lambda=1.53$} & \multicolumn{3}{c}{$\lambda=2.15$}\\
$p$ (\AA$^{-1}$) & Born & chains & $N$ & Born & chains & $N$\\
\midrule
0.0005 & 3546.1 &  &  & 576.4 &  & \\
0.001 & 3547.3 &  &  & 576.4 &  & \\
0.002 & 3552.1 &  &  & 576.5 &  & \\
0.005 & 3585.9 & 3595$\pm$64 & $2^{27}$ & 576.8 & 574$\pm$8 & $2^{29}$\\
0.00707 & 3627.0 & 3660$\pm$72 & $2^{25}$ & 577.2 & 577$\pm$7 & $2^{28}$\\
0.01 & 3712.6 & 3658$\pm$63 & $2^{23}$ & 577.9 & 581$\pm$8 & $2^{26}$\\
0.01414 & 3898.1 & 3978$\pm$55 & $2^{21}$ & 579.5 & 586$\pm$8 & $2^{24}$\\
0.02 & 4339.5 & 4418$\pm$54 & $2^{19}$ & 582.6 & 585$\pm$6 & $2^{22}$\\
0.02828 & 5667.5 & 5935$\pm$71 & $2^{17}$ & 588.9 & 590$\pm$9 & $2^{20}$\\
0.04 & 16576.7 & 14382$\pm$232 & $2^{13}$ & 602.9 & 611$\pm$9 & $2^{18}$\\
0.05657 & -- & -- & -- & 641.5 & 647$\pm$10 & $2^{16}$\\
0.08 & -- & -- & -- & 827.5 & 934$\pm$12 & $2^{14}$\\
\bottomrule
\end{tabular}
\end{table}

\begin{table}[!htbp]\centering\footnotesize
\caption{Two-variable kernel of the hole transparent energy in the full face model (\texttt{collapse14.py}; independent random streams for every point; $2^{14}$ sequences per point at $N\le4096$, 4096 at $N=16384$; aligned leads at $\varphi=73^\circ$): the mean transmission at fixed $u=C_\lambda Np^4$ and $v=N|\delta|^{1/2}$ ($\delta>0$) for four thicknesses, and the finite-variance kernel $\Phi_1(2u)$ at $v=0$. $I=\sum_\pm\int\dd u\int2v\,\dd v\,\langle T\rangle$ by the trapezoidal rule on the base grid (207 points: 9 values of $u\le8$, 12 of $v\le16$, both signs, $v=0$ shared), with its sampling error; Richardson: the grid refined by all midpoints (4096 sequences per added point at $N=4096$, 2048 at $N=16384$) gives $I_{\rm fine}$, and $I_{\rm fine}+(I_{\rm fine}-I_{\rm base})/3$ removes the leading quadrature error; $A_{\rm th}=\pi Cw^*I/2C_\lambda\mathcal Z$ of $g\simeq A_{\rm th}N^{-3}$ (300\,K) from the best $I$ at each $N$ (Richardson where refined).}\label{tab:collapse}
\setlength{\tabcolsep}{4pt}
\begin{tabular}{llccccc}
\toprule
$\lambda$ & & $N=256$ & $N=1024$ & $N=4096$ & $N=16384$ & $\Phi_1(2u)$\\
\midrule
1.53 & $(0.1, 0)$ & 0.8015$\pm$0.0012 & 0.8168$\pm$0.0011 & 0.8273$\pm$0.0011 & 0.8273$\pm$0.0022 & 0.8328\\
1.53 & $(1, 0)$ & 0.0431$\pm$0.0010 & 0.1906$\pm$0.0018 & 0.2335$\pm$0.0020 & 0.2464$\pm$0.0039 & 0.2703\\
1.53 & $(0.1, 1)$ & 0.3702$\pm$0.0028 & 0.3919$\pm$0.0028 & 0.4049$\pm$0.0029 & 0.3984$\pm$0.0057 & \\
1.53 & $(1, 1)$ & 0.0200$\pm$0.0007 & 0.0955$\pm$0.0015 & 0.1217$\pm$0.0016 & 0.1287$\pm$0.0033 & \\
1.53 & $(0.1, 4)$ & 0.0448$\pm$0.0011 & 0.0508$\pm$0.0012 & 0.0542$\pm$0.0012 & 0.0553$\pm$0.0025 & \\
1.53 & $(1, 4)$ & 0.0027$\pm$0.0002 & 0.0139$\pm$0.0006 & 0.0184$\pm$0.0007 & 0.0225$\pm$0.0014 & \\
1.53 & $(4, 2)$ & 0.0000$\pm$0.0000 & 0.0002$\pm$0.0001 & 0.0028$\pm$0.0002 & 0.0058$\pm$0.0007 & \\
1.53 & $I$, base & 3.50$\pm$0.01 & 5.85$\pm$0.03 & 7.36$\pm$0.04 & 8.38$\pm$0.10 & \\
1.53 & $I$, Richardson & & & 7.13$\pm$0.04 &  & \\
1.53 & $A_{\rm th}$ & 53.6$\pm$0.2 & 89.4$\pm$0.4 & 109.1$\pm$0.6 & 128.1$\pm$1.5 & \\
\midrule
2.15 & $(0.1, 0)$ & 0.8289$\pm$0.0011 & 0.8285$\pm$0.0011 & 0.8316$\pm$0.0011 & 0.8317$\pm$0.0021 & 0.8328\\
2.15 & $(1, 0)$ & 0.2369$\pm$0.0020 & 0.2582$\pm$0.0020 & 0.2587$\pm$0.0020 & 0.2533$\pm$0.0039 & 0.2703\\
2.15 & $(0.1, 1)$ & 0.4657$\pm$0.0029 & 0.4715$\pm$0.0029 & 0.4878$\pm$0.0029 & 0.4766$\pm$0.0058 & \\
2.15 & $(1, 1)$ & 0.1370$\pm$0.0017 & 0.1541$\pm$0.0018 & 0.1628$\pm$0.0018 & 0.1649$\pm$0.0037 & \\
2.15 & $(0.1, 4)$ & 0.0940$\pm$0.0016 & 0.1010$\pm$0.0017 & 0.1041$\pm$0.0017 & 0.1061$\pm$0.0035 & \\
2.15 & $(1, 4)$ & 0.0282$\pm$0.0008 & 0.0356$\pm$0.0009 & 0.0381$\pm$0.0010 & 0.0390$\pm$0.0019 & \\
2.15 & $(4, 2)$ & 0.0022$\pm$0.0002 & 0.0063$\pm$0.0004 & 0.0085$\pm$0.0004 & 0.0084$\pm$0.0008 & \\
2.15 & $I$, base & 12.45$\pm$0.06 & 14.51$\pm$0.07 & 15.28$\pm$0.07 & 15.67$\pm$0.15 & \\
2.15 & $I$, Richardson & & & 15.06$\pm$0.08 & 15.11$\pm$0.12 & \\
2.15 & $A_{\rm th}$ & 689$\pm$3 & 804$\pm$4 & 834$\pm$5 & 837$\pm$6 & \\
\bottomrule
\end{tabular}
\end{table}

%% file: sec_appendix.tex
\section{Proof of the thickness theorem}\label{app:theorem}
Let $\|\mu\|=\int\dd\mu$ and $U_\delta=\{x:{\rm dist}(x,\Zg)<\delta\}$, with $\delta$ small enough that the factor $1+O(\rho)$ of H2 lies in $[\frac12,2]$ in every tube, so that $2N\gamma\in[q/2,2q]$ below. By H3, $\int_{U_\delta^c}t_N\,\dd\mu\le C\|\mu\|e^{-\mu_\delta N}=o(N^{-\alpha})$. In the tube around $Z_i$ use the coordinates of H1; by H2,
\begin{equation}
\int_{U_\delta\cap{\rm tube}_i}t_N\,\dd\mu=\int\dd\zeta\oint\dd\hat n\int_0^\delta\omega_i(\zeta,\hat n)\,\rho^{s_i+c_i-1}\,t_N\,[1+O(\rho)]\,\dd\rho .
\end{equation}
Substitute $q=2N\Gamma_i(\zeta,\hat n)\rho^{2m_i}$. Then $\rho^{s_i+c_i}=(q/2N\Gamma_i)^{\alpha_i}$ and $\dd\rho/\rho=\dd q/(2m_iq)$, so
\begin{equation}
N^{\alpha_i}\!\int_{U_\delta\cap{\rm tube}_i}\!t_N\,\dd\mu=\int\dd\zeta\oint\dd\hat n\,\frac{\omega_i(\zeta,\hat n)}{2m_i}\big[2\Gamma_i(\zeta,\hat n)\big]^{-\alpha_i}\int_0^{Q_N}q^{\alpha_i-1}\,t_N\big(x(\zeta,\rho(q)\hat n)\big)[1+O(\rho(q))]\,\dd q,
\end{equation}
with $Q_N=2N\Gamma_i\delta^{2m_i}\to\infty$. At fixed $q$, $\rho(q)=(q/2N\Gamma_i)^{1/2m_i}\to0$. By H2 the argument of the kernel is $2N\gamma=q[1+O(\rho)]$, and by H4 $t_N=\Phi_{\zeta\hat n}(q[1+O(\rho)])[1+O(\rho)]$, with the kernel $\Phi_{\zeta\hat n}=\Phi_i(\zeta,\hat n;\cdot)$ of the point $\zeta$ and the normal direction $\hat n$. The integrand is dominated, independently of $N$, by $C\,q^{\alpha_i-1}\bar\Phi(q/2)$, which is integrable because $\bar\Phi$ is non-increasing and $t_N\le\bar\Phi(2N\gamma)\le\bar\Phi(q/2)$; the positivity bound $\Gamma_i\ge\Gamma_{\min}$ makes the constant uniform in $(\zeta,\hat n)$. Dominated convergence gives the limit
\begin{equation*}
\int\dd\zeta\oint\dd\hat n\,\frac{\omega_i(\zeta,\hat n)}{2m_i}\,\big[2\Gamma_i(\zeta,\hat n)\big]^{-\alpha_i}J_{\alpha_i}(\zeta,\hat n),\qquad J_{\alpha_i}(\zeta,\hat n)=\int_0^\infty q^{\alpha_i-1}\Phi_{\zeta\hat n}(q)\,\dd q,
\end{equation*}
which divided by $\|\mu\|$ is $A_i$ of Eq.~\eqref{eq:general}; the moment stays inside the integrals over $\zeta$ and $\hat n$ (it factors out only if the kernel is the same everywhere on $Z_i$). For the rate, write, at fixed $(\zeta,\hat n)$, $t_N-\Phi_{\zeta\hat n}(q)=[t_N-\Phi_{\zeta\hat n}(2N\gamma)]+[\Phi_{\zeta\hat n}(2N\gamma)-\Phi_{\zeta\hat n}(q)]$. The first bracket is bounded by $C\rho R(2N\gamma)\le C\rho R(q/2)$ throughout the tube (H4$'$, $R$ non-increasing); the second by $C\rho\,q|\Phi_{\zeta\hat n}'(\xi)|$ with $\xi$ between $q$ and $2N\gamma=q[1+O(\rho)]$, which H4$'$ bounds uniformly in $(\zeta,\hat n)$. With $\rho(q)\propto(q/N)^{1/2m_i}$ both integrate, by the moment condition of H4$'$, to relative corrections $O(N^{-1/2m_i})$. A remainder controlled only at fixed $q$ is not enough, because the region that matters moves with $N$ (counterexample in Sec.~\ref{sec:theorem}). Components with $\alpha_i>\alpha$ contribute $O(N^{-\alpha_i})$. $\square$

\emph{Uniform positivity is necessary.} Without $\Gamma_i\ge\Gamma_{\min}$ the normal direction in which $\Gamma_i\to0$ is a further zero set, and without H3 regions far from $\Zg$ can contribute algebraically. The transparent-energy construction violates the bare-disorder version of H3 but satisfies H1--H4 for the effective exponent, which is why the theorem covers it.

\emph{Scalar chains.} For on-site disorder of variance $\sigma^2$ and hopping $T$ at wavenumber $k$, $\gamma=\sigma^2/(8T^2\sin^2k)$ (App.~\ref{app:born}); with $\sigma\propto\rho^\kappa$, $T\propto\rho^r$, $2m=2(\kappa-r)$.

\section{The continuum coupling}\label{app:coupling}
\emph{Bilinear form.} For layers 1 and 2 with Bloch states at crystal momenta that correspond to the same lab momentum $\bp$, $\psi_1=\sum_{\mathbf G\in L_1^*}a_{\bp+\mathbf G}(z)e^{i(\bp+\mathbf G)\cdot\mathbf r}$ and $\psi_2=\sum_{\mathbf G'\in L_2^*}b_{\bp+R_\theta\mathbf G'}(z)e^{i(\bp+R_\theta\mathbf G')\cdot\mathbf r}$. The overlap and the kinetic matrix element conserve the in-plane wavevector, so between states at the same lab momentum of incommensurate layers ($R_\theta\mathbf G'\in L_1^*\Rightarrow\mathbf G'=0$)
\[
S_{12}=\int\dd z\;a_{\bp}(z)^\dagger b_{\bp}(z),\qquad K_{12}=\int\dd z\;a_{\bp}(z)^\dagger\Big[-\frac{\hbar^2}{2m}\big(\partial_z^2-p^2\big)\Big]b_{\bp}(z),
\]
both bilinear in the components at wavevector exactly $\bp$ (spinors in $z$). Eq.~\eqref{eq:H12} then expresses $H_{12}$ through $S_{12}$, $K_{12}$ and the layer energies alone, with no potential matrix element; Eq.~\eqref{eq:Torth} is its first-order L\"owdin form. Expanding $a_{\bp}(z)$ and $b_{\bp}(z)$ near the interface in a finite set of profiles $\chi_\alpha(z)$, $a_{\bp}(z)=\sum_\alpha\chi_\alpha(z)G^{\rm t}_\alpha(\bp)+\dots$, gives Eq.~\eqref{eq:formfactor} with $\tau_{\alpha\beta}$ the matrix elements of the kernel between profiles (energies of the two layers included). The expansion is exact if the profiles span the interface region; truncating it is the separability approximation. Rotating layer 2 by $\theta$ maps its component at $\bp$ to $U_s(\theta)b^{(0)}_{R_\theta^{-1}\bp}(z)$, which gives the lab-frame form. For a single profile concentrated near a plane $z_f$ in the gap, $G^{\rm t}(\bp)$ is proportional to the planar Fourier component $a_{\bp}(z_f)$ of the Bloch state, a property of the wavefunction that a plane-wave calculation gives directly (Sec.~\ref{sec:validation}).

\emph{Scope: forward channel.} The same matrix elements between $\psi_1$ at $\bp$ and $\psi_2$ at $\bp'=\bp+\mathbf G-R_\theta\mathbf G'$ ($\mathbf G,\mathbf G'\neq0$) do not vanish: they are the umklapp couplings, carried by the reciprocal-lattice components of the two states in the gap. The forward-channel model keeps only $\bp'=\bp$. Umklapp couplings $H_{PQ}$ enter the chain of channel $\bp$ at second order, as $H_{PQ}(E-H_{QQ})^{-1}H_{QP}$, only if the partner states are off shell: $|E-\varepsilon_Q|\ge\Delta>0$ with $\|H_{PQ}\|/\Delta\ll1$. For carriers at a band edge that is a global extremum of the monolayer bands (the $\Gamma$ edges of phosphorene) every partner at $|\bp'|\gtrsim|\mathbf G-R_\theta\mathbf G'|$ lies away from $E$ when the interface is far from alignment, and $\|H_{PQ}\|/\Delta$ must be checked; for other carriers, or near alignment, partners can be on shell and scatter between channels (Sec.~\ref{sec:umklapp}).

\emph{Consequences.} (a) At $\bp=0$ only $m_l=0$ components have a non-vanishing in-plane average; their $J_z$ equals $s_z$. (b) A C$_n$ class without $J_z\equiv s_z$ has no in-plane average, and the amplitude at small $\bp$ is generically $O(p^{\mu_0})$ from the lowest $|J_z-s_z|$ of the class; two such factors give $r=2\mu_0$. (c) Twist phases of pure-$J_z$ content are a gauge of the chain.

\emph{Relation to two-centre models.} A two-centre hopping $t_{ab}(\mathbf R)$ between orbitals of angular momenta $m_a=m_b=m\neq0$ can have a non-zero registry average if its Slater--Koster parameters are fitted freely. If it is derived from an atom-centred potential, $\int\dd^2R\,t_{ab}(\mathbf R)=\int\dd^3r\,\phi_a^*V_a\int\dd^2R\,\phi_b(\mathbf r-\mathbf R)=0$ for $m_b\neq0$. Fitted two-centre models can therefore introduce a spurious twist dependence of $T(0)$; the T2 test detects it.

\section{The matrix index}\label{app:matrix}
\emph{Lab-gauge harmonics.} Let layer $l$ be rotated by $\theta_l$. With $G(R\mathbf q)=U_s(t)G(\mathbf q)D(t)^\dagger$ for $t\in\frac{2\pi}{n}\mathbb Z$ and $G(\mathbf q)=\sum_mg_m(p)e^{im\varphi_{\mathbf q}}$ entrywise, the lab-gauge entry $(s_z,J_z)$ of $U_s(\theta)G(R_\theta^{-1}\bp)D(\theta)^\dagger$ is $\sum_mg_m(p)e^{im\varphi_{\bp}}e^{i(J_z-s_z-m)\theta}$, with $m\equiv J_z-s_z\pmod n$ and $g_m=O(p^{|m|})$. The twist-independent part is $m=J_z-s_z$; the twist-dependent part starts at $\mu_1=\min\{|m|:m\equiv J_z-s_z,\ m\neq J_z-s_z\}$. For a $J_z=\pm\frac12$ Kramers pair ($n\ge2$): spin-conserving entries have $\mu_1=n$, spin-flip entries $\mu_1=n-1$. Onsite entries follow with $s_z$ replaced by the column's $J_z$.

\emph{Conventions.} Face channels are $(\alpha,s)$: real interface profiles $\chi_\alpha(z)$, which time reversal maps to themselves, times the spin $s=\pm\frac12$ along the lab $z$ axis ($s=0$ for spinless states). Time reversal acts as $G(\mathbf q)=W_f^{\!\top}G(-\mathbf q)^*W_m$ with $W_f=\oplus_\alpha i\sigma_y$ and $W_m=i\sigma_y$ for a Kramers pair ($W=1$ spinless). The kernel $\tau_{\alpha\beta}(p)$ between the top-face profiles of one layer and the bottom-face profiles of the next is real: time reversal gives $\tau(\bp)^*=\tau(-\bp)$, and the forward kernel of a twisted interface is isotropic in the plane. It need not be symmetric (the Bardeen form of Sec.~\ref{sec:validation} is antisymmetric in value and derivative profiles). We take the lab momentum along $x$; $a=\pm$ labels $s,J_z=\pm\frac12$.

\begin{lemma}[face amplitudes near the axis]\label{lem:faces}
Let $n\ge2$ and let the multiplet be a $J_z=\pm\frac12$ Kramers pair with $p_z$/$s$-type face channels. (i)~At $p=0$ every profile block is $G_\alpha(0)={\rm diag}(c_\alpha,c_\alpha^*)$ in the $(s,J_z)$ basis. (ii)~In the lab gauge the twist-dependent part of $G_\alpha$ is, to leading order, the spin-flip harmonic of in-plane order $n-1$,
\begin{equation}
\delta G_\alpha(\theta)=p^{\,n-1}\begin{pmatrix}0&(-1)^na_\alpha^*e^{-in\theta}\\ a_\alpha e^{in\theta}&0\end{pmatrix}+O(p^n).\label{eq:dGlead}
\end{equation}
(iii)~Twist-dependent on-site entries start at order $p^{n-1}$ for even $n$ and at order $p^n$ for odd $n$.
\end{lemma}
\noindent\emph{Proof.} The entry $(s,J_z)$ of $G(\mathbf q)$ contains harmonics $e^{im\varphi_{\mathbf q}}$ with $m\equiv J_z-s\pmod n$, of order $p^{|m|}$, and acquires the lab-gauge phase $e^{i(J_z-s-m)\theta}$. (i)~At $\mathbf q=0$ only $m=0$ survives, which for $n\ge2$ requires $J_z=s$; time reversal maps the entry $(+,+)$ to the conjugate of $(-,-)$. (ii)~Twist dependence requires $m\neq J_z-s$. The smallest such $|m|$ is $n-1$ in the spin-flip entries ($m=1-n$ in $(-,+)$, $m=n-1$ in $(+,-)$) and $n$ in the spin-conserving ones. Time reversal gives $G_{+-}(\mathbf q)=-G_{-+}(-\mathbf q)^*$; with $G_{-+}=\sum_mg_mp^{|m|}e^{im\varphi}$ and $G_{+-}=\sum_mb_mp^{|m|}e^{im\varphi}$ this is $b_m=-(-1)^mg_{-m}^*$, and for $m=n-1$, with $a_\alpha\equiv g_{\alpha,1-n}$, it is Eq.~\eqref{eq:dGlead}. (iii)~Time reversal and Hermiticity make the on-site spin-flip entry odd in $\mathbf q$, $X_{+-}(\mathbf q)=-X_{+-}(-\mathbf q)$, so only odd $|m|$ occur in it; its smallest twist-dependent order $n-1$ is odd only for even $n$. Spin-conserving on-site harmonics start at $|m|=n$. $\square$

\begin{lemma}[bond parity]\label{lem:bond}
Under the conditions of Lemma~\ref{lem:faces}, for any number of interface profiles and any real kernel $\tau$, the twist-averaged bond is $\bar T={\rm diag}(t_0,t_0^*)+O(p)$ with $t_0=\sum_{\alpha\beta}c^*_{{\rm t}\alpha}\tau_{\alpha\beta}c_{{\rm b}\beta}$, and its fluctuation is
\begin{equation}
\delta T_l=-\Omega_{{\rm t},l}\,\bar T+\bar T\,\Omega_{{\rm b},l+1}+O(p^n),\label{eq:bondrot}
\end{equation}
where $\Omega_{{\rm t},l}$ depends on $\theta_l$ only, $\Omega_{{\rm b},l+1}$ on $\theta_{l+1}$ only, and both are $O(p^{n-1})$, purely spin flip, \emph{anti-Hermitian for odd $n$ and Hermitian for even $n$.}
\end{lemma}
\noindent\emph{Proof.} In $\delta T_l=\delta G_{{\rm t},l}^\dagger\tau\bar G_{\rm b}+\bar G_{\rm t}^\dagger\tau\,\delta G_{{\rm b},l+1}+\delta G_{{\rm t},l}^\dagger\tau\,\delta G_{{\rm b},l+1}$ the last term is $O(p^{2n-2})\subset O(p^n)$. Summing Eq.~\eqref{eq:dGlead} over profiles with a real $\tau$,
\[
\delta G_{{\rm t},l}^\dagger\tau\bar G_{\rm b}=p^{n-1}\!\begin{pmatrix}0&A_{\rm t}^*e^{-in\theta_l}\\(-1)^nA_{\rm t}e^{in\theta_l}&0\end{pmatrix}\!,\quad
\bar G_{\rm t}^\dagger\tau\,\delta G_{{\rm b},l+1}=p^{n-1}\!\begin{pmatrix}0&(-1)^nA_{\rm b}^*e^{-in\theta_{l+1}}\\A_{\rm b}e^{in\theta_{l+1}}&0\end{pmatrix}\!,
\]
with $A_{\rm t}=\sum_{\alpha\beta}a_{{\rm t}\alpha}\tau_{\alpha\beta}c_{{\rm b}\beta}$ and $A_{\rm b}=\sum_{\alpha\beta}c_{{\rm t}\alpha}\tau_{\alpha\beta}a_{{\rm b}\beta}$. Writing the first matrix as $-\Omega_{{\rm t},l}\bar T$ and the second as $\bar T\Omega_{{\rm b},l+1}$ gives the off-diagonal entries (times the phases $e^{\mp in\theta}$)
\[
(x,y)_{\Omega_{\rm t}}=\big(-A_{\rm t}^*/t_0^*,\,-(-1)^nA_{\rm t}/t_0\big)\,p^{n-1},\qquad(x,y)_{\Omega_{\rm b}}=\big((-1)^nA_{\rm b}^*/t_0,\,A_{\rm b}/t_0^*\big)\,p^{n-1}.
\]
 They obey $y=-x^*$ (anti-Hermitian) for odd $n$ and $y=x^*$ (Hermitian) for even $n$. $\square$

\noindent The lemma works on the full bond contraction. It needs no square or invertible face amplitude: a factorization $\delta G=\bar GA$ through a left inverse of the stacked average $\bar G$ would drop the part of $\delta G$ outside the range of $\bar G$, which does couple through a general kernel.

\begin{lemma}[frame rotations do not backscatter at Born order]\label{lem:frame}
Let the averaged chain have $\bar T=c\,U$ ($c>0$, $U$ unitary) and $\bar X=\varepsilon\,\mathbb 1$, and let $\delta T_l=-\Omega_{{\rm t},l}\bar T+\bar T\Omega_{{\rm b},l+1}$ with anti-Hermitian $\Omega$ and $\delta X_l=0$. Then the first-order Born amplitude $m_l$ of Eq.~\eqref{eq:matrixborn} vanishes for every forward channel $a$ and backward channel $b$ at the same energy.
\end{lemma}
\noindent\emph{Proof.} Channels lie in eigenspaces of $U$, $Uu=e^{i\phi}u$, with $E=\varepsilon+2c\cos(k+\phi)$, so a forward channel $(k_a,\phi_a)$ and a backward channel $(k_b,\phi_b)$ at the same energy obey $k_b+\phi_b=-(k_a+\phi_a)$, whichever eigenspaces they belong to. With $\delta T_l^\dagger=\bar T^\dagger\Omega_{{\rm t},l}-\Omega_{{\rm b},l+1}\bar T^\dagger$,
\[
m_l=c\,(\Omega_{{\rm t},l})_{ba}\big[e^{-i(k_b+\phi_b)}-e^{i(k_a+\phi_a)}\big]+c\,e^{i(\phi_b-\phi_a)}(\Omega_{{\rm b},l+1})_{ba}\big[e^{i(k_a+\phi_a)}-e^{-i(k_b+\phi_b)}\big]=0.\qquad\square
\]
For Hermitian $\Omega$ the same algebra gives $m_l=-2c\,e^{i(k_a+\phi_a)}\big[(\Omega_{{\rm t},l})_{ba}-e^{i(\phi_b-\phi_a)}(\Omega_{{\rm b},l+1})_{ba}\big]$, generically non-zero. Non-scalar corrections of order $p^a$ to $\bar T$ and $\bar X$ spoil the cancellation at relative order $p^a$: a frame rotation of order $p^\mu$ backscatters at order $p^{\mu+a}$.

\emph{Proof of Proposition~\ref{prop:matrix}.} Kramers $J_z=\pm\frac12$, $n\ge2$: at order $p^{n-1}$ the twist dependence is Eq.~\eqref{eq:bondrot}, plus for even $n$ a Hermitian on-site term (Lemma~\ref{lem:faces}(iii)). For even $n$ the rotation generators are Hermitian and the Born amplitude is $O(p^{n-1})$ and generically non-zero: $\kappa=n-1$. For odd $n$, $\bar T(0)=|t_0|\,U$ with $U$ diagonal unitary and $\bar X(0)=\varepsilon\mathbb 1$ (Kramers degeneracy), so Lemma~\ref{lem:frame} removes the whole $O(p^{n-1})$ amplitude. At order $p^n$ remain (a)~the spin-conserving harmonics $|m|=n$ of faces and on-site block, which time reversal does not restrict to an anti-Hermitian form (it relates the two diagonal entries to each other, not each to itself); (b)~the $O(p^{n-1})$ rotation acting on the non-scalar part of the averaged chain, $O(p^{n-1+a})$ with $a=1$ when the lab-constant spin-flip harmonic $|m|=1$ is present; (c)~products of two fluctuations, $O(p^{2n-2})$. The amplitude is therefore $O(p^n)$ and generically non-zero: $\kappa=n$. No relation between the two faces, no invertibility of $\bar G$ and no restriction on the number of profiles enters. Spinless A-type ($d=1$): the bond is a scalar and, by time reversal with real profiles and a real kernel, $\bar T(0)$ is real. The twist harmonic of order $n$ is $p^n(g_ne^{-in\theta}+(-1)^ng_n^*e^{in\theta})$ per profile: real for even $n$, a modulus change ($\kappa=n$), and imaginary for odd $n$, a phase, which is an exact gauge of a scalar chain at every order; the modulus is then first twist dependent at $2n$ (the $|m|=2n$ harmonic, or the square of the phase term), $\kappa=2n$. C$_1$: $|m|=1$ is a phase, $\kappa=2$. The on-site block (real, even harmonics) gives the same orders. C$_1$ Kramers pairs: every entry is allowed at $p=0$, and time reversal makes each profile's $p=0$ block a scalar times an SU(2) frame. With one profile (or a common frame) the $p=0$ bond is a scalar times a unitary, a removable rotation, and the first coupling change is the $|m|=1$ spin-flip harmonic (Hermitian), $\kappa=1$; with profiles in different frames the $p=0$ bond is in general not a scalar times a unitary and the axis is not protected. $\square$

\emph{Interpretation: accumulated rotations.} A rotation common to both faces of a layer, $\Omega_{{\rm t},l}=\Omega_{{\rm b},l}$, is removed exactly by the local gauge $\psi_l\to e^{\Omega_l}\psi_l$ (with $X_l\to e^{-\Omega_l}X_le^{\Omega_l}$), at every order and for every $N$. If the faces rotate differently, the unitary bond factors are removed by a recursive gauge, which accumulates: the frame $V_l$ performs a random walk with steps $O(p^\mu)$. If the step generators do not all commute with the non-scalar $O(p^a)$ part of the averaged chain (a non-degeneracy condition), the rotated corrections decorrelate over $j$ layers as $e^{-D|j|}$, $D\propto p^{2\mu}$, and their backscattering density at $2k$ is $p^{2a}D/(2\sin^2k)$ for $\sin k\gg D$, of order $p^{2(\mu+a)}$: the order of Lemma~\ref{lem:frame}. If the generators commute with those corrections, nothing decorrelates and the order can only be higher. This picture is not used in the proof. We make no all-order statement for accumulated rotations and none within $D$ of a band edge; Sec.~\ref{sec:matrixval} compares the Born orders with exact exponents.

\emph{The face chain.} For square, invertible face amplitudes and $\tau=t$, the substitution $\phi_l=G^{\rm b}_l\psi_l$ turns the chain into
\begin{equation}
Y_l\,\phi_l=t\,K_l^\dagger\phi_{l+1}+t\,K_{l-1}\phi_{l-1},\qquad K_l=G^{\rm t}_l(G^{\rm b}_l)^{-1},\quad Y_l=(G^{\rm b}_l)^{-\dagger}(E-X_l)(G^{\rm b}_l)^{-1},\label{eq:facechain}
\end{equation}
in which the multiplet basis has disappeared: $K_l$ (the transfer from the bottom to the top face of layer $l$) and $Y_l$ are lab-frame matrices on the face channels. For identical faces and a scalar on-site block $K_l=1$ and $Y_l=(E-\varepsilon)(G_lG_l^\dagger)^{-1}$, the inverse Gram matrix. Single-bond singular values see neither the rotation of the eigenbasis of $Y_l$ nor the unitary part of $K_l$ relative to it; both scatter.

\section{Stress-free interfaces}\label{app:stress}
\emph{Identities.} Let $U(\bp;d)=\rho_1\rho_2\int\dd^2r\,e^{i\bp\cdot\mathbf r}\phi(\sqrt{r^2+d^2})$, with a cutoff at which $\phi$ and $\phi'$ vanish. Then $\int\dd^2r\,e^{i\bp\cdot\mathbf r}\partial_i\partial_j\phi=-p_ip_jU$, $\int\dd^2r\,e^{i\bp\cdot\mathbf r}\partial_i\partial_z\phi=-ip_i\partial_dU$ and $\partial_z^2\to\partial_d^2$, which is Eq.~\eqref{eq:Ftensor}.

\emph{Small-$p$ expansion.} $U(\bp;d)=\sum_{k\ge0}(-1)^k(p/2)^{2k}M_{2k}(d)/(k!)^2$ with moments $M_{2k}=2\pi\rho_1\rho_2\int r^{2k+1}\phi\,\dd r$. Writing $M_2$ with $s=r^2+d^2$ gives $\partial_dM_2=-2dM_0$, hence $\partial_d^2M_2=-2M_0-2d\,\partial_dM_0$. At a force-free spacing ($\partial_dM_0=\partial_dU(0;d)=0$) this gives, to leading order, $F_{zz}(0)-F_{zz}(\bp)=\frac{p^2}{4}\partial_d^2M_2+O(p^4)=-\frac{p^2}{2}U(0;d)+O(p^4)$, positive for a bound interface ($U(0)<0$) at small $p$. Under a normal load the term $-2d\,\partial_dM_0$ adds, and for stretched interfaces it can make the in-phase bending mode soft; together with the $O(p)$ coupling $F_{iz}=-ip_i\partial_dU(0;d)$ this is the mechanism of the instabilities of pre-stressed stacks. The quadrature reproduces these identities to $10^{-3}$ relative at $p=10^{-3}$\,\AA$^{-1}$ and $F_{xz}\propto p^{3.000}$ at $d_{\rm eq}$ (against $p^{1.001}$ at 3.40\,\AA) (\texttt{interface7.py --selftest}).

\emph{The family.} Morse $=$ repulsive $+$ attractive part, each shifted-force at $r_c$, so $U_l=a_lU_{\rm rep}+b_lU_{\rm att}$ is linear in $(a_l,b_l)$. The two conditions of Eq.~\eqref{eq:family} are a $2\times2$ linear system. The rigid-shift family $a=e^{2\alpha\delta}$, $b=e^{\alpha\delta}$ (Morse radius $r_0+\delta$) is the special case with $\beta_{\rm eq}=-0.2455$\,\AA$^{-1}$.

\emph{Stability gate.} The Hessian of a checkpoint system (sample of $N$ layers with both semi-infinite leads) is positive definite iff the Schur complement obtained by eliminating the leads is. The leads contribute their static self-energies $\Sigma_{L,R}(\omega^2=0)=-K^\dagger D_{\rm lead}^{-1}K$ (negative semidefinite), computed by decimation at $\omega=0$. The block-$LDL^\dagger$ pivots $S_l=A_l-K_{l-1}^\dagger S_{l-1}^{-1}K_{l-1}$ have, by Sylvester's law of inertia, as many negative eigenvalues in total as the Hessian. The gate counts them for every checkpoint and every channel, checks the lead dispersion $\omega^2(\bp,k_z)\ge0$, and computes the lowest eigenvalue of the dense $768\times768$ Hessian of the smallest channel at $N=256$ as an independent check. Production stops (exit code 3) on any negative eigenvalue.

\section{Born exponents}\label{app:born}
At weak disorder $\gamma(k)=L\langle|\langle-k|\delta H|k\rangle|^2\rangle/(2v^2)$, $v=\dd E/\dd k$. Diagonal disorder of spectrum $S_V$: $\gamma=S_V(2k)/(8t^2\sin^2k)$. Bond disorder of electrons: $\gamma_e=S_\delta(2k)/(2\sin^2k)$. Out-of-plane phonons at finite $\bp$ with on-site $F_0(d_{l-1})+F_0(d_l)$ and bond $F_p(d_l)$ that vary along the equilibrium family with slopes $\beta_0=-\partial\ln F_0$ and $\beta_p=-\partial\ln F_p$:
\begin{equation}
\gamma_{\rm ph}(k;\bp)=\frac{S_d(2k)\big[\beta_0F_0\cos k-\beta_pF_p\big]^2}{2F_p^2\sin^2k},\label{eq:gam_php}
\end{equation}
which reduces at $\bp=0$ to $S_d(2k)\beta_0^2\tan^2(k/2)/2$. For the stress-free family $\beta_0=\beta_{\rm eq}$, and $\beta_p(\bp)=\beta_{\rm eq}+O(p^2)$; for family R ($\beta_{\rm eq}=0$) $\beta_0=0$ and the phonon exponent is $O(p^4)$; the rigid shift has $\beta_0=\beta_{\rm eq}=-0.25$\,\AA$^{-1}$. The bracket vanishes on the surface of Eq.~\eqref{eq:cone} (Sec.~\ref{sec:joint_asym}). Lognormal spectra: for $\delta_l=e^{-sx_l-s^2/2}-1$ with Gaussian $x$, $S_\delta(q)=\sum_m(e^{s^2C(m)}-1)e^{iqm}$.

\section{Thermoelectric coefficients}\label{app:thermo}
With $f$ the Fermi--Dirac function and $x=(E-\mu)/k_BT$, $L_n(\mu)=\int\dd E\,\mathcal T(E)(-\partial_Ef)x^n$, $\sigma\propto L_0$, $|S|=(k_B/e)L_1/L_0$ and $\kappa_e=(k_B/e)^2T\sigma_{\rm u}(L_2-L_1^2/L_0)$ in units $\sigma_{\rm u}$ of conductance per unit $L_0$. Normalizing every $L_n$ by the aligned $L_0$ at $\mu$ at the band edge and defining the quality factor $B=\sigma_{\al}(\mu_{\rm edge})T(k_B/e)^2/\kappa_{\rm ph,\al}$,
\begin{equation}
ZT(\mu)=\frac{B\,L_1^2/L_0}{g_{\rm ph}+B\,(L_2-L_1^2/L_0)},\label{eq:ztB}
\end{equation}
with $g_{\rm ph}=\kappa_{\rm ph}/\kappa_{\rm ph,\al}$. The optimum over $\mu$ is taken for the aligned and the disordered stack separately. This includes the electronic thermal conductance exactly; optimizing the power factor instead would neglect it. In the Boltzmann limit $L_1/L_0\to\langle\varepsilon\rangle/k_BT+(E_{\rm edge}-\mu)/k_BT$.

\section{Numerical methods and reproducibility}\label{app:num}
{\sloppy
The code is available at \url{https://github.com/kpleo/twist_theory}. Its folders hold the scripts named below: \texttt{models/} (model calculations), \texttt{blackphosphorus/} (transport of Sec.~\ref{sec:dft} and the transparent-energy kernels), \texttt{firstprinciples/} (workflow and analysis of the first-principles runs), \texttt{kit/} (Sec.~\ref{sec:validation}) and \texttt{figures/}. The first-principles results that it reads are part of this Supplemental Material.
\begin{itemize}
\item \textbf{Interface tensor} (\texttt{models/\allowbreak interface7.py}): Gauss--Legendre quadrature with 480 nodes on $[0,\sqrt{r_c^2-d^2}]$, converged to $10^{-13}$; transforms of the repulsive and attractive Morse parts and their first two $d$ derivatives evaluated analytically under the integral, at every channel momentum and interface spacing.
\item \textbf{Chains} (\texttt{models/\allowbreak chain7.py}): recursive Green functions; lead surface Green functions by decimation with an imaginary part $10^{-13}$ of $\omega^2$ relative to the band scale (the in-plane bands of the aligned lead are extremely flat, width $\propto p^2|U|$; with $10^{-9}$ single channels lose up to 3\% over 256 layers and the aligned total 0.02\%, \texttt{eta\_check7.py}); consistency tests: aligned chain transparent to $10^{-8}$, recursion equal to the dense Caroli formula to $10^{-14}$ (\texttt{test\_chain7.py}).
\item \textbf{Paired runs} (\texttt{models/\allowbreak pair\_v7.py}): 42 configurations (aligned or random twists; five spectra at 0.03, 0.045, 0.06, 0.09\,\AA), 16 sequences each, $N=4$--256; electrons 1127 channels $\times$ 96 energies; phonons 120 channels ($|\bp|\le0.90$\,\AA$^{-1}$) $\times$ 48 frequencies, three stiffness families; per-sample channel spectra stored for all calibration configurations.
\item \textbf{Branch exponents and scans} (\texttt{models/\allowbreak phys7.py}): C$_2$, C$_4$, C$_6$ and flexural C$_2$ layers; $z$, in-plane and full $3\times3$ models; $N\le2048$, 4 sequences; channel domain inside the zone; stability gate on every sequence.
\item \textbf{Analysis}: \texttt{analyze\_v7.py} (electrons), \texttt{analyze7\_ph.py} (hold-out tests, bootstrap), \texttt{thermo\_v7.py} (Fermi--Dirac), \texttt{thermo\_beta7.py} ($\beta_{\rm eq}$ thresholds), \texttt{boot7.py} (sampling floors), \texttt{incoherent7\_drude.py}, \texttt{domain\_check7.py} and \texttt{eta\_check7.py} (supporting checks), \texttt{registry7.py}, \texttt{transparent7.py}, \texttt{index7.py}, \texttt{formfactor7.py}, \texttt{incoherent7.py}.
\item \textbf{Driver}: \texttt{reproduce.sh} runs every step in order and stops on the first failure (exit codes propagated, the stability gate included), in three stages: \texttt{models} (the model calculations, about 17\,h on 2 cores), \texttt{bp} (the kit self-tests, the analysis of the first-principles results and the transport of Sec.~\ref{sec:dft}, about 5\,h; \texttt{FAST=1} reuses the long transport results) and \texttt{figures} (the tables and every figure of the letter and of this Supplemental Material). The first-principles runs are cluster jobs and are not part of it.
\item \textbf{Matrix index}: \texttt{models/\allowbreak matrix8.py} (face-resolved random models: C$_n$-, time-reversal- and face-symmetric coefficients up to sixth order, an interior block that keeps the multiplet normalized, L\"owdin normalization at every $\mathbf q$, lab-gauge chain; Born rate between open channels with degenerate crossings grouped and the velocity diagonalized in each group; transfer-matrix Lyapunov exponents with QR re-orthogonalization every 8 layers and the clean-chain estimator floor), \texttt{matrix8\_table.py} (Table~\ref{tab:kappa}), \texttt{index\_rule8.py} (Proposition~\ref{prop:matrix} and the parity lemma), \texttt{validate8.py}, \texttt{validate8\_extra.py}, \texttt{validate8\_bc.py}, \texttt{validate8\_e.py}, \texttt{validate8\_parse.py} (Sec.~\ref{sec:matrixval}), \texttt{c3\_face\_check.py} (explicit C$_3$ doublet), \texttt{multiprofile8.py} (two interface profiles, Sec.~\ref{sec:multiprofile}), \texttt{tb8.py} (Slater--Koster examples).
\item \textbf{Held-out families and thresholds}: \texttt{pair8.py} (families Q$_1$ and Q$_3$, same model, seeds and streams as \texttt{pair\_v7.py}; joint runs, then the spacing-only calibration runs with \texttt{--aligned-only}; a resumed run refuses to continue if the stored family, $\beta$, grid or code version differ, or if no code version is stored, and it keeps only complete configurations: requested sample count, one gate record per sample, channel means and, for calibration runs, per-sample spectra of the right shape; incomplete ones are recomputed), \texttt{interp8\_test.py} (five interpolation rules), \texttt{thermo8.py} (three-family thresholds, the selection test), \texttt{boot8.py} (paired and unpaired bootstrap), \texttt{loo8.py} (out-of-sample tests and five-family thresholds; rotation-only electrons and phonons share their configuration index), \texttt{thr\_table8.py}, \texttt{thermo8\_sens.py} (Bose weighting), \texttt{mu\_window8.py} (chemical-potential window), \texttt{inplane\_margin8.py}, \texttt{transparency8.py}.
\item \textbf{Grid sensitivity at the thresholds} (\texttt{models/\allowbreak grid10.py}, \texttt{grid10.log}): eight threshold conditions (alternating and curvature-correlated noise at 0.045\,\AA; $B=0.3$, 1; $N=128$, 256) rerun on a grid refined twice in every direction, 79 radii $\times$ 6 azimuths and 95 frequencies (instead of 40 $\times$ 3 and 48, nested in the production grid), same families, seeds and streams, stability gate on every refined sample; the production values are recomputed from the stored spectra by the same code path, so the differences are the grid effect alone. The largest changes are 0.004 in $\beta^*/\beta_e$ (break-even and gain over rotation alone) and 0.04\% in the rotation-only $g_{\rm ph}$:
\begin{center}\small
\begin{tabular}{lcc@{\quad}cc@{\quad}c}
\hline
condition & \multicolumn{2}{c}{break-even $\beta^*/\beta_e$} & \multicolumn{2}{c}{gain $\beta^*/\beta_e$} & $g_{\rm ph}$ change \\
 & production & refined & production & refined & \\
\hline
alt, $B=1$, $N=128$ & 0.000 & 0.000 & 0.513 & 0.514 & $-0.04\%$\\
alt, $B=1$, $N=256$ & 0.038 & 0.038 & 0.558 & 0.562 & $+0.02\%$\\
alt, $B=0.3$, $N=128$ & 0.000 & 0.000 & 0.529 & 0.529 & $-0.04\%$\\
alt, $B=0.3$, $N=256$ & 0.078 & 0.078 & 0.576 & 0.580 & $+0.02\%$\\
curv, $B=1$, $N=128$ & 0.453 & 0.450 & 0.742 & 0.740 & $-0.04\%$\\
curv, $B=1$, $N=256$ & 0.563 & 0.559 & 0.818 & 0.816 & $+0.02\%$\\
curv, $B=0.3$, $N=128$ & 0.542 & 0.540 & 0.828 & 0.827 & $-0.04\%$\\
curv, $B=0.3$, $N=256$ & 0.642 & 0.639 & 0.916 & 0.915 & $+0.02\%$\\
\hline
\end{tabular}
\end{center}
These shifts are an order of magnitude below the leave-one-out scale (rms 1.6--4\%), so the shared-grid integration error is not hidden inside it.
\item \textbf{Kit}: \texttt{kit/} (Sec.~\ref{sec:validation}); every tool has a self-test (\texttt{--selftest}); \texttt{sh kit/run\_selftests.sh} runs them and the end-to-end test, 12 in all, and all pass. \texttt{vk\_face\_from\_pw.py} reads VASP (scalar and spinor) and Quantum ESPRESSO wavefunctions or a generic array file, normalizes planar components as $\Omega^{-1/2}$ and forms the Bardeen coupling in eV; on real data its reader is checked against the code's charge density (Sec.~\ref{sec:dft}). Earlier versions of four of its tools stay with the scripts that used them: \texttt{models/\allowbreak kit/} (\texttt{tb8.py}, \texttt{checks10.py}) and \texttt{firstprinciples/\allowbreak kit/} (the first-principles analysis).
\item \textbf{First-principles workflow} (\texttt{firstprinciples/}): \texttt{twistdft.py} (inputs from the documented settings and the parent run, scientific gates, bilayer splittings), \texttt{twistpost.py} (reader check, faces and index, $\Gamma$ data, Kramers and run comparisons, spin--orbit reference), \texttt{twistbench.py} (run and benchmark table); every run in a fresh directory from an immutable job file with hashed inputs and scripts (the job scripts of the cluster are not included); \texttt{analyse\_dft10.py} and \texttt{analyse\_dft11.py} make the tables and the data of the letter's Fig.~3 from the fetched results.
\item \textbf{Figures}: \texttt{figures/\allowbreak make\_figures.py} makes the letter's Figs.~1--3 and every figure of this Supplement: those of Sec.~\ref{sec:dft} and Fig.~\ref{fig:umklapptail} itself, the others through \texttt{figures/\allowbreak supplement\_models.py} (Figs.~\ref{fig:matrix}, \ref{fig:matrixval}, \ref{fig:theoremtests}, \ref{fig:zero}, \ref{fig:phonon}, \ref{fig:transparent} and \ref{fig:outofsample}), \texttt{figures/\allowbreak spacing\_umklapp.py} (Figs.~\ref{fig:spacing} and \ref{fig:umklapp}) and \texttt{figures/\allowbreak predictions.py} (Figs.~\ref{fig:predict} and \ref{fig:thermo}); \texttt{figures/\allowbreak figstyle.py} draws the Supplement's figures in the style of the letter's at their printed size; \texttt{figures/\allowbreak make\_tables.py} writes the transport tables of Sec.~\ref{sec:dft}; all from saved results.
\end{itemize}
}